\documentclass[11pt]{article}
\usepackage{latexsym}
\usepackage{amssymb,amsmath}
\usepackage{amsthm}
\usepackage{mathtools}
\usepackage{graphicx}
\usepackage{color}
\usepackage{blkarray}
\usepackage{multirow}
\usepackage{hyperref}
\usepackage{float}
\usepackage{relsize}
\usepackage[nice]{nicefrac}

\usepackage[numbers,sort&compress]{natbib}

\hypersetup{
  colorlinks   = true, %Colours links instead of ugly boxes
  urlcolor     = blue, %Colour for external hyperlinks
  linkcolor    = blue, %Colour of internal links
  citecolor   = red %Colour of citations
}

\usepackage{breqn}

\usepackage{xcolor}
\usepackage{cancel}

\usepackage{hyperref}
\usepackage{empheq}

\usepackage[utf8]{inputenc}

\usepackage{setspace}

\usepackage{enumerate}

\usepackage[titletoc,toc]{appendix}

\usepackage[font=small,labelfont=bf]{caption}

\newcommand{\RR}{\mathbb{R}}

\newcommand{\ds}{\displaystyle}

\DeclareMathOperator{\argmax}{argmax}

\DeclareMathOperator{\argsup}{argsup}

\newcommand{\bbm}{\begin{bmatrix}}
\newcommand{\bpm}{\begin{pmatrix}}
\newcommand{\ebm}{\end{bmatrix}}
\newcommand{\epm}{\end{pmatrix}}

 \newcommand{\dsdel}[2]{\displaystyle\frac{\partial #1}{\partial #2}}

\newcommand{\ddt}[1]{\frac{d #1}{dt}}

\newcommand{\dsddx}[2]{\displaystyle\frac{d #1}{d #2}}
\newcommand{\dsddt}[1]{\displaystyle\frac{d #1}{dt}}

\renewcommand{\abstractname}{Abstract}

\numberwithin{equation}{section}
\numberwithin{figure}{section}

\usepackage[T1]{fontenc}
\usepackage[utf8]{inputenc}
\usepackage{authblk}
\usepackage[utf8]{inputenc}

\title{A PDE Model for Protocell Evolution and the Origin of Chromosomes via Multilevel Selection}
\author[1,2]{Daniel B. Cooney}
\author[3]{Fernando W. Rossine}
\author[4]{Dylan H. Morris}
\author[3]{Simon A. Levin}

\affil[1]{\small{Department of Mathematics, University of Pennsylvania, Philadelphia, PA, USA}}
\affil[2]{\small{Center for Mathematical Biology, University of Pennsylvania, Philadelphia, PA, USA}}
\affil[3]{Department of Ecology and Evolutionary Biology, Princeton University, Princeton, NJ, USA}
\affil[4]{Department of Ecology and Evolutionary Biology, University of California, Los Angeles, CA, USA}

\begin{document}

\renewcommand{\baselinestretch}{1.1}
\newtheorem{definition}{Definition}[section]
\newtheorem{theorem}{Theorem}[section]
\newtheorem{lemma}[theorem]{Lemma}
\newtheorem{corollary}[theorem]{Corollary}
\newtheorem{claim}[theorem]{Claim}
\newtheorem{fact}[theorem]{Fact}
\newtheorem{proposition}{Proposition}[section]
\newtheorem{remark}{Remark}[section]
\newtheorem{example}{Example}[section]

% to get nice proofs ...
\newcommand{\qedsymb}{\mbox{ }~\hfill~{\rule{2mm}{2mm}}}
\newenvironment{proof1}{\begin{trivlist}
\item[\hspace{\labelsep}{\bf\noindent Proof: }]
}{\qedsymb\end{trivlist}}

\newenvironment{hackyproof}{\begin{trivlist}
\item[\hspace{\labelsep}{\bf\noindent Proof: }]
}{%\qedsymb
\end{trivlist}}

\maketitle

\begin{abstract}
    The evolution of complex cellular life involved two major transitions: the encapsulation of self-replicating genetic entities into cellular units and the aggregation of individual genes into a collectively replicating genome. In this paper, we formulate a minimal model of the evolution of proto-chromosomes within protocells. We model a simple protocell composed of two types of genes: a ``fast gene'' with an advantage for gene-level self-replication and a ``slow gene'' that replicates more slowly at the gene level, but which confers an advantage for protocell-level reproduction. Protocell-level replication capacity depends on cellular composition of fast and slow genes. We use a partial differential equation to describe how the composition of genes within protocells evolves over time under within-cell and between-cell competition. We find that the gene-level advantage of fast replicators casts a long shadow on the multilevel dynamics of protocell evolution: no level of between-protocell competition can produce coexistence of the fast and slow replicators when the two genes are equally needed for protocell-level reproduction. By introducing a ``dimer replicator'' consisting of a linked pair of the slow and fast genes, we show analytically that coexistence between the two genes can be promoted in pairwise multilevel competition between fast and dimer replicators, and provide numerical evidence for coexistence in trimorphic competition between fast, slow, and dimer replicators. Our results suggest that dimerization, or the formation of a simple chromosome-like dimer replicator, can help to overcome the shadow of lower-level selection and work in concert with deterministic multilevel selection to allow for the coexistence of two genes that are complementary at the protocell-level but compete at the level of individual gene-level replication.
\end{abstract}

{\hypersetup{linkbordercolor=black, linkcolor = black}
\begin{spacing}{0.01}
\renewcommand{\baselinestretch}{0.1}\normalsize
\tableofcontents
\addtocontents{toc}{\protect\setcounter{tocdepth}{2}}
%\addtocontents{toc}{~\vspace{-3\baselineskip}}
\end{spacing}
%\clearpage
\singlespacing

\section{Introduction}

\subsection{Protocell Evolution and the Evolution of Chromosomes}

The development of a genome with sufficient size to encode for complex biological function was a major step toward the evolution of cellular life. This major evolution transition required overcoming a substantial hurdle known as Eigen's paradox: the error-rate of self-replicating genetic templates is too high to support large genomes without a self-correcting mechanism, but a relatively substantial genetic sequence is required to encode for an enzyme capable of performing such a task \cite{eigen1971selforganization}. How could individual genetic replicators ``cooperate'' (serving complementary roles to produce complex functions) while competing with each other for their own individual-level replication?

Eigen and Schuster \cite{eigen1977principle,eigen1979abstract,eigen1979realistic,eigen1980hypercycles} proposed that coexistence of complementary genes could be produced by a hypercycle: a collection of replicators with a cyclic interdependency for catalyzing their reproduction. Others, including Szathmary, proposed that cooperative coexistence could be maintained by the formation of protocells---collections of genes encased in a lipid membrane---which could localize the benefits conferred by a given replicator to nearby copies of a complementary replicator \cite{szathmary1987group,bresch1980hypercycles,szostak2001synthesizing,chen2005rna}. These two hypotheses are themselves potentially complementary; Michod \cite{michod1983population} showed that although encapsulation is not strictly needed to sutain a hypercycle, spatial structuring plays an important role in maintaining the coexistence of complementary genes. Further work  by Hogeweg, Takeuchi, and coauthors has explored a variety of comparisons between the mechanisms of spatial self-organization and pre-existing protocell structure in allowing coexistence of complementary genetic templates \cite{hogeweg2003multilevel,takeuchi2009multilevel,takeuchi2012evolutionary}.

The evolution of protocells is a particularly interesting problem in the framework of major evolution transitions \cite{szathmary1995major,szathmary2015toward}; natural selection acts both through competition among genes within protocells for replication and also through replication competition among the protocells themselves. Bresch and coauthors \cite{bresch1980hypercycles,niesert1981origin} introduced a ``package model'' of protocells to explore the role that a protocell-structure could play in maintaining coexistence of a hypercycle of RNA replicators. Packages divide upon reaching a certain density of replicators; this produces a new level of selection (package-level selection). Szathmary and Demeter \cite{szathmary1987group,grey1995re} introduced a stochastic corrector model of gene-gene cooperation. They found that coexistence between two complementary genes could be maintained via compartmentalization into self-replicating protocells even if one of the two genes had an advantage for individual-level replication. Recent experimental work using serial transfer methods has shown that compartmentalization into cell-sized water-in-oil droplets can promote the coexistence of host and parasitic RNA replicators in cases when well-mixed competition produces exclusion of host replicators by parasites \cite{bansho2016host}. Further theoretical work on the evolution of protocells has considered nested models for template coexistence both in finite populations \cite{alves2000group,fontanari2006coexistence,silvestre2005template} and in the PDE limit of large populations of protocells each containing many genetic templates \cite{fontanari2013solvable,fontanari2014effect,fontanari2014nonlinear}, as well as a model showing how information-sharing between protocells can facilitate the acquisition of sufficient genetic information necessary for the evolution of complex biological function \cite{sinai2018primordial}.

Another major evolutionary transition was the emergence of chromosomes: individual self-replicating genes were joined to form multi-gene polymers that replicate together. While the evolution of linkage of complementary genes has been explored in the context of free-living replicators \cite{levin2020social,levin2017evolution}, the origin of chromosomes is typically explored in the context of the transition between protocells with unlinked genes and decentralized gene-level reproduction to cells whose genes are combined into chromosomes, and an even later evolution of mitotic mechanisms to synchronize gene-level and cell-level replication \cite{boyden1953comparative, gabriel1960primitive}. Gabriel \cite{gabriel1960primitive} proposed a simple model for the evolution of chromosomes: supposing that a protocell requires a given number of necessary genes, one can calculate the probability that at least one necessary gene is lost during cell division. Two ways to reduce this probability of stochastic loss are to increase the copy number of genes (polyploidy) and to link genes into chromosomes. Gabriel concluded that chromosomal gene linkage was the more economical solution to protecting necessary genetic information from stochastic loss \cite{gabriel1960primitive}.

Stochastic gene loss was then considered in a dynamical multilevel setting by Maynard Smith and Szathmary \cite{smith1993origin}. They applied a ``stochastic corrector'' model to study the evolution of protocells that feature complementary independent template genes and proto-chromosomes consisting of a copy of each gene. Assuming that the proto-chromosomes replicate at half the speed of the independent templates but confer a collective advantage to templates in their protocell, Maynard Smith and Szathmary showed that proto-chromosomes could emerge under multilevel selection by preventing the stochastic loss of complementary genes. Szathmary and Maynard Smith also proposed a molecular mechanism for the origin of genetic linkage \cite{szathmary1993evolution}, and further work on the evolution of genomes has explored the evolution of division-of-labor of early replicators into roles corresponding to enzymes and genomes \cite{boza2014evolution,takeuchi2017origin,takeuchi2019origin} and has shown that the existence of chromosomes can facilitate the evolution of specialist enzymes facilitating more complex cell function \cite{szilagyi2012early,szilagyi2020evolution}.  

In both the package model and the stochastic corrector model, protocell-level selection is most capable of achieving coexistence of unlinked complementary genes for intermediate numbers of replicators per protocell \cite{niesert1981origin,szathmary1987group,grey1995re}. For small copy numbers, the random segregation of templates during protocell fission can result in stochastic loss of one of the necessary genes. For large copy numbers, the individual advantage of parasitic replicators dominates the dynamics and results in competitive exclusion of templates that are less viable under gene-level selection. Maynard Smith and Szathmary showed that dimerization---the formation of proto-chromosomes---can evolve because it allows protocells to overcome template stochastic loss when gene copy numbers are low. 
In this paper, we consider the problem of maintaining complementary genes in protocells when gene copy number is high and within-cell dynamics are quasi-deterministic. We show that protecell-level selection can still promote the formation of proto-chromosomes.  In the large copy number case, proto-chromosomes evolve because they permit protocells to overcome the barrier to gene coexistence caused by within-cell replication competition among genes, competition that can be parasitic from the perspective of the cell.

\subsection{Past Work on Modeling Multilevel Selection}

Recently, the tension between evolutionary forces operating at multiple levels of selection has been explored through a variety of stochastic and deterministic models for group-structured populations using the framework of nested birth-death processes. Luo introduced a finite population ball-and-urn model for multilevel selection featuring two types of individuals: defectors with a faster rate of individual-level replication rate and cooperators who confer a benefit to their group's collective replication rate \cite{luo2014unifying,van2014simple}. Luo and Mattingly considered a PDE description of this model in the limit of infinitely many groups and infinite group size, studying the long-time behavior of the PDE, and characterizing whether the population would converge to full-defection or a steady state density supporting cooperation depending on the relative strength of within-group and between-group selection \cite{luo2017scaling}. Extensions of the Luo-Mattingly model have been used to study fixation probabilities in the finite population setting \cite{mcloone2018stochasticity}, quasi-stationary states in a diffusive PDE scaling limit \cite{velleret2019two,velleret2020individual}, and the evolutionary dynamics of host-parasite competition \cite{pokalyuk2019diversity,pokalyuk2019maintenance}. Related PDE models from Simon and coauthors have also incorporated additional mechanisms including group-level carrying capacities and fission/fussion dynamics to further explore multilevel selection and the evolution of cooperation \cite{simon2010dynamical,simon2012numerical,simon2013towards,simon2016group,henriques2019acculturation}.

The Luo-Mattingly model was also extended to incorporate within-group and between-group competition that depends on the payoffs received by playing two-strategy games within groups \cite{cooney2019replicator,cooney2019assortment,cooney2020analysis}. 
For games in which an intermediate level of cooperation maximized collective payoff, it was shown that the individual-level selection cast a long shadow on the dynamics of multilevel selection: the average payoff of the population was limited by the payoff of the all-cooperator group, and the optimal collective payoff could not be achieved even in the limit of infinitely strong between-group competition. This shadow of lower-level selection is particularly extreme when the strategies interpreted as cooperation and defection are actually equally desirable at the group level. In this case in which a fifty-fifty mix of cooperators and defectors was optimal at the group level, it was shown that any slight selective advantage for defectors resulted in defectors taking over the whole population regardless of the strength of between-group competition \cite{cooney2019replicator}. This work was further generalized to incorporate any continuously differentiable within-group and between-group replication functions, and it was shown that a similar shadow of lower-level selection occured in this broader class of models \cite{cooney2021long}.

Fontanari and Serva used a PDE model for multilevel selection were used by  to study the question of template coexistence in protocell evolution, adapting a degenerate parabolic equation introduced by Kimura \cite{kimura1984evolution,kimura1986diffusion,ogura1987stationary,ogura1987stationary2} to study the evolution of complementary genes in a protocell-structured population. Assuming that one gene has an individual-level selective advantage but that the two genes are perfect complements for protocell-level competition, Fontanari and Serva find parameter regimes in terms of the strength of selection at the two levels, genetic drift, and migration such that the two genes can coexist in the long term. In this paper, we consider a similar question of template coexistence in models that focus on the deterministic evolutionary forces of within-group and between-group competition. On a related note, Fontanari has also studied the evolution of chromosomes in protocells in a finite population Moran process model with protocell fissioning \cite{fontanari2012genetic}, and explored the fixation probability of chromosomes using approaches that have been applied to finite population models of the evolution of cooperation via multilevel selection \cite{traulsen2005stochastic,traulsen2006evolution,traulsen2008analytical,bottcher2016promotion}.

\subsection{Outline of the Paper}

In this paper, we model selective dynamics both for replication of genes within protocells and the replication of entire protocells. We will consider two types of genes: one with a faster fixed replication rate within cells (fast replicators) and another with a slower replication rate (slow replicators). We consider protocell-level reproduction rates that can either favor as many slow genes as possible (analogous to the Luo-Mattingly model \cite{luo2017scaling}) or that most favor protocells with a mix of fast and slow replicators (similar to the Fontanari-Serva model  \cite{fontanari2013solvable}). These gene-level and protocell-level birth rates can be used to formulate a PDE multilevel selection model of the form studied by Cooney and Mori \cite{cooney2021long}, and we can apply existing results to characterize the long-time genetic composition of the population of protocells. We find that the multilevel dynamics of this protocell model display a shadow of lower-level selection, and, in the case in which fast and slow replicators are equally necessary for protocell-level replication, any slight gene-level advantage for the fast replicators prevents any coexistence of the two replicators under our model of multilevel selection. This inspires us to introduce a third replicator type, a dimer consisting of a linked fast and slow gene, and we show that inclusion of slow-fast dimers can help protocells overcome the shadow of lower-level selection and promote the coexistence of the complementary fast and slow genes.

The remainder of the paper is organized in the following manner. In Section \ref{sec:protocell}, we present our baseline model for multilevel competition in a population of protocells featuring fast and slow replicators, review the relevant backround for analyzing PDE models of mutlilevel competition between two replicators, and characterize the longtime behavior of this baseline model. In Section \ref{sec:trimorphicformulation}, we extend the baseline protocell model to include dimer replicators and describe our PDE model for protocell evolution featuring all three types of replicators and the corresponding gene-level and protocell-level replication rates in our extended model. In Section \ref{sec:simplexedgedynamics}, we consider restrictions of the full trimorphic model to understand pairwise competition between fast and dimer replicators and competition between slow and dimer replicators, using our exisiting PDE framework to show analytically how the presence of dimers can help to promote coexistence of the fast and slow gene. In Section \ref{sec:trimorphicnumerics}, we provide numerical results for the full competition between slow, fast, and dimer replicators, showing how the introduction of dimers can allow for the coexistence of the fast and slow genes for cases in which no coexitence was possible under the pairwise competition of fast and slow replicators. Section \ref{sec:discussion} provides a discussion of our results and directions for future work for models of dimerization via multilevel selection. The derivation of the PDE models from stochastic two-level birth-death processes are provided in Section \ref{sec:derivation},  and the finite volume scheme to study numerical solutions of our multilevel PDE models are provided in Section \ref{sec:finitevolume}.

\section{Baseline Model for Protocell Evolution: Fast and Slow Replicators}  \label{sec:protocell}

In this section, we introduce a baseline PDE model of multilevel selection in a population of protocells featuring fast and slow replicators. In Section \ref{sec:PDEprotocellformulation}, we provide the gene-level replication rates of the two replicators and formulate how the composition of replicators impacts the rate of protocell-level replication. In that subsection, we also present our PDE for multilevel protocell composition, but postpone the derivation of this equation from an individual-based process to Section \ref{sec:derivationbaseline}. In Section \ref{sec:existingresults}, we recall existing results for a class of PDE models for multilevel selection that includes our baseline protocell model as a special case. Finally, in Section \ref{sec:protocellongtime}, we apply the results from Section \ref{sec:existingresults} to study the long-time behavior of our protocell model and show how the shadow of lower-level selection can impede the possibility of achieving coexistence of complementary genes via mulltilevel selection.

\subsection{Formulation for PDE Model of Protocell Evolution}
\label{sec:PDEprotocellformulation}
As a baseline model for protocell evolution, we consider protocells that contain two types of individuals, fast and slow replicators. Evolutionary competition acts at two levels of selection: individual replicators compete within protocells for gene-level replication and protocells also compete for collective replication based upon their composition of replicators. We can describe the evolutionary dynamics at each level of selection using a nested birth-death process, in which both gene-level / within-protocell and between-protocell replication events are modeled by continuous-time Moran processes. 

We assume that within-protocell competition is frequency-independent, with slow replicators producing copies of themselves at rate $1 + w_I b_S$ and fast replicators producing copies of themselves at rate $1 + w_I b_F$, where $b_F > b_S$ and $w_I$ describes a measure of the strength of gene-level selection relative to the background rate $1$ of neutral birth events. We assume that every birth is paired with a death of a randomly chosen replicator, so the number of replicators is unchanged due to gene-level competition. As a special case inspired by the Luo-Mattingly and Fontanari-Serva models, we will consider a family of selective birth rates given by $b_S = 1$ and $b_F = 1 + s$, where where $s > 0$ is the individual-level advantage for fast-replicators. %
For between-protocell competition, we assume that a protocell composed of a fraction $x$ of slow replicators and a fraction $1-x$ fast replicators will produce a copy of itself at a rate $\Lambda \left( 1 + w_G G_{FS}(x)\right)$, where $\Lambda$ is a relative rate of gene-level and protocell-level birth events, $w_G$ describes the strength of selection for protocell-level replication events, and $G_{FS}(x)$ encodes the dependence of protocell-level replication rate on cellular gene composition. We further assume that offspring protocells replace a randomly chosen protocell in the population, and therefore the number of protocells in the population is constant over time. This protocell-level reproduction function $G_{FS}(x)$, which we will sometimes call the protocell-level fitness, is given by the following quadratic function 
\begin{equation} \label{eq:Gofxfdeta}
G_{FS}(x) =  x \left( 1 - \eta x \right),
\end{equation}
where $\eta \in [0,1]$ is a parameter describing the complementarity of fast and slow genes for protocell-level reproduction. The purpose of introducing this family of protocell-level reproduction rates is to consider the range of possible protocell-level selective effects ranging from promoting as many slow replicators as possible (when $\eta = 0$ and $G_{FS}(x) = x$) to most favoring a fifty-fifty mix of fast and slow replicators (when $\eta = 1$ and $G_{FS}(x) = x (1-x)$). In particular, the formula for $G(x)$ interpolates between the protocell-level replication functions studied by Luo and Mattingly for $\eta = 0$ \cite{luo2014unifying,luo2017scaling} %
and by Fontanari and Serva for $\eta = 1$ \cite{fontanari2014effect,fontanari2014nonlinear,fontanari2013solvable}. 

For various values of the complementarity parameter $\eta$, we see that the composition maximizing the protocell-level reproduction rate $x^{*}_{FS} := \argmax_{x \in [0,1]} G(x)$ has the piecewise characterization
\begin{equation} \label{eq:xstarFS}
    x^*_{FS} = \left\{
    \begin{array}{cr}
    1 &: \eta < \frac{1}{2} \vspace{2mm} \\ 
    \ds\frac{1}{2 \eta} &: \eta \geq \frac{1}{2}
    \end{array}
    \right. .
\end{equation}
This tells us that protocell-level replication is maximized by the all-slow composition when $\eta \leq \frac{1}{2}$, while a mix of fast and slow replicators is most favored for $\eta > \frac{1}{2}$. We illustrate the different possible collective reproduction rates $G_{FS}(x)$ in Figure \ref{fig:Gfsfunction}; slow and fast replicators display a greater degree of complementarity as the parameter $\eta$ increases from $0$ to $1$. Note that the reproduction rate of the all-slow protocell $G_{FS}(1) = 1 - \eta$ is a decreasing function of $\eta$ and that all-slow and all-fast protocells both achieve the value $G_{FS}(0) = G_{FS} = 0$ when $\eta = 1$ (where the genes are perfect complements, and protocell-level fitness is maximized by a fifty-fifty mix of fast and slow genes).

\begin{figure}[ht]
    \centering
    \includegraphics[width = 0.6\textwidth]{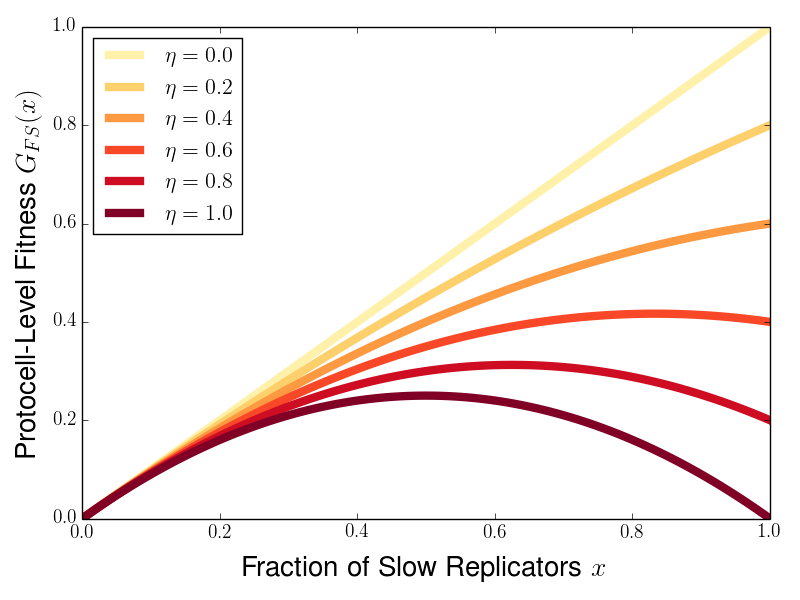}
    \caption{Protocell-level reproduction rates $G_{FS}(x)$ for various levels of the complementarity parameter $\eta$. When $\eta = 0$, protocell-level fitness is a linear increasing function of the fraction of slow replicators, while protocell-level fitness is maximized by a mix of slow and fast replicators for $\eta > \frac{1}{2}$. When $\eta = 1$, slow and fast replicators are perfect complements at the group level, with the full-fast $x=0$ and full-slow $x=1$ compositions minimizing the collective reproduction rate and the fifty-fifty mix $x=\frac{1}{2}$ maximizing collective reproduction.}
    \label{fig:Gfsfunction}
\end{figure}
In the limit of infinitely many protocells, each containing infinitely many genes, we describe the composition of the protocell-structured population by $f(t,x)$, the density of protocells having fraction $x$ slow replicators at time $t$. To take this limit, we use an approach introduced by Luo and coauthors \cite{luo2014unifying,van2014simple,luo2017scaling} and later applied to evolutionary games \cite{cooney2019replicator}. This yields following PDE for $f(t,x)$ given the two-level birth-death process described above:
\begin{equation} \label{eq:protocellPDEG}
    \dsdel{f(t,x)}{t} = \underbrace{\dsdel{}{x} \left[s x (1-x) f(t,x) \right]}_{\textnormal{Gene-Level Competition}} + 
\lambda \underbrace{ f(t,x) \left[G_{FS}(x) - \int_0^1 G_{FS}(y) f(t,y) dy \right]}_{\textnormal{Between-Protocell Competition}},
\end{equation}
where $\lambda := \frac{\Lambda w_G}{w_I}$ describes the relative strength of gene-level and protocell-level competition. The advection term $\left[ s x (1-x) f(t,x) \right]_x$ describes the effect of protocell-level competition favoring fast replicators, while the nonlocal term $\lambda f(t,x) \left[ G_{FS}(x) - \int_0^1 G_{FS}(y) f(t,y) dy\right]$ describes the effect of between-protocell competition favoring protocells with high collective reproduction rate $G_{FS}(x)$.
Using the expression for the protocell-level reproduction rate from Equation \eqref{eq:Gofxfdeta}, we can also write out our PDE model for multilevel protocell competition more explicitly as 
\begin{equation} \label{eq:protocell} \dsdel{f(t,x)}{t} =  \dsdel{}{x} \left[s x(1-x) f(t,x) \right] + \lambda f(t,x) \left[x - \eta x^2 - \langle G(\cdot) \rangle_{f(t,x)} \right],\end{equation}
where $\langle G(\cdot) \rangle_{f(t,x)} := \int_0^1 G(y) f(t,y) dy = \int_0^1 y \left( 2 - \eta y\right) f(t,y) dy$ denotes the average protocell-level replication rate across the whole population at time $t$.  

This model for protocell evolution described by Equation \eqref{eq:protocell} is a special case of a class of PDE models for multilevel selection studied by Cooney and Mori \cite{cooney2021long}. In Section \ref{sec:existingresults}, we summarize the main properties of solutions and long-time behaviors for this broader class of PDE models. In Section \ref{sec:protocell}, we present the application of these general results to characterize the long-time behavior of solutions to Equation \eqref{eq:protocellPDEG} for the protocell model and to understand implications for the coexistence of slow and fast replicators via multilevel selection.

\subsection{Existing Results for PDE Replicator Equations Describing Multilevel Selection} \label{sec:existingresults}

Using our terminology of fast and slow replicators, we now introduce a broader class of models for two-level selection with two types of individuals that allows for a variety of gene-level and protocell-level replication rates \cite{cooney2021long}. We consider gene-level birth rates of $1 + w_I \pi_F(x)$ and $1 + w_I \pi_S(x)$ for fast and slow replicators, and a protocell-level reproduction rate $\Lambda \left( 1 + w_G G(x)\right)$ for any continuously differentiable functions $\pi_F(x)$, $\pi_S(x)$, and $G(x)$ of the fraction of slow replicators $x$. After writing $\pi(x) = \pi_F(x) - \pi_S(x)$ to represent the within-group selection advantage of fast replicators over slow replicators in an $x$-slow group, the two-level dynamics in the large population limit for this class of models is described by the PDE
\begin{equation} \label{eq:generalPDEreplicator}
    \dsdel{f(t,x)}{t} = \dsdel{}{x} \left[x \left(1 -x\right) \pi(x) f(t,x) \right] + \lambda f(t,x) \left[G(x) - \int_0^1 G(y) f(t,y) dy \right],
\end{equation}
where $\lambda := \frac{\Lambda w_G}{w_I}$ characterizes the relative intensity of within-group and between-group competition. The characteristic curves of this model correspond to the replicator dynamics for within-group selection given by
\begin{equation} \label{eq:characteristicsgeneral}
    \dsddt{x(t)} = - x \left(1 - x\right) \pi(x) = - x \left(1 - x\right) \left( \pi_F(x) - \pi_S(x) \right) \: \:, \: \: x(0) = x_0.
\end{equation}
Competition between protocells is governed by the second term on the right-hand side of Equation \eqref{eq:generalPDEreplicator}, and protocells with composition $x$ increase in frequency if their protocell-level replication rate $G(x)$ exceeds the average replication rate $\langle G(\cdot) \rangle_{f(t,x)} := \int_0^1 G(y) f(t,y) dy$ across the population of protocells. Combining the effects of gene-level and protocell-level competition, we can think of Equation \eqref{eq:generalPDEreplicator} as a replicator equation for the evolution of the protocell population under multilevel selection, encoding the nested deterministic competition taking place at each level. 

Generalizing both the protocell model from Equation \eqref{eq:protocell} and models for multilevel selection in evolutionary games, we can place relatively weak assumptions on the functions $G(x)$ and $\pi(x)$ to characterize a variety of models featuring a tug-of-war between what is advantageous for the individual (here a gene) and what is collectively beneficial for the group (here a protocell). For models based upon the Prisoners' Dilemma from evolutionary game theory where $x$ represents the fraction of cooperative individuals and $1-x$ the fraction of defecting (or cheating) individuals, the assumptions $\pi(x) > 0$ and $G(1) > G(0)$ encode the properties that defectors have an individual advantage over cooperators under individual-level competition while a group of cooperators has a collective advantage over a group of defectors under group-level competition. For the protocell model from Equation \eqref{eq:protocell} with complementarity parameter $\eta < 1$, we have that $\pi(x) \equiv s > 0$ and $G(1) - G(0) = 1 - \eta > 0$, so the protocell model satisfies the same assumptions on $\pi(x)$ and $G(x)$ that are used in the Prisoners' Dilemma scenario by Cooney and Mori \cite{cooney2021long}. When $\eta = 1$ in the protocell model (and, correspondingly, fast and slow replicators are perfect complements for protocell-level competition), we instead have that $G(1) = G(0) = 0$, and the full-slow protocell does not have a collective advantage over the full-fast protocell. 

To analyze the dynamics of our model of multilevel protocell competition, it can be helpful to consider a weak formulation of Equation \eqref{eq:protocell} to allow for the possibility of concentration of the population upon equilibria of the within-group dynamics. Multiplying both sides of Equation \eqref{eq:protocell} by a $C^1([0,1])$ test-function $v(x)$ and integrating with respect to $x$ from $0$ to $1$, we obtain after integrating the advection term by parts the following weak version of the multilevel protocell dynamics
\begin{equation} \label{eq:weakdensity}
    \dsddt{} \int_0^1 v(x) f(t,x) dx = - \int_0^1 v'(x) x (1-x) \pi(x) f(t,x) dx + \lambda \int_0^1 f(t,x) \left[ G(x) - \int_0^1 G(y) f(t,y) dy \right] dx. 
\end{equation}
We then say that a probability density $f(t,x)$ is a weak solution to Equation \eqref{eq:generalPDEreplicator} if Equation \eqref{eq:weakdensity} holds for every possible test function $v \in C^1([0,1])$. This formulation allows us to deal with initial densities $f(0,x)$ that are not sufficiently differentiable to satisfy Equation  \eqref{eq:generalPDEreplicator} in the strong sense.

In fact, we can further weaken our notion of a solution to the multilevel dynamics by considering a measure $\mu_t(dx)$ describing the distribution of the fraction of slow replicators within a population of protocells. Then we can consider a measure-valued formulation of Equation \eqref{eq:generalPDEreplicator} in which, for any $C^1([0,1])$ test-function $v(x)$, the measure $\mu_t(dx)$ satisfies
\begin{equation} \label{eq:generalPDmeasurevalued}
  \dsddt{} \int_0^1 v(x) \mu_t(dx) = - \int_0^1 v'(x) x(1-x) \pi(x) \mu_t(dx) + \lambda \int_0^1 \left[ G(x) - \int_0^1 G(y) \mu_t(dy) \right] \mu_t(dx)  
\end{equation}
This description is particularly convenient because it allows us show that the delta-functions $\delta(x)$ and $\delta(1-x)$ concentrated at the all-fast and all-slow equilibria are steady state solutions to Equation \eqref{eq:generalPDmeasurevalued} \cite{cooney2019replicator}. To study the dynamics of Equation \eqref{eq:generalPDmeasurevalued}, we must supply the equation with an initial probability measure $\mu_0(dx)$. 

The long-time behavior of the multilevel dynamics depends on a property of the tail of the initial measure $\mu_0(dx)$ near the full-slow equilibrium called the H{\"o}lder exponent and defined as follows.

\begin{definition} \label{def:holderexponent}
A measure $\mu_t(dx)$ has H{\"o}lder exponent $\theta_t \geq 0$ near $x = 1$ if
\begin{equation} \label{eq:holderexponent}
    \theta_t = \inf\left\{\Theta \geq 0  \: \: \bigg| \: \: \ds\lim_{x \to 0} \frac{\mu_t\left(\left[1-x,1\right] \right)}{x^{\Theta}} > 0  \right\}.
\end{equation}
If a measure $\mu_t(dx)$ has such a H{\"o}lder exponent near $x=1$, then it has an associated H{\"o}lder exponent $C_{\theta_t} \in \RR_{\geq 0} \cup \{\infty\}$ satisfying
\begin{equation} \label{eq:holderlimit}
    \ds\lim_{x \to 0} \frac{\mu_t\left[1-x,1\right]}{x^{\theta_t}} = C_{\theta_t}.
\end{equation}
\end{definition}

\begin{remark}
This definition of the Hölder exponent for the measure $\mu(dx)$ near $x=1$ is related to the notion of pointwise H{\"o}lder continuity that characterizes the local regularity of functions. To see this, we introduce the cumulative distribution function $F(x) = \int_0^x \mu(dy)$ associated with the measure $\mu(dx)$, and note that $\mu\left([1-x,1]\right) = F(1^+) - F((1-x)^-)$ (where our choice of left-hand and right-hand limits allows us to include the mass accumulated at the endpoints $1-x$ and $1$). From Equation \eqref{eq:holderexponent}, we can see that an equivalent characterization of the H{\"o}lder exponent near $x=1$ in terms of $F(x)$ is given by  
\begin{equation} \label{eq:pointwiseholderexponent}
    \theta = \inf\left\{\Theta \geq 0  \: \: \bigg| \: \: \ds\lim_{x \to 0} \frac{F(1^+) - F((1-x)^-)}{x^{\Theta}} > 0  \right\}.
    \end{equation}
Equation \eqref{eq:pointwiseholderexponent} tells us that the cumulative distribution $F(x)$ has pointwise H{\"o}lder exponent $\theta$ at $x=1$ \cite{gilbarg2015elliptic,ayache2004identification}, showing that this degree of regularity of $F(x)$ at $x=1$ corresponds to the rate at which the tail measure $\mu\left(1-x,1\right]$ vanishes near the all-slow composition. 
\end{remark}

As an example, we can use Definition \ref{def:holderexponent} to see that measures of the form $\mu^{\theta}(dx) = \theta (1-x)^{\theta-1} dx$ for $\theta > 0$ have H{\"o}lder exponent $\theta$ near $x=1$. Notably, the uniform measure $\mu(dx) = 1 dx$ is a member of this family with H{\"o}lder exponent $1$ near $x=1$. This family of measures suggests that one way to interpret the H{\"o}lder exponent $\theta$ near $x=1$ is an equivalence class of initial measures whose survival functions $S(1-x) = \int_{1-x}^1 \mu(dx)$ are asymptotically equivalent to the survival function of $\mu^{\theta}(dx)$. Biologically, we can think of the H{\"o}lder exponent $\theta$ of the initial measure as representing a kind of inverse size of the initial cohort of nearly all-slow groups, as smaller $\theta$ corresponds to a larger concentration of groups near $x=1$ (the all-slow group composition). It can also be shown that, given an initial measure $\mu_0(dx)$ with H{\"o}lder exponent $\theta$ near $x=1$, the measure-valued solution $\mu_t(dx)$ to Equation \eqref{eq:generalPDEreplicator} will also have H{\"o}lder exponent $\theta$ near $x=1$ for all finite times $t > 0$ \cite{cooney2019replicator,cooney2020analysis,cooney2021long}. The fact the H{\"o}lder exponent near the all-slow composition is preserved under the multilevel dynamics suggests that this quantity can be used to characterize the long-time behavior for an entire equivalence class of initial populations.

The H{\"o}lder exponent near $x=1$ can also be used to characterize the steady state densities for the multilevel dynamics of Equation \eqref{eq:generalPDEreplicator}. Using Definition \ref{def:holderexponent}, it can be shown that, up to multiplication by a constant, there exists a unique steady state density solving Equation \eqref{eq:generalPDEreplicator} with H{\"o}lder exponent $\theta$ near $x=1$ \cite{cooney2021long}. This density is given by
\begin{equation}\label{eq:flambdatheta}
f^{\lambda}_{\theta}(x)= x^{\mathlarger{ \left[\pi(0)^{-1} \left( \lambda \left[G(1) - G(0) \right] - \theta \pi(1) \right) - 1\right]}}(1-x)^{\mathlarger{\theta-1}}\frac{\pi(1)}{\pi(x)}\exp\left( - \lambda \int_x^1\frac{C(u)du}{\pi(u)}\right),
\end{equation}
where the term $-\lambda C(x)$ takes the form
\begin{dmath} \label{eq:lambdaCofx}
- \lambda C(x) = \lambda \left( \frac{G(x) - G(0)}{x} \right) + \left(\frac{\lambda \left[ G(1) - G(0) \right] - \theta \pi(1)}{\pi(0)} \right) \left( \frac{\pi(x) - \pi(0)}{x} \right)  
+ \lambda \left( \frac{G(x) - G(1) }{1 - x} \right) - \theta \left( \frac{\pi(x) - \pi(1)}{1 - x} \right).
\end{dmath}
We can also consider a normalized version of $f^{\lambda}_{\theta}(x)$, which is the unique probabilty density with H{\"o}lder exponent $\theta$ that is a steady state solution to Equation \eqref{eq:generalPDEreplicator}.
\begin{equation} \label{eq:plambdatheta}
    p^{\lambda}_{\theta}(x)=\frac{f^{\lambda}_{\theta}(x)}{\int_0^1 f^{\lambda}_{\theta}(x)dx}
\end{equation}
Under the assumption that $G(x)$ and $\pi(x)$ are continuously differentiable functions, the expression in Equation \eqref{eq:lambdaCofx} is bounded on $[0,1]$. Because we also consider H{\"o}lder exponent $\theta > 0$ near $x = 1$, we see from Equation \eqref{eq:flambdatheta} that $\lambda$ must exceed a critical level of the strength of between-group selection
\begin{equation} \label{eq:lambdastargeneral}
    \lambda^* := \frac{\theta \pi(1)}{G(1) - G(0)}
\end{equation}
to ensure integrability of the steady state density $f^{\lambda}_{\theta}(x)$. Using the original notation in terms of the fast replicator and slow replicator birth rates $\pi_F(x)$ and $\pi_S(x)$, this threshold condition can be rewritten in the following form
\begin{equation} \label{eq:lambdastarexplained}
    \lambda^* = \frac{\left(\pi_F(1) - \pi_S(1)\right) \theta}{G(1) - G(0)}
\end{equation}

This threshold determines whether the long-time behavior of the multilevel dynamics starting from an initial population with H{\"o}lder exponent $\theta$ near $x=1$, separating a regime in which fast replicators take over the whole population from a regime in which fast and slow replicators can coexist at steady state. In Theorem \ref{thm:PDconvergencetosteady} (originally \cite[Theorem 1.5]{cooney2021long}), we show that if $\lambda > \lambda^*$, then the population will converge to a steady state density featuring groups with all possible fractions of fast and slow replicators. In Theorem \ref{thm:PDconvergencetodelta}, we present a modified version of \cite[Theorem 1.11]{cooney2021long}, showing that if $\lambda < \lambda^*$, then fast replicators will take over the population, with the population concentrating upon a delta-function supported at the all-fast equilibrium. We also include in Theorem \ref{thm:PDconvergencetodelta} the result of \cite[Proposition 5.2]{cooney2021long}, which proves convergence to the delta-function at the all-fast equilibrium when $\lambda = \lambda^*$ under the additional assumptions that $G(0)$ is unique minimum of $G(x)$ on $[0,1]$ and that the initial measure $\mu_0(dx)$ has a positive, finite H{\"o}lder constant $C_{\theta}$ near $x=1$. Notably, this assumption on $G(x)$ will hold for all of models we consider in this paper, so Theorems \ref{thm:PDconvergencetosteady} and \ref{thm:PDconvergencetodelta} cover the long-time behavior for any relative protocell-level selection strength $\lambda$ for a given initial measure $\mu_0(dx)$ with H{\"o}lder exponent $\theta > 0$ near $x=1$ with associated finite, positive H{\"o}lder constant $C_{\theta}$. 

For both Theorem \ref{thm:PDconvergencetosteady} and Theorem \ref{thm:PDconvergencetodelta}, we consider convergence in the sense of weak convergence of probability measures. Specifically, for a family of probability measure $\{\mu_t(dx) \}_{t \geq 0}$ and a limit measure $\mu_{\infty}(dx)$, we say that $\mu_t(dx) \rightharpoonup \mu_{\infty}(dx)$ ($\mu_t(dx)$ converges weakly to $\mu_{\infty}$) as $t \to \infty$ if, for every continuous test-function $v(x)$, $\int_0^1 v(x) \mu_t(dx) \to \int_0^1 v(x) \mu_{\infty}(dx)$ as $t \to \infty$. 

\begin{theorem}[\bfseries Convergence to Steady State Density Supporting Coexistence of Both Types {\cite[Theorem 1.5]{cooney2021long}}] \label{thm:PDconvergencetosteady} Suppose that $G(x)$ and $\pi(x)$ satisfy the assumptions of the multilevel Prisoners' Dilemma scenario:  $G(x),\pi(x) \in C^1\left([0,1]\right)$, $G(1) > G(0)$, and $\pi(x) > 0$ for $x \in [0,1]$. Consider an initial measure $\mu_0(dx)$ having a H{\"o}lder exponent $\theta > 0$ near $x=1$ with corresponding positive, finite H{\"o}lder constant  $C_{\theta}$.
If $\lambda \left[G(1) - G(0) \right] > \theta \pi(1)$, then the solution $\mu_t(dx)$ to Equation \eqref{eq:generalPDmeasurevalued} converges weakly to the probability measure defined by the density function $p^{\lambda}_\theta(x)$ defined in Equation \eqref{eq:plambdatheta}:
\begin{equation}\label{vmubarlim}
\lim_{t\to \infty} \int_0^1 v(x) \mu_t(dx)=\int_0^1 v(x)p^{\lambda}_{\theta}(x)dx
\end{equation}
where $v(x)$ is an arbitrary continuous function on $[0,1]$.
\end{theorem}

\begin{theorem}[\bfseries Convergence to Delta-Function at All-Fast Composition {\cite[Theorem 1.11 and Proposition 5.2]{cooney2021long}}]
\label{thm:PDconvergencetodelta}
\leavevmode
\begin{itemize}
\item Suppose that $G(x), \pi(x)$ satisfy the assumptions of Theorem \ref{thm:PDconvergencetosteady} and that the initial distribution $\mu_0(dx)$ has H{\"o}lder exponent $\theta$ near $x=1$. If $\lambda \left[G(1) - G(0) \right] < \theta \pi(1)$, then the solution $\mu_t(dx)$ to Equation \eqref{eq:generalPDmeasurevalued} has the long-time behavior $\mu_t(dx) \rightharpoonup \delta(x)$ as $t \to \infty$.
\item Suppose further that $G(x) > G(0)$ for $x \in (0,1]$ and that the initial measure $\mu_0(dx)$ has positive H{\"o}lder constant $C_{\theta} < \infty$ near $x=1$. If $\lambda \left[G(1) - G(0) \right] = \theta \pi(1)$,  then $\mu_t(dx) \rightharpoonup \delta(x)$ as $t \to \infty$.
\end{itemize}
\end{theorem}

Combining the results of Theorems \ref{thm:PDconvergencetosteady} and \ref{thm:PDconvergencetodelta}, we can now describe the long-time limit of the average protocell-level fitness of a steady state for an initial measure with given H{\"o}lder exponent $\theta$. It was shown by Cooney and Mori that the long-time average payoff is given by the piecewise characterization 
\begin{equation} \label{eq:Gaveragelambda}
    \ds\lim_{t \to \infty} \int_0^1 G(x) \mu_t(dx) = \left\{
     \begin{array}{cr}
       G(0) & : \lambda \leq \lambda^* \\
       G(1) - \ds\frac{\theta \pi(1)}{\lambda} & : \lambda > \lambda^* 
     \end{array}
   \right. ,
   \end{equation}
 and that the threshold between-protocell selection strength $\lambda^*$ from Equation \eqref{eq:lambdastargeneral} can be used to further see that
   \begin{equation} \label{eq:Gaveragethresh}
    \ds\lim_{t \to \infty} \int_0^1 G(x) \mu_t(dx) =  \left\{
     \begin{array}{cr}
       G(0) & : \lambda \leq \lambda^* \\
       \left(\ds\frac{\lambda^*}{\lambda} \right) G(0) + \left( 1 - \ds\frac{\lambda^*}{\lambda} \right) G(1) & : \lambda > \lambda^* 
       \end{array}
   \right. .
\end{equation} 
In particular, this tells us that 
\begin{equation}
   \ds\lim_{t \to \infty} \int_0^1 G(x) \mu_t(dx) \leq G(1),   
\end{equation}
and therefore the long-time population cannot outperform the protocell-level fitness for an all-slow protocell. This means that the individual-level advantage of fast replicators casts a long shadow on the multilevel protocell dynamics: no level of between-protocell competition can promote the best possible collective outcome when an intermediate mix of fast and slow genes is optimal for protocell-level replication. 

We can also consider the long-time behavior of the multilevel dynamics of Equation \eqref{eq:generalPDmeasurevalued} when gene-level dynamics selection for increasing fractions of slow replicators and protocell-level competition favors all-fast protocells over all-slow protocells. Mathematically, this corresponds to the assumptions that $\pi(x) > 0$ for $x \in [0,1]$ and that $G(0) > G(1)$. In \ref{prop:PDelconvergencetodelta}, we present a modified version of \cite[Proposition 1.14]{cooney2021long}, showing that fast replicators will take over the population of protocells provided that there are fast replicators in the initial population. The original version of \cite[Proposition 1.14]{cooney2021long} dealt with the assumptions $\pi(x) < 0$ and $G(0) < G(1)$ to model a cooperative trait that is favored over a cheating trait at both levels of selection, but the analysis of the present case carries over after applying the change-of-variable $x \mapsto 1-x$.

\begin{proposition}[\bfseries Convergence to Delta-Function at All-Fast Equilibrium When Both Levels of Selection Favor Fast Replicators {\cite[Proposition 1.14]{cooney2021long}}] \label{prop:PDelconvergencetodelta}
Suppose that $G(x), \pi(x) \in C^1([0,1])$, $G(0) > G(1)$, and $\pi(x) > 0$ for $x \in [0,1]$. If $\mu_0([0,1)) > 0$ and $\lambda > 0$, then the solution $\mu_t(dx)$ to Equation \eqref{eq:generalPDmeasurevalued} has long-time behavior characterized by $\mu_t(dx) \rightharpoonup \delta(x)$ as $t \to \infty$.
\end{proposition}
\begin{remark}
There exist initial measures $\mu_0(dx)$ that do not have well-defined H{\"o}lder exponents or H{\"o}lder constants near $x=1$, as the limits defining these quantities in Equation \eqref{eq:holderexponent} and Equation \eqref{eq:holderlimit} do not necessarily exist. It is possible to generalize the idea of the H{\"o}lder exponent as a measure of the tail of a measure near the all-slow composition by defining quantities called the infimum and supremum H{\"o}lder exponents  defined by respectively replacing the limits in  Equations \eqref{eq:holderexponent} and Equation \eqref{eq:holderlimit} with limits infimum and limits supremum \cite{cooney2021long}. These quantities exist for any measure $\mu(dx)$ on $[0,1]$, and there the supremum holder exponent $\underline{\theta}$ can be used to characterize an analogous threshold $\lambda^*_{\theta}$ by plugging $\underline{\theta}$ into Equation \eqref{eq:lambdastargeneral}. For an initial measure $\mu_0(dx)$ with supremum H{\"o}lder exponent $\underline{\theta} > 0$ near $x=1$, an analogue of Theorem \ref{thm:PDconvergencetodelta} holds showing that $\mu_t(dx) \rightharpoonup \delta(x)$ as $t \to \infty$ when $\lambda \leq \lambda^*(\underline{\theta})$ \cite[Theorem 1.11]{cooney2021long}. For the case of $\lambda > \lambda^*(\underline{\theta})$, coexistence is achieved  between slow and fast replicators in the sense of weak persistence of slow replicators, as $\limsup_{t \to \infty} \int_0^1 x \mu_t(dx) > 0$ \cite[Corollary 1.13]{cooney2021long}.

In this paper, we choose for simplicity to restrict attention to the class of initial measures with well-defined H{\"o}lder exponent near $x=1$. However, the generalization of these results for the broader class of initial distributions highlights the fact that our results on coexistence of complementary genes in our protocell models are not restricted to populations starting with this special class of initial data. In particular, the persistence result of \cite[Corollary 1.13]{cooney2021long} may be the more natural benchmark for exploring the question of whether dimerization can help to promote long-time coexistence of the complementary fast and slow genes.  
\end{remark}

\subsection{Long-Time Behavior of Protocell Model}
\label{sec:protocellongtime}

In this section, we will apply the results presented in Section \ref{sec:existingresults} to the special case of our baseline protocell model. Using the gene-level relative birth rate $\pi(x) = s$ and the protocell-level replication function $G_{FS} = x \left( 1 - \eta x \right)$, we characterize the long-time behavior of solutions to the multilevel protocell dynamics given by Equation \eqref{eq:protocellPDEG}. When protocell-level replication most favors protocells featuring a majority of slow replicators, we see that, for sufficiently strong between-protocell competition, the population can converge to a steady state featuring coexistence of fast and slow replicators. We illustrate such steady states in Figure \ref{fig:etadensities}, and see that these densities feature more fast replicators than is optimal for protocell-level fitness when between-protocell competition most favors a mix of fast and slow genes. We further characterize this discrepancy between the modal composition at steady state and the optimal protocell-level fitness in Proposition \ref{prop:FSmodal}, and illustrate this gap in the limit of infinite between-protocell selection strength in Figure \ref{fig:etaghost}. This shadow cast by gene-level advantage for fast replicators is found to be most extreme in the case in which fast and slow genes are perfect complements for protocell-level replication (when $\eta = 1$), as no level protocell-level competition can allow for coexistence of the two types.

We first look to express the family of steady state solutions to Equation \eqref{eq:protocellPDEG}. 
Using Equation \eqref{eq:plambdatheta} and the fact that $G_{FS}(1) = 1 - \eta$, $G_{FS}(0) = 0$, and $\pi_{FS}(x) \equiv s$, we see that the steady states are probability densities of the form
\begin{subequations} \label{eq:plambdathetaFSgeneral}
     \begin{align}
        p^{\lambda}_{\theta}(x) &= Z_p^{-1} \: \mathlarger{x}^{\mathlarger{\left[(\lambda / s) \left(1- \eta\right) - \theta - 1\right]}} \left(\mathlarger{1 - x}\right)^{\mathlarger{\theta - 1}} \exp\left( - \lambda \int_x^{1} \frac{C_{FS}(u)}{s} du \right) \\
         Z_p &= \int_0^1 \mathlarger{y}^{\mathlarger{\left[ (\lambda / s) \left(1 - \eta \right) - \theta - 1\right]}} \left(\mathlarger{1 - y}\right)^{\mathlarger{\theta - 1}} \exp\left( - \lambda \int_y^{1} \frac{C_{FS}(u)}{s} du \right) dy,
     \end{align}
\end{subequations}
where $C_{FS}(x)$ is given by applying the replication rates for the protocell model to Equation \eqref{eq:lambdaCofx}. 
Noting that $\pi_{FS}(x)$ is a constant and that the between-protocell replication rates satisfy 
\begin{subequations}
\begin{align}
G_{FS}(x) - G_{FS}(0) &= x \left( 1 - \eta x \right) \\
G_{FS}(x) - G_{FS}(1) &= (1-x) \left(\eta - 1 + \eta x \right) \label{eq:GFS1diff},
\end{align}
\end{subequations}
we can see from Equation \eqref{eq:lambdaCofx} that
\begin{equation} \label{eq:lambdaCofxFS}
- \lambda C_{FS}(x) = \lambda \left(\frac{ G_{FS}(x) - G_{FS}(0)}{x} \right) + \lambda \left(\frac{G_{FS}(x) - G_{FS}(1)}{1-x} \right) = \lambda \eta.
\end{equation}
This allows us to further compute that
\begin{equation} \label{eq:expCintFS}
    \exp\left( - \lambda \int_x^1 \frac{C_{FS}(u)}{\pi(u) ds}  \right) = \exp\left( \lambda \int_x^1 \left(\frac{\eta}{s}\right) du \right) = \exp\left( \frac{\lambda \eta}{s} \left(1 - x \right) \right).
\end{equation}
and, after introducing the constant $\tilde{Z}_p = Z_p \exp\left(-\frac{\lambda \eta}{s} \right)$, we can use Equation \eqref{eq:plambdathetaFS} and \eqref{eq:expCintFS} to write the family of steady states $p^{\lambda}_{\theta}(x)$ in the form

\begin{equation} \label{eq:plambdathetaFS}
    p^{\lambda}_{\theta}(x) = \tilde{Z}_p^{-1} \: x^{\mathlarger{\left[\left(\lambda / s \right) \left(1 - \eta\right) - \theta - 1\right]}} \left(1 - x \right)^{\mathlarger{\theta - 1}} \exp\left(-\frac{\lambda \eta x}{s} \right).
\end{equation}

We note that a density given by Equation \eqref{eq:plambdathetaFS} with $\theta > 0$ is integrable (and therefore actually a probability distribution) provided that the relative intensity of between-protocell competition $\lambda$ exceeds the following threshold value
\begin{equation} \label{eq:lambdastarFS}
    \lambda^*_{FS} = \frac{s \theta}{1 - \eta}.
\end{equation}
This tells us that $\lambda^*_{FS}$ is a decreasing function of the complementarity parameter $\eta$, and therefore it is easier to achieve coexistence of the fast and slow replicators via multilevel selection when between-protocell competition pushes for as many slow replicators as possible. Furthermore, the threshold has the property that $\lambda^*_{FS} \to \infty$ as $\eta \to 1$, so there is no integrable density of the form given by Equation \eqref{eq:plambdathetaFS} for the case in which $\eta = 1$ and the fast and slow replicators are perfect complements for between-protocell replication.

This threshold quantity also helps to determine the long-time behavior for solutions to Equation \eqref{eq:protocellPDEG}. In particular, for a given initial distribution with H{\"o}lder exponent $\theta$ near $x=1$, we can use Theorem \ref{thm:PDconvergencetosteady} to say that, when $\lambda > \lambda^*{FS}$, the population will converge to the steady state from Equation \eqref{eq:plambdathetaFS} for the corresponding values of $\lambda$ and $\theta$. When $\lambda < \lambda^*$, we can similarly apply Theorem \ref{thm:PDconvergencetodelta} to deduce that, when $\lambda < \lambda^*_{FS}$, the population will concentrate upon a delta-function $\delta(x)$ concentrated upon the all-fast protocell composition. We summarize these two results in Proposition \ref{prop:steadystateetamodel}. 

\begin{proposition}  \label{prop:steadystateetamodel} Suppose the population of protocells has initial measure $\mu_0(dx)$ with H{\"o}lder exponent of $\theta$ near $x=1$, and consider a measure-valued solution $\mu_t(dx)$ to \eqref{eq:protocellPDEG}. Then, in the limit as $t \to \infty$, the solution $\mu_t(dx)$ to the multilevel dynamics will have the following long-time behavior \begin{displaymath} \mu_t(dx)  \rightharpoonup \left\{
     \begin{array}{cr}
       \delta(x) & :  \lambda(1-\eta) \leq s \theta \\
       \tilde{Z}_p^{-1} \:  x^{\mathlarger{\left[(\lambda / s) (1 - \eta) - \theta - 1\right]}} (1-x)^{\mathlarger{\theta - 1}} \exp\left( -\frac{\ds\lambda \eta x}{s} \right) dx & : \lambda(1-\eta) >  s \theta
     \end{array} \right. .
     \end{displaymath} 
\end{proposition}

\begin{remark} When $\eta = 0$, Proposition \ref{prop:steadystateetamodel} recovers the result for the Luo-Mattingly model. As $\eta \to 1$, critical $\lambda^*_{FS}$ needed to achieve $\lambda > \frac{s \theta}{1 - \eta} \to \infty$, so the long-run steady-state is $\delta(x)$ for every finite relative strength of between-protocell competition $\lambda$ when the group reproductive fitness is given by $G_1(x) = x(1-x)$. As a result, it is not possible to achieve long-time coexistence of fast and slow replicators when the two genes are perfect complements at the between-protocell level and there is any gene-level advantage for fast replicators.  \end{remark}

In Figure \ref{fig:etadensities}, we illustrate two families of steady-state densities for various values of between-protocell selection strength $\lambda$ and for two choices of complementarity parameter $\eta$ for which protocell-level reproduction is maximized by an all-slow composition ($\eta = \frac{1}{3}$, left) or protocell-level reproduction is maximized by protocells composed of 75 percent slow replicators and 25 percent fast replicators ($\eta = \frac{2}{3}$, right). In the first case in which all-slow protocells are optimal, we see increasing the level of between-protocell competition $\lambda$ can allow for as many slow replicators as possible at steady state. For the case in which 75 percent slow replicators are collectively optimal, we see that the steady state densities concentrate around a fifty-fifty mix of fast and slow replicators, yielding a composition achieving a suboptimal protocell-level fitness even in the limit of strong between-protocell competition. 
\begin{figure}[ht]
    \centering
    \includegraphics[width=0.48\textwidth]{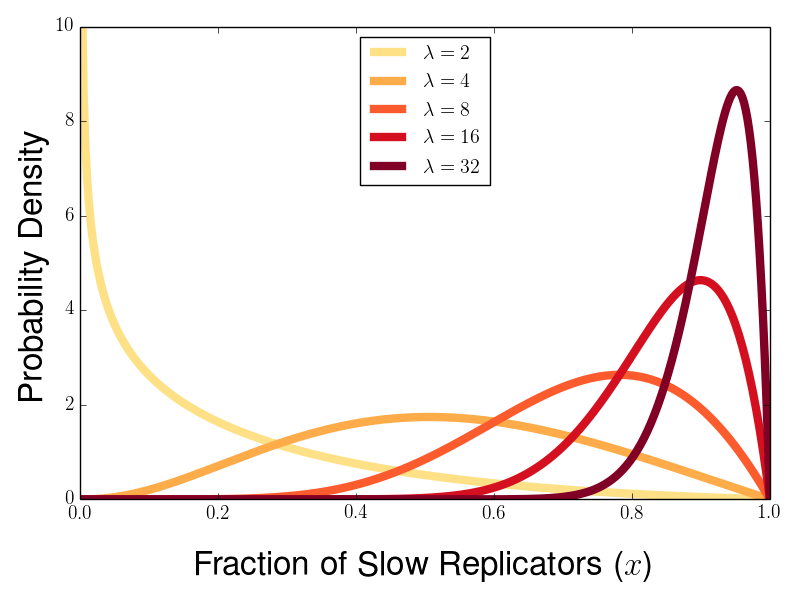}
    \includegraphics[width=0.48\textwidth]{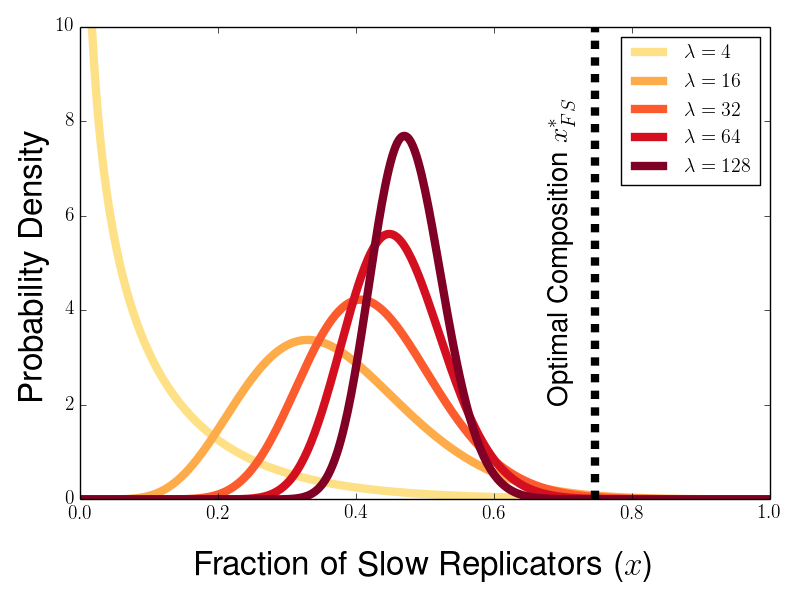}
    
    \caption{Steady state densities given by Equation \eqref{eq:plambdatheta} for various relative selection strengths $\lambda$ when $\eta = \frac{1}{3}$ and $G_{FS}(x)$ is maximized by protocells composed entirely of slow replicators (left), and when $\eta = \frac{2}{3}$ and $G_{FS}(x)$ is maximized by a protocell composition $x^*_{FS} = \frac{3}{4}$ featuring 75 percent slow replicators (right). The dotted vertical line in the right panel corresponds to the maximal payoff group $x^*_{FS}$.}
    \label{fig:etadensities}
\end{figure}

The discrepancy seen in Figure \ref{fig:etadensities}(right) between the compositions achieved at steady state and the composition providing the maximal rate of protocell-level replication can be further explored by studying the average protocell-level fitness at steady state. Using Equation \eqref{eq:Gaveragethresh} and the fact that $G_{FS}(0) = 0$, we see that the average protocell-level fitness for a steady state from Equation \eqref{eq:plambdathetaFS} is given by 
\begin{equation}
  \langle G \rangle_{p^{\lambda}_{\theta}(x)} = 
  \left\{
  \begin{array}{cr}
      0 &:  \lambda \leq \lambda^*_{FS} \\
      \left( \frac{\lambda^*_{FS}}{\lambda} \right) G_{FS}(1)  &: \lambda > \lambda^*_{FS} 
  \end{array}
  \right.
\end{equation}
and satisfies $\langle G \rangle_{p^{\lambda}_{\theta}(x)} \to G_{FS}(1)$ as $\lambda \to \infty$. When $\eta > \frac{1}{2}$ and $G_{FS}(x)$ is maximized by an interior fraction of slow replicators, we can use Equation \eqref{eq:GFS1diff} to see that $G_{FS}(x) = G_{FS}(1)$ for both $x=1$ and for $x = \frac{1}{\eta} - 1 \in (0,1)$. For the case considered in Figure \ref{fig:etadensities}(right) with $\eta = \frac{2}{3}$, we see that $\frac{1}{\eta} - 1 = \frac{1}{2}$, and therefore, when between-protocell competition is strong, the steady states of Figure \ref{fig:etadensities}(right) appear to concentrate upon the other point $x$ at which $G_{FS}(x) = G_{FS}(1)$ .

\sloppy{We formalize this observation about concentration upon the composition $\min\left(1,\tfrac{1}{\eta} -1\right)$ for strong between-protocell competition by studying the modal composition for the steady state densities $p^{\lambda}_{\theta}(x)$. In Proposition \ref{prop:FSmodal}, we study the most abundant protocell composition at steady state $\hat{x}_{\lambda} := \argsup_{x \in [0,1]} p^{\lambda}_{\theta}(x)$, with particular emphasis placed on the limit as $\lambda \to \infty$. We see that, when the all-slow protocell is most favored under between-protocell reproduction ($\eta \leq \frac{1}{2}$), the modal composition at steady state approaches the all-slow protocell in the limit as $\lambda \to \infty$. In the alternate case in which a mix of fast and slow replicators is most favored under between-protocell competition, we find that the modal composition approaches $\frac{1}{\eta} - 1$ as $\lambda \to \infty$, and the population concentrates upon a composition featuring fewer slow replicators than optimal for protocell-level replication and with the same collective replication rate as the all-slow protocell (Figure \ref{fig:etadensities}).

In Proposition \ref{prop:FSmodal}, we assume that $\lambda (1 - \eta) > s ( \theta + 1)$ and that $\theta \geq 1$. It can be seen from Equation \eqref{eq:flambdatheta} that the steady state densities $f^{\lambda}_{\theta}(x)$ are bounded on $[0,1]$ under these assumptions. While the former assumption holds for any initial condition under sufficiently strong between-protocell competition, the latter assumption restricts the set of steady states under consideration to those that remain bounded up to the all-slow composition. However, restricting attention to steady states with $\theta \geq 1$ allows us to study modal outcomes that depend on the relative strength of gene-level and protocell-level competition, rather than reflecting the blowup of the initial distribution near the all-slow composition.}

\begin{proposition}[\bfseries Most Abundant Protocell Composition at Steady State Features Fewer Slow Replicators Than Optimal, Even in the Limit of Infinite Strength of Between-Protocell Composition] \label{prop:FSmodal}
Consider the steady state density $p^{\lambda}_{\theta}(x)$ and suppose that $\lambda (1 - \eta) > s \left(\theta + 1\right)$ and $\theta \geq 1$. Then, for $\eta \in (0,1]$, the most abundant composition at steady state $\hat{x}_{FS}^{\lambda} := \argmax_{x \in [0,1]} p^{\lambda}_{\theta}(x)$ is given by 
\begin{equation} \label{eq:modalFSlambda}
    \hat{x}^{\lambda}_{FS} = \frac{\lambda - 2s - \sqrt{\left(\lambda - 2s \right)^2 - 4 \lambda \eta \left[ \lambda \left(1-\eta\right) - s \left(\theta +  1 \right) \right]}}{2 \lambda \eta}.
\end{equation}
and, for $\eta = 0$, the most abundant composition is given by $\hat{x}^{\lambda}_{FS} = 1$. Furthermore, in the limit of infinite intensity of between-protocell composition, the modal composition $\hat{x}^{\infty}_{FS} := \lim_{\lambda \to \infty} \hat{x}_{FS}^{\lambda}$ satisfies 
\begin{equation} \label{eq:hatxFSinf}
    \hat{x}^{\infty}_{FS} = \left\{  
    \begin{array}{cr}
    1 &: \eta < \frac{1}{2} \vspace{2mm} \\
    \ds\frac{1}{\eta} - 1 &: \eta \geq \frac{1}{2}
    \end{array}
    \right. .
\end{equation}
Comparing this expression with the group composition $x^*_{FS}$ from Equation \eqref{eq:xstarFS} that maximizes between-protocell replication rate, we see that, when $\eta > \frac{1}{2}$ (and corresponding $1 > \frac{1}{2 \eta}$), 
\begin{equation} \label{eq:xhatvsxstarFS} \hat{x}^{\infty}_{FS} = \frac{1}{\eta} - 1 = \frac{1}{2 \eta} + \underbrace{\left( \frac{1}{2\eta} - 1 \right)}_{< 0} <  \frac{1}{2 \eta} = x^*_{FS}. \end{equation}
\end{proposition}

This means that, despite taking $\lambda \to \infty$, %
the signature of the selective advantage of fast replicators over slow replicators under within-protocell gene-level competition remains when intermediate levels of fast and slow replicators are favored $\eta > \tfrac{1}{2}$. In Figure \ref{fig:etaghost}, we illustrate this discrepancy between the most abundant group type at steady state as $\lambda \to \infty$ and the cell composition with the fastest replication rate for between-protocell competition.

\begin{figure}[ht] 
  \centering
    \includegraphics[width=0.7\textwidth]{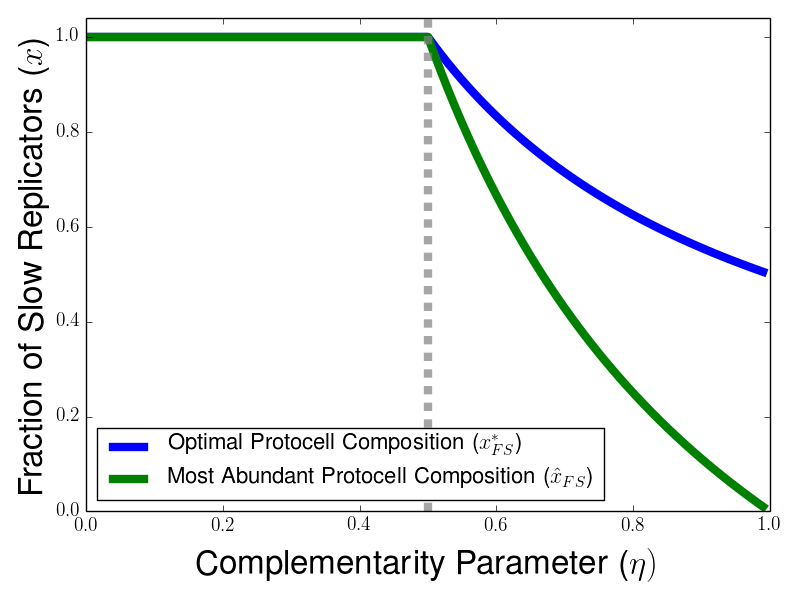}
      \caption[Protocell gene composition $x$ with maximum protocell-level reproduction rate $G(x)$ compared to most abundant protocell composition at steady state as $\lambda \to \infty$, plotted in terms of the parameter $\eta$.]{Protocell gene composition $x$ with maximum group reproduction rate $G(x)$ compared to peak abundance for the steady-state $f_{\theta}(x)$ as $\lambda \to \infty$, plotted in terms of the parameter $\eta$. For $\eta \leq \tfrac{1}{2}$ both peak protocell fitness and most abundant protocell type are full slow replicators protocells, while for $\eta > \frac{1}{2}$, protocell composition $\tfrac{1}{2 \eta}$ with maximal collective replication rate features more slow replicators than the most abundant protocell composition achieved at steady-state $\tfrac{1}{\eta} - 1$. }
      \label{fig:etaghost}
\end{figure}

\begin{proof}[Proof of Proposition \ref{prop:FSmodal}]
For any $\eta \in [0,1]$, we can differentiate the expression for steady state $p^{\lambda}_{\theta}(x)$ from Equation \eqref{eq:plambdatheta} to find that
\begin{equation}  \label{eq:FSdensityderiv} \dsddx{p^{\lambda}_{\theta}(x)}{x} = \tilde{Z}_p^{-1} 
a(x) x^{(\lambda/s)(1- \eta) - \theta - 2} (1-x)^{\theta - 2} \exp\left(-\frac{\lambda \eta x}{s} \right),  \end{equation}
where $a(x)$ is a quadratic function given by 
\begin{equation} \label{eq:axFS}
    a(x) = \frac{\lambda}{s} \left( 1 - \eta \right) - \theta - 1 + \left[2 - \frac{\lambda}{s} \right] x + \frac{\lambda \eta}{s} x^2.
\end{equation}
From the form of Equation \eqref{eq:FSdensityderiv}, we see that the critical points of $p^{\lambda}_{\theta}(x)$ are the endpoints $0$ and $1$, as well as any roots of $a(x)$ that are located in $(0,1)$. 

For the case of $\eta = 0$, we can see that Equation \eqref{eq:axFS} simplifies to
\begin{equation}
  a(x) = \frac{\lambda}{s} - \theta - 1 + \left[ 2 - \frac{\lambda}{s} \right] x,
\end{equation}
which is a linear function of $x$ taking on the values $a(0) = \frac{\lambda}{s} - \theta - 1$ and $a(1) = 1 - \theta$. Therefore we deduce $a(x)$ is positive on $[0,1]$ that under our assumptions that $\theta \geq 1$ and $\lambda > s (\theta + 1)$ when $\eta = 0$, and we can conclude that $p^{\lambda}_{\theta}(x)$ is maximized at the full-slow composition $\hat{x}^{\lambda}_{FS} = 1$ in this case.

For the case of $\eta \in (0,1]$, we use Equation \eqref{eq:axFS} to see that the roots of $a(x)$ are given by
\begin{equation}
x_{\pm}^{\lambda} = \frac{\lambda - 2 s  \pm \sqrt{(\lambda - 2 s)^2 - 4  \lambda \eta \left[\lambda(1- \eta) - s \left( \theta + 1 \right) \right]}}{2  \lambda \eta}, 
\end{equation}
and we can use our assumptions that $\lambda \left(1 - \eta \right) > s \left( \theta + 1 \right)$ and $\theta \geq 1$ to deduce that both of these roots are real and positive. Using these assumptions, we can also see from Equation \eqref{eq:axFS} that $a(x)$ takes on the following values at the endpoints $0$ and $1$
\begin{subequations}
\begin{align}
a(0) &= \frac{\lambda}{s} \left( 1 - \eta \right) - \theta - 1 > 0 \\
a(1) &= 1 - \theta \leq 0.
\end{align}
\end{subequations}
Because $a(x)$ is a convex, quadratic function, this tells us that $a(x)$ crosses $0$ from above at the unique point $ \hat{x}^{\lambda}_{-} \in [0,1]$, as this is the smaller of the two roots of $a(x)$. From Equation \eqref{eq:FSdensityderiv}, this tells us that $p^{\lambda}_{\theta}(x)$ is increasing on $(0,x^{\lambda}_{FS})$ and non-increasing on $[x^{\lambda}_{FS},1]$ (where the second interval can collapse to a single point if $x^{\lambda}_{FS} = 1$). This allows us to deduce that  $\hat{x}^{\lambda}_{FS} = \argmax_{x \in [0,1]} p^{\lambda}_{\theta}(x) = \hat{x}^{\lambda}_{-}$, and therefore we have shown that the maximizer of $p^{\lambda}_{\theta}(x)$ is given by Equation \eqref{eq:modalFSlambda} when our complementarity parameter satisfies $\eta \in (0,1]$.

In the limit of strong between-protocell competition as $\lambda \to \infty$, we see that
\begin{equation}\hat{x}^{\infty}_{FS} := \lim_{\lambda \to \infty}  \hat{x}^{\lambda}_{FS} = \frac{1}{2 \eta} - \sqrt{\frac{1}{4 \eta^2} - \frac{1}{\eta} + 1} = \frac{1}{2 \eta} - \sqrt{\left(\frac{1}{2 \eta} - 1 \right)^2}.
\end{equation}
We can further simplify the square root depending on the value of $\eta$ to see that the modal composition at steady state in the large $\lambda$ limit is
\begin{equation}
\hat{x}^{\infty}_{\eta} = \left\{ 
\begin{array}{cr}
1 &: \eta < \frac{1}{2} \\
\ds\frac{1}{\eta} - 1 &: \eta \geq \frac{1}{2}
\end{array}
\right..
\end{equation}
We can then compare this modal outcome $\hat{x}^{\infty}_{FS}$ to the maximum possible protocell-level fitness $x^*_{FS}$ from Equation \eqref{eq:xstarFS}, noting that $\hat{x}_{FS}^{\infty} = x^*_{FS}$ when $\eta \leq \frac{1}{2}$ (and all-slow protocells maximize collective fitness) and using Equation \eqref{eq:xhatvsxstarFS} to see that $\hat{x}^{\infty}_{FS} < x^*_{FS}$ when $\eta > \frac{1}{2}$ (and the collectively optimal protocell features a mix of fast and slow replicators). 
 
\end{proof}

\begin{remark}
From our our protocell model above and previous work on multilevel selection in evolutionary games \cite{cooney2019replicator,cooney2020analysis}, we have observed that coexistence of fast and slow genes becomes impossible to achieve at any relative selection strength $\lambda$ when the protocell-level reproduction function is given by $G(x) = x(1-x)$. This particular replication function has two notable properties: the all-slow composition and all-fast composition are equally capable of protocell-level replication with $G_{FS}(1) = G_{FS}(0) = 0$ and collective replication is maximized by protocells with a fifty-fifty mix of fast and slow genes. It is worth noting that the latter property of this reproduction function is not fundamental to the failure to promote coexistence, as we could consider protocell families of the form $G(x) = x^2 (1 - \eta x)$ or $G(x) = x (1 - \eta x)^2$ feature optimal mixes of $x = \frac{2}{3}$ and $x = \frac{1}{3}$ when $\eta = 1$, respectively. These other families of functions also satisfy the former property that $G(1) = G(0) = 0$ for $\eta = 1$, and we could similarly apply Theorem \ref{thm:PDconvergencetodelta} that no coexistence would be possible in these cases as well. From Theorem \ref{thm:PDconvergencetosteady} and the threshold condition of Equation \eqref{eq:lambdastargeneral}, we see that key importance played by the requirement for the collective reproduction rate of the all-slow equilibrium $G_{FS}(1)$ to exceed that of the all-fast equilibrium $G_{FS}(0)$ to allow coexistence of both genes at steady state. This threshold criterion and the role of the H{\"o}lder exponent near $x=1$ requiring an initial support of groups near full-cooperation bears resemblance to numerical findings in models for the origin of life \cite{szathmary1987group,markvoort2014computer}, in which the success of a small number of protocells featuring many slow/cooperative replicators can all for the long-time survival of genes necessary for collective reproduction. 
\end{remark}

\section{Formulation of Protocell Model with Slow-Fast Dimers: Trimorphic Dynamics} \label{sec:trimorphicformulation}

In Section \ref{sec:protocellongtime}, we found that no level of between-protocell competition could allow for coexistence of fast and slow replicators when the two genes were perfect complements under protocell-leve replication. In this section, we present one approach for overcoming this extreme shadow of lower-level selection, which consists of linking together the fast and slow gene into a dimer, or protochromosome, and allowing the dimer to compete along with fast and slow replicators within protocells. In this model, we augment our baseline model for multilevel selection in a population of protocells with the approach of Maynard Smith and Szathmary to describe how dimerization impacts gene-level competition and protocell-level replication rates.

In Section \ref{sec:dimerwithin}, we describe the gene-level (within-protocell) birth-death dynamics of our three kinds of replicators, and, in Section \ref{sec:dimergroupmultilevel}, we describe a cell reproduction function that depends on the fraction of slow and fast genes present in both the free replicators and dimers. In Section \ref{sec:dimergroupmultilevel}, we also present our  PDE model describing the coupled gene-level and protocell-level competition, illustrating how the model corresponds to a nested replicator equation now featuring three types of genetic replicators (fast, slow, and dimer). The formulation of our multilevel PDE for the fast-slow-dimer competition in this section will be based the intuition of a nested within-protocell, between-protocell replicator equation studied previously for our baseline fast-slow protocell model and in previous work on multilevel selection models with two types of individuals \cite{luo2014unifying,cooney2019replicator}. In Section \ref{sec:derivation}, we  provide a derivation for our trimorphic PDE dynamics from an underlying nested two-level Moran model describing the birth-death dynamics of the three types of replicators in a finite population.

\subsection{Gene-Level (Within-Protocell) Dynamics} \label{sec:dimerwithin}

In this section, we describe the within-cell dynamics of the fast-slow-dimer system, showing what the expected dynamics would be for the three-types in the absence of any between-cell competition. We denote the fraction of replicators that are fast replicators, slow replicators, and dimers by $x$, $y$, and $z$, respectively. We make the simplifying assumption that the total composition of replicators satisfies $x+y+z = 1$, so we can think of a dimer as having half of a fast template and half of a slow template. \textcolor{red}{}One could also treat the dimer as equivalent to two templates, so the net loss of a dimer via death would need to be compensated by the birth of two free replicators, but we're not as concerned with that detail because we are most interested in the individual-level disadvantages of dimers and the benefits that dimers can provide for a cell. }

We assume that the fast replicators, slow replicators, and fast-slow dimers produce a copy of itself with rates $1 + w_I b_F$, $1 + w_I b_S$, and $1 + w_I b_D$, respectively, and the copy replaces a randomly chosen replicator. We will assume that %
$b_F > b_S > b_D$ to capture the fact that fast replicators replicate faster than slow replicators and that dimers should be the slowest replicators because they require reproducing both a fast gene and a slow gene in a single birth event.  %
In a protocell with a large number of genes, we can use these rules for birth and death to characterize the within-protocell replicator dynamics for the fast, slow, and fast-slow dimer system are governed by the following system of ODEs 
\begin{subequations} \label{eq:withincelltrimorphic}
\begin{eqnarray}
\dsddt{x}&=& w_I x \left[ b_S - \left( b_S x +  b_F y +  b_D z \right)\right]\\
\dsddt{y}&=& w_I y \left[ b_F - \left( b_S x +  b_F y +  b_D z \right) \right]\\
\dsddt{z} &=& w_I z \left[ b_D - \left( b_S x +  b_F y +  b_D z \right)  \right]. \label{eq:withindimer}
\end{eqnarray}
\end{subequations} 
In subsequent analysis of the gene-level dynamics, we will rescale time to eliminate the factor of $w_I$ describing the strength of selection for gene-level replication events. Because the cell composition satisfies the conserved quantity $x+y+z = 1$, we can rewrite the proportion of fast-slow dimers as $z = 1 - x - y$, yielding the reduced two-dimensional within-protocell dynamics given by 
\begin{subequations} \label{eq:withincelldimorphic}
\begin{eqnarray}
\dsddt{x}&=& x \left[ b_S - b_D  + \left( b_D - b_S \right) x + \left( b_D - b_F \right) y \right]\\
\dsddt{y}&=&  y \left[ b_F - b_D + \left(b_D - b_S \right)x +  \left( b_D - b_F \right) y \right]. \label{eq:withincelly}
\end{eqnarray}
\end{subequations} 

In Figure \ref{fig:withinphase}, we plot the vector field and sample trajectories for the within-protocell dynamics given by Equation \eqref{eq:withincelldimorphic}, showing that the gene-level dynamics will eventually reach the equilibrium composition consisting entirely of fast replicators. This dominance of fast replicators under gene-level competition is shown analytically in Proposition \ref{prop:withinglobal}, as we see that the all-fast equilibrium is globally asymptotically stable for initial conditions on the interior of the simplex under the dynamics of Equation \eqref{eq:withincelldimorphic}. 

\begin{figure}[ht]
    \centering
    \includegraphics[width = 0.6\textwidth]{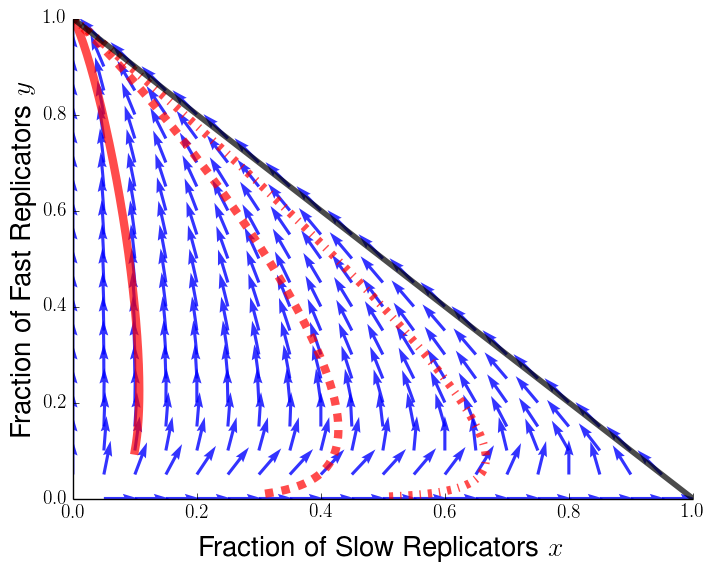}
    \caption{Vector field (blue arrows) and sample trajectories for various initial condition (red lines) for within-cell dynamics of Equation \eqref{eq:withincelldimorphic} plotted on the simplex. While individual trajectories may show initial increases in the fraction of slow replicators, we can see by following the vector field that the within-protocell dynamics eventually push for fixation upon the all-fast composition (shown in the figure by the point $(0,1)$ at the top-left of the simplex).}
    \label{fig:withinphase}
\end{figure}

\begin{proposition}[\bfseries Global Stability of All-Dimer Equilibrium Under Within-Protocell Dynamics] \label{prop:withinglobal}
Consider any point $(x_0,y_0)$ on the interior of the three-type simplex, therefore satisfying the conditions $x_0, y_0 > 0$ and $x_0 + y_0 < 1$. If $b_F > b_S > b_D$, then solutions $(x(t),y(t))$ to Equation \eqref{eq:withincelldimorphic} with initial condition $(x(0),y(0)) = (x_0,y_0)$ satsifies $(x(t),y(t)) \to (0,1)$ as $t \to \infty$. In other words, the all-fast equilibrium is the global attractor for initial compositions featuring an nontrivial mix of slow, fast, and dimer replicators under the within-protocell dynamics.
\end{proposition}
\begin{proof}

Using Equation \eqref{eq:withincelly}, and our assumption on the birth rates that $b_F > b_S > b_D$, we see that within-protocell dynamics for the fast replicator satisfy
\begin{dmath}  \dsddt{y} = y \left[ b_F \left(1 - y\right) + b_D \left( x + y - 1\right) - b_S x \right]  \geq \left( b_F - b_S \right)  y \left( 1 - y\right). \end{dmath}
Denoting by $w(t)$ the solution to the logistic ODE
\begin{equation} \label{eq:wtODE}
    \dsddt{w} = \left( b_F - b_S \right) w (1-w) \textnormal{ with initial condition } w(0) = y_0,
\end{equation}
we see that the fraction of fast replicators $y(t)$ solving Equation \eqref{eq:withincelly} satisfies the comparison principle $y(t) \geq w(t)$. Noting that $w(t) \to 1$ as $t \to \infty$ under our assumptions that $b_F > b_S$ and $y_0 > 0$, we can deduce from our comparison principle that $\liminf_{t \to \infty} y(t) = 1$.

Furthermore, if $y(t) = 1$ at any time $t$, we see from Equation \eqref{eq:withincelly} that 
\begin{equation}
    \dsddt{y} \bigg|_{y = 1} = \left( b_D - b_S \right) x \leq 0,
\end{equation}
and therefore we can deduce that $y(t)$ cannot exceed $1$ for trajectories of Equation \eqref{eq:withincelly} with initial condition $y_0 < 1$. Therefore we can deduce that $\limsup_{t \to \infty} y(t) \leq 1$. We can combine with our previous bound to conclude that $\lim_{t \to \infty} y(t) = 1$, provided an initial condition satisfying $y(t) = y_0 > 0$. 

Finally, we can show by similar arguments that $x(t) \to 0$ as $t \to \infty$ for initial conditions satisfying $0 \leq x_0 \leq 1$ and $y_0 > 0$. Putting these two results together, we see that $(x(t),y(t)) \to (0,1)$ as $t \to \infty$ for any initial on the interior of the simplex. 

\end{proof}

Now that we understand the within-group dynamics for the trimporphic competition for the fast, slow, and dimer replicators satisfying the general ranking of birth rates $b_F > b_S > b_D$, we can consider a special case of the birth rates motivated by the protocell model from Section \ref{sec:protocell}. In that model, we assumed that  $b_{S} = 1$ and $b_F = 1 + s$. In an attempt to extend this model to incorporate the role of dimers, we will consider in our numerical simulations the following replication rates for fast, slow, and dimer replicators
\begin{subequations}
 \label{eq:trimorphicspecialbirthrates}
\begin{align}
b_S &= 1 \\  
b_F &= 1 + s \\
b_D &= 1 - \frac{1}{2 + s} = \frac{1+s}{2+s}
\end{align}    
   
\end{subequations}
Our assumption on the birth rate of dimers is based on the assumption that a replicase would be required to replicate both the slow and fast gene in a dimer in order for the whole dimer to be replicated, so the time taken for a dimer to replicate should be related to the time needed to replicate both the fast and slow gene. Because fast and slow replcators are assumed to replicate at rate $1+s$ and $1$, this means that the mean time to replication is $\tau_F = \ds\frac{1}{1+s} < 1$ for a fast replicator and $\tau_S = 1$ for a slow replicator. If we take as a proxy guess that the average time to replicate a dimer is the sum of average replication times of its component fast and slow genes, then we can assume a replication time of $\tau_D = 1 + \ds\frac{1}{1+s} = \ds\frac{2+s}{1+s}$, resulting in a birth rate for dimers of $b_D = \ds\frac{1+s}{2+s} = 1 - \ds\frac{1}{2+s} < 1$. While there are a variety of possible assumptions that can be made for the gene-level replication rate of dimers, this choice captures the rough idea that the expected replication time for a dimer could exceed those of fast and slow replicators due to the need to replicate more genetic material. %

\subsection{Protocell-Level Reproduction Functions}\label{sec:dimergroupmultilevel}

Now we need to introduce a group-level reproduction function $G(x,y,z)$ describing the rate of protocell-level reproduction as a function of the composition slow, fast, and dimer replicators. In particular, 
to understand various possible complementarities between fast and slow genes, we will look to generalize the group reproduction function $G(x) = x\left(1 - \eta x \right)$ to see what role including dimers can play on between-protocell competition.

Because the slow and fast genes can appear either in their pure monomer form or in a fifty-fifty mix in the dimer form, we would like the group reproduction function to depend on the total fraction of fast genes ($\% \mathrm{Fast}$) and slow genes ($\% \mathrm{Slow}$) in the protocell in either form. One possible group reproduction function possessing properties can be represented schematically by the following formula 
\begin{equation} \label{eq:grouptrimorphicpseudocode} 
G(x,y,z) = G(\%\mathrm{Fast},\%\mathrm{Slow}) = \left( \% \mathrm{Fast} \right) \left( \% \mathrm{Slow} \right) - c \left( \% \mathrm{Slow} \right)^2
\end{equation}
The first term corresponds to the complementary nature of the fast and slow genes, while the second term describes an intrinsic cost (if $c > 0$) or benefit (if $c < 0$) of the presence of slow replicators. Since the fractions of fast and slow genes are given by $\% \mathrm{Fast} = y + \frac{z}{2}$ and $\% \mathrm{Slow} = x + \frac{z}{2}$, we can write the following actual formula for cellular replication as 
\begin{equation} \label{eq:grouptrimorphicgeneral} 
G(x,y,z) = \left( x + \frac{z}{2} \right) \left( y + \frac{z}{2} \right) - c \left( x + \frac{z}{2} \right)^2
\end{equation}
Here, we can relate the parameter $c$ describing the intrinsic cost or benefit of slow genes to the complementarity parameter $\eta$ 
by considering the group payoff function in the absence of dimers (when $z=0$), yielding 
\[ G(x,1-x,0) = x \left( 1 - x \right) - c x^2 = x \left( 1 - \left(c+1\right) x \right) \]
which agrees with the previous group payoff function $G(x)$ when $\eta = c + 1$. For the case of the Luo-Mattingly model ($\eta = 0$), this formulation corresponds to an intrinsic benefit $c = -1$ of slow replicators, which cancels with the quadratic term coming from the complementarity of slow and fast genes $x(1-x)$ and resulting in a linear protocell-level reproduction function. The Fontanari-Serva protocell model ($\eta = 1$), describes the case in which $c = 0$, so there is no intrinsic benefit or cost of slow genes, and the complementary role of slow and fast genes is the only property that impacts the protocell-level reproduction rate.

Going forward, we will use the parameter $\eta = c + 1$ to describe our group payoff functions, so our trimorphic group payoff function can be rewritten as 
\begin{equation} \label{eq:Gxyzeta}
G(x,y,z) = \left( x + \frac{z}{2} \right) \left[ \left( 1 - \eta\right) x + y + \left( 1 - \frac{\eta}{2} \right) z \right]
\end{equation}

Because the compositions of our cells live on the three-type simplex, we know that $z = 1 - x -y$, which allows us to rewrite our trimorphic group reproduction function as 
\begin{equation} \label{eq:Gtrixy}
G(x,y) = \frac{1}{4} \left(x - y + 1 \right) \left[ 2 - \eta  \left( x - y + 1 \right) \right]
\end{equation}

We can illustrate how different compositions $(x,y)$ shape the protocell-level reproduction rate by plotting $G(x,y)$ on the three-type simplex. In Figure \ref{fig:Gofxyplots}, we illustrate $G(x,y)$ for the complementarity scenarios characterized by $\eta = 1$ (left), $\eta = 0.7$ (center), and $\eta = 0$ (right). In both cases, we see that the collective reproduction rate is constant along lines with slope $1$. We can understand this observation analytically by noting that, on the simplex satisfying $z = 1 - x - y$, the percentage of slow genes in a protocell is given by $\%\mathrm{Slow} = x + \frac{z}{2} = \frac{1}{2} \left(x - y + 1\right)$, and therefore we see from Equation \eqref{eq:Gtrixy} that the level sets of $G(x,y)$ are given by lines of the form $y = x + 1 - \%\mathrm{Slow}$. When protocell-level replication is maximized by the mix of slow and fast genes given by $\%\mathrm{Slow} = \frac{1}{2\eta}$ and $\%\mathrm{Fast} = 1 - \frac{1}{2 \eta}$ (when $\eta > \frac{1}{2}$), we then see that collective reproduction is maximized on the line $y = x + 1 - \frac{1}{\eta}$. In the alternate case when collective-replication is maximized by the all-slow composition (when $\eta \geq \frac{1}{2}$), this line of maximizers is given by $y = x - 1$, which only intersects the simplex at the all-slow equilibrium $(1,0)$.

\begin{figure}[ht]
    \centering
    \includegraphics[width = 0.32\textwidth]{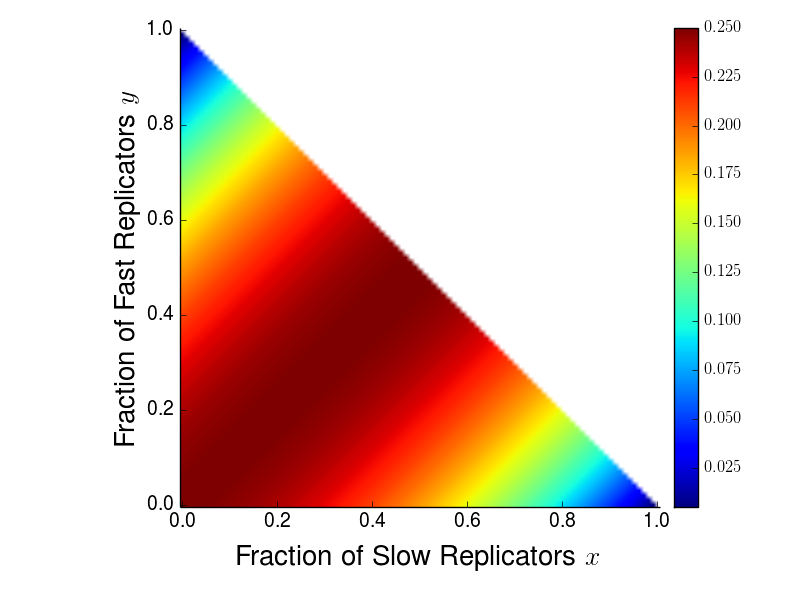}
    \includegraphics[width = 0.32\textwidth]{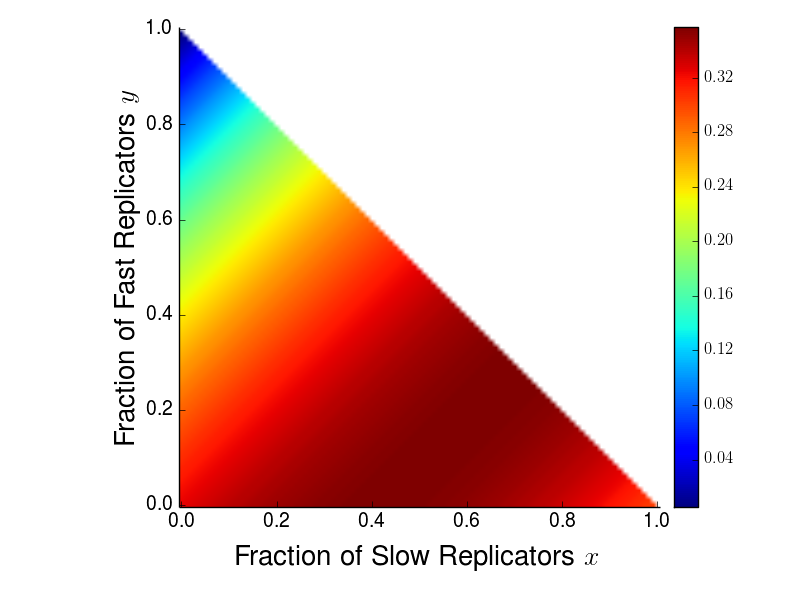}
    \includegraphics[width = 0.32\textwidth]{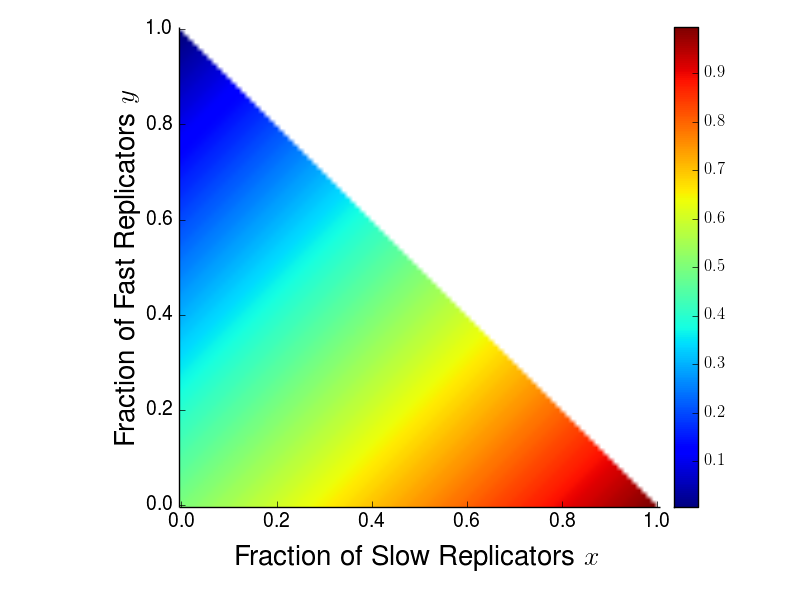}
    \caption{Heatmap of protocell-level reproduction rates $G(x,y)$ for $\eta = 1$ (left),  $\eta = 0.7$ (center), and $\eta = 0$ (right). For $\eta = 1$, collective reproduction $G(x,y)$ is maximized along the the line $y = x$, so both the all-dimer composition $(x,y) = (0,0)$ and the fifty-fifty mix of fast and slow replicators $(x,y) = (\tfrac{1}{2},\tfrac{1}{2})$. For $\eta = 0.7$, collective reproduction $G(x,y)$ is maximized along the line $y = x - \frac{3}{7}$, which passes through the composition $(x,y) = (\nicefrac{1}{2 \eta},1 - \nicefrac{1}{2 \eta}) \, |_{\eta = \nicefrac{7}{10}} = (\nicefrac{5}{7},\nicefrac{2}{7})$ maximizing $G_{FS}(x)$ on the fast-slow edge and the composition $(\nicefrac{3}{7},0)$ maximizing protocell-level replication on the fast-dimer edge of the simplex. For $\eta = 0$, $G(x,y)$ is maximized on the line $y = x - 1$, which only intersects the simplex at the all-slow composition $(1,0)$.}
    \label{fig:Gofxyplots}
\end{figure}

Combining the within-cell dynamics described by Equation \eqref{eq:withincelldimorphic} and competition for cellular birth-death dynamics according to the group reproduction function from Equation \eqref{eq:Gtrixy}, we can now describe multilevel selection in the fast-slow-dimer system in limit of infinitely many protocells and infinitely many genes per protocell. We assume as described in Section \ref{sec:dimerwithin} that within-group replication events for type $X$ take place at rates $1 + w_I b_{X}$ and that protocell-level replication of protocells featuring fractions $x$ of slow replicators and $y$ fast replicators take place at rate $\Lambda \left( 1 + w_G G(x,y)\right)$. Denoting the probability density for a cell with genes composes with fractions of $(x,y,1-x-y)$ of slow, fast, and dimer replicators at time $t$ by $\rho(t,x,y)$, we show in Section \ref{sec:derivationtrimorphic} that this density evolves in time according to
\begin{dmath} \label{eq:multileveltrimorphic}
\dsdel{\rho(t,x,y)}{t} = - \dsdel{}{x} \left[x \left( b_S - b_D  + \left( b_D - b_S \right) x + \left( b_D - b_F \right) y \right) \rho(t,x,y) \right] - \dsdel{}{y} \left[ y \left( b_F - b_D + \left(b_D - b_S \right)x +  \left( b_D - b_F \right) y \right) \rho(t,x,y) \right]
+ \lambda \rho(t,x,y) \left[G(x,y,1-x-y) - \int_{0}^1 \int_0^{1-x} G(u,v,1-u-v) \rho(t,u,v) dv du \right],
\end{dmath}
where $\lambda := \frac{\Lambda w_G}{w_I}$ again describes the relative selection strength at the two levels. The characteristic curves are given by the system of ODEs from Equation \eqref{eq:withincelldimorphic}. Because we now have a system of two characteristic ODEs, we cannot apply the same strategy for analyzing the long-time behavior of Equation \eqref{eq:multileveltrimorphic} that we have used for our multilevel selection models with two types of individuals. To make some progress, we will now explore the dynamics of this model reduced to the fast-dimer and slow-dimer edges in Sections \ref{sec:fastdimer} and \ref{sec:slowdimer}, respectively. In particular, we look to compare the steady-state behavior on these edges of the simplex with the behavior of the dynamics from the fast-slow edge studied in Section \ref{sec:protocell}, and to see the ways in which introduction of dimers can help to establish coexistence of fast and slow genes at steady state and to help to erase the shadow of lower-level selection. In Section \ref{sec:trimorphicnumerics}, we take a preliminary look at a strategy for extending our finite volume approach to describe our fast-slow-dimer multilevel dynamics in which cell compositions live on the three-type simplex.

\section{The Effect of Dimer Replicators on Long-Time Coexistence of Fast and Slow Genes: Dynamics on the Edges of the Simplex} \label{sec:simplexedgedynamics}

We now consider the dynamics of the protocell model of Equation \eqref{eq:multileveltrimorphic} when the state space is restricted to edges of the slow-fast-dimer simplex and competition takes place between protocells that feature at most two of the possible replicators. In Section \ref{sec:fastdimer}, we consider competition on the fast-dimer edge of the simplex, showing that multilevel competition between protocells featuring fast and dimer replicators can promote coexistence of the fast and slow genes, even for the case in which $\eta = 1$ and no coexistence was possible with protocells featuring only fast and slow replicators. In Section \ref{sec:fdfscomparison}, we compare the threshold selection strengths and steady-state protocell-level fitness achieved on the fast-dimer edge of the simplex with the analogous quantities derived on the fast-slow edge of the simplex in Section \ref{sec:protocell}. Finally, in Section \ref{sec:slowdimer}, we study multilevel competition on the slow-dimer edge of the simplex, showing how the long-time behavior varies depending on whether all-slow or all-dimer protocells replicate faster under protocell-level competition.

\subsection{Dynamics on Fast-Dimer Edge of Simplex} \label{sec:fastdimer}

In this section, we introduce the reduced dynamics of Equation \eqref{eq:multileveltrimorphic} when protocells are restricted to compositions on the fast-dimer edge of the simplex. We show that the protocell-level replication function is always maximized by the all-dimer protocell on this edge of the simplex, and characterize the threshold relative intensity of between-protocell competition required to allow the long-time coexistence of fast and dimer replicators. We discuss the convergence of the population to steady state densities for sufficiently strong between-protocell competition in Proposition \ref{prop:longtimefastdimer}, and we illustrate in Figure \ref{fig:eta1density} how steady state densities support increasing levels of dimers as between-protocell competition increase. We formalize this observation in Proposition \ref{prop:FDmodal}, showing that the modal composition of dimers at steady state increases to 100 percent in the limit of infinite strength of between-protocell competition. This shows how the use of dimers can help to support coexistence of the fast and slow genes via multilevel selection, helping to overcome the limitations provided by the shadow of lower-level selection seen on the fast-slow edge of the simplex in Section \ref{sec:protocell}.   

On the fast-dimer edge of the simplex (where $x=0$, $y = 1 - z$), we will describe the composition of a protocell by its fraction $z$ of dimer replicators. We introduce the protocell-level replication rate $G_{FD}(z)$ for compositions on the fast-dimer edge, and can use Equation \eqref{eq:Gxyzeta} to see that the the replication rate reduces to
\begin{equation} \label{eq:GFDz} G_{FD}(z) = G(0,1-z,z) = \frac{z}{2} - \frac{\eta}{4} z^2.  \end{equation}
We can then compute that 
\[ G'_{FD}(z) = \frac{1}{2} \left( 1 - \eta z \right) > 0 \textnormal{ for any } \eta \in [0,1] \textnormal{ and } z \in [0,1), \] so therefore the fraction of dimers $z^*_{FD}$ maximizing the protocell reproduction function $G_{FD}(z)$ is given by
\begin{equation} \label{eq:zstarFD}
    z^*_{FD}(\eta) = 1 \: \: \mathrm{for} \: \: \eta \in [0,1].
\end{equation}

In other words, for multilevel competition in protocells composed only of fast and dimer replicators, protocell-level replication always favors compositions with as many dimers as possible. 
In particular, for any value of $\eta$ in which $G_{FS}(z)$ has an intermediate cell fitness optimum, simply replacing our slow replicators with slow-fast dimers produces a regime in which cells are best off with all-dimers rather than a mix of dimers and fast replicators. We illustrate these properties of the collective replication rate $G_{FD}(z)$ in Figure \ref{fig:Gzfdfunction.png}, showing that $G_{FD}(z)$ increases with $z$ and decreases with $\eta$. 

\begin{figure}[htp!]
    \centering
    \includegraphics[width = 0.6\textwidth]{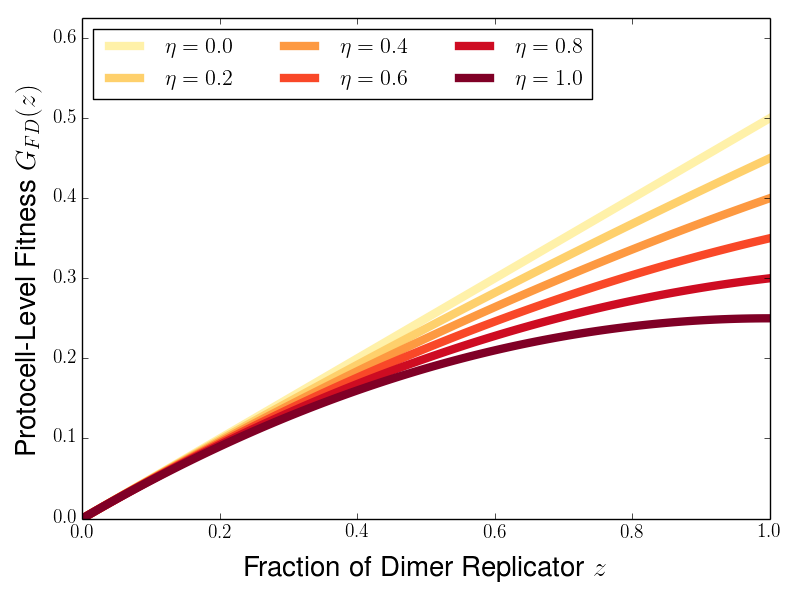}
    \caption{Protocell-level reproduction rates $G_{FD}(z)$ on fast-dimer edge of the simplex for various levels of the complementarity parameter $\eta$. For all values of $\eta$, $G_{FD}(z)$ is an increasing function of the fraction of dimers $z$, and the collective reproduction rate is maximized by the all-dimer composition.}
    \label{fig:Gzfdfunction.png}
\end{figure}

We can study the within-protocell dynamics on the fast-dimer edge by applying the restriction $x = 0$ and $y = 1-z$ to Equation \eqref{eq:withindimer}. This allows us to see that, in protocells featuring only fast and dimer replicators, the fraction of dimers evolves according to the follow gene-level replicator equation  
\begin{equation} \label{eq:FDwithin}
\dsddt{z(t)} = z \left[b_D - \left(b_S x + b_F y + b_D z\right) \right] \bigg|_{\substack{x = 0 \\ y = 1-z}} = 
- \left( b_{F} - b_{D} \right) z \left( 1 - z \right).
\end{equation}
Notably, this ODE is of the form of the characteristic curves given by Equation \eqref{eq:characteristicsgeneral} with the net gene-level replication function $\pi_{FD}(z) = b_F - b_D > 0$. 

We can now study the combined effects of the protocell-level reproduction function $G_{FD}(z)$ of Equation \eqref{eq:GFDz} and the gene-level dynamics of Equation \eqref{eq:FDwithin} to study how the composition of protocells evolves due to multilevel competition on the fast-dimer edge of the simplex.  Coupling the dynamics at the two levels, we describe the probability density $g(t,z)$ of protocells composed of fraction $z$ dimers and $1-z$ fast replicators at time $t$ by following the multilevel PDE
\begin{dmath} \label{eq:FDmultilevelPDE}
\dsdel{g(t,z)}{t} = \dsdel{}{z} \left[ \left(b_F - b_D \right) z (1-z)   g(t,z) \right] + \lambda g(t,z) \left[ G_{FD}(z) - \int_0^1 G_{FD}(w) g(t,w) dw \right],
\end{dmath}
which is a special case of Equation \eqref{eq:generalPDEreplicator} with the replication functions $\pi_{FS}(z) = b_F - b_D$ and $G_{FD}(z) = \frac{z}{2} \left( 1 - \frac{\eta z}{2} \right)$. We can apply the results from Section \ref{sec:existingresults} to study how multilevel competition can help to promote coexistence of fast and slow genes in protocells on the fast-dimer edge of the simplex.

First, we look to study the density steady states of Equation \eqref{eq:FDmultilevelPDE}. Using Equation \eqref{eq:plambdatheta} and the fact that $G_{FD}(1) = \frac{1}{2} \left( 1 - \frac{\eta}{2}\right)$, $G_{FD}(0) = 0$, and $\pi_{FD}(x) \equiv b_F - b_D$, we see that the steady states are probability densities of the form
\begin{subequations} \label{eq:glambdatheta}
     \begin{align}
         g^{\lambda}_{\theta}(z) &= Z_g^{-1} \: \mathlarger{z}^{\mathlarger{\left[\left(\nicefrac{1}{2 (b_F - b_D)}\right) \lambda \left(1 - \nicefrac{\eta}{2}\right) - \theta - 1\right]}} \left(\mathlarger{1 - z}\right)^{\mathlarger{\theta - 1}} \exp\left( - \lambda \int_z^{1} \frac{C_{FD}(u)}{b_F - b_D} du \right) \\
         Z_g &= \int_0^1 \mathlarger{w}^{\mathlarger{\left[\left(\nicefrac{1}{2 (b_F - b_D)}\right) \lambda \left(1 - \nicefrac{\eta}{2}\right) - \theta - 1\right]}} \left(\mathlarger{1 - w}\right)^{\mathlarger{\theta - 1}} \exp\left( - \lambda \int_w^{1} \frac{C_{FD}(u)}{b_F - b_D} du \right) dw,
     \end{align}
\end{subequations}
where $C_{FD}(z)$ is given by applying the replication rates on the fast-dimer edge to Equation \eqref{eq:lambdaCofx}. 
Noting that $\pi_{FD}(z)$ is a constant and that the between-protocell replication rates satisfy 
\begin{subequations}
  \begin{align}
G_{FD}(z) - G_{FD}(0) &= \frac{z}{2} \left( 1 - \frac{\eta z}{2} \right) \\
G_{FD}(z) - G_{FD}(1) &= \frac{1}{2} \left( 1 - z \right) \left[\frac{\eta}{2} - 1 + \frac{\eta z}{2} \right],
\end{align}   
\end{subequations}
we can see from Equation \eqref{eq:lambdaCofx} that
\begin{equation}
    - \lambda C_{FD}(z) = \lambda \left( \frac{G_{FD}(z) - G_{FD}(0)}{z} \right) + \lambda \left( \frac{G_{FD}(z) - G_{FD}(1)}{z} \right) = \frac{\lambda \eta}{4}.
\end{equation}
This allows us to further compute that
\begin{equation}
    - \lambda \int_z^1 \frac{C_{FD}(u)}{b_F - b_D} du = \frac{\lambda \eta}{4 \left(b_F - b_D\right)} \left( 1 - z \right),
\end{equation}
and, after introducing the constant $\tilde{Z}_g = Z_g e^{- \nicefrac{\lambda \eta}{4 (b_F - b_D)}}$, we can write our steady states in the form
\begin{equation} \label{eq:FDsteadysimple}
    g^{\lambda}_{\theta}(z) = \tilde{Z}_g^{-1} \: \mathlarger{z}^{\mathlarger{\left[\left(\nicefrac{1}{2 (b_F - b_D)}\right) \lambda \left(1 - \nicefrac{\eta}{2}\right) - \theta - 1\right]}} \left(\mathlarger{1 - z}\right)^{\mathlarger{\theta - 1}} \exp\left(- \frac{\lambda \eta z}{4 \left(b_F - b_D \right)}  \right).
\end{equation}

We note from Equation \eqref{eq:FDsteadysimple} that a density given by $g^{\lambda}_{\theta}(z)$ will be integrable provided that $\lambda$ exceeds the following threshold value
\begin{equation} \label{eq:lambdastarFDbvalues}
\lambda^*_{FD}(\eta) = \frac{\left(b_F - b_{D}\right)\theta}{G_{FD}(1) - G_{FD}(0)}  = \frac{4\left(b_F - b_{D}\right)\theta}{2  - \eta}.
\end{equation}
In particular, we see that $\lambda^*_{FD}(\eta)$ remains finite for all possible complementarity parameters $\eta \in [0,1]$, so multilevel competition on the fast-dimer edge of the simplex can always produce coexistence between fast and slow genes provided that between-protocell competition is sufficiently strong. 

This threshold quantity also determines the long-time behavior Equation \eqref{eq:FDmultilevelPDE} given an initial measure with H{\"o}lder exponent of $\theta$ near the all-dimer composition $z = 1$. In Proposition \ref{prop:longtimefastdimer}, we summarize our application of Theorems \ref{thm:PDconvergencetosteady} and \ref{thm:PDconvergencetodelta} for the fast-dimer dynamics, showing that fast replicators take over the population when $\lambda \leq \lambda^*_{FD}(\eta)$, while the population reaches a density steady state supporting both fasts and dimers if $\lambda > \lambda^*_{FD}(\eta)$.

\begin{proposition}  \label{prop:longtimefastdimer} Suppose the population of protocells composed of fast and dimer replicators has initial measure $\mu_0(dz)$ with H{\"o}lder exponent of $\theta$ near $z=1$, and consider a measure-valued solution $\mu_t(dx)$ to \eqref{eq:protocellPDEG}. Then, in the limit as $t \to \infty$, the solution $\mu_t(dz)$ to Equation \eqref{eq:FDmultilevelPDE} will have the following long-time behavior \begin{displaymath} \mu_t(dz)  \rightharpoonup \left\{
     \begin{array}{cr}
       \delta(z) & : \lambda \leq \lambda^*_{FS}(\eta) \\ % 
       \tilde{Z}_g^{-1} \: \mathlarger{z}^{\mathlarger{\left[\left(\nicefrac{1}{2 (b_F - b_D)}\right) \lambda \left(1 - \nicefrac{\eta}{2}\right) - \theta - 1\right]}} \left(\mathlarger{1 - z}\right)^{\mathlarger{\theta - 1}} \exp\left(-\frac{\lambda \eta z }{4(b_F - b_D)}\right)  & : \lambda > \lambda^*_{FS}(\eta) 
     \end{array} \right. .
     \end{displaymath} 

\end{proposition}

In Figure \ref{fig:eta1density}, we display the steady state solutions for various values of $\lambda$ and the choice of group replication tradeoff parameter $\eta = 1$. For the slow-fast competition, this choice of complementarity parameter resulted in no coexistence of fast and slow replicators at steady state. We see that the densities supports increasing fractions of dimers as the relative strength of between-protocell competition $\lambda$. In particular, we see that the mean and modal fraction of dimers appears to approach 1 as $\lambda$ increases, suggesting that multilevel competition on the dimer-fast edge of the simplex can approach the optimal composition of a fifty-fifty mix of fast and slow genes if there is sufficiently strong between-protocell competition.   

\begin{figure}[htp!]
    \centering
    \includegraphics[width = 0.7\textwidth]{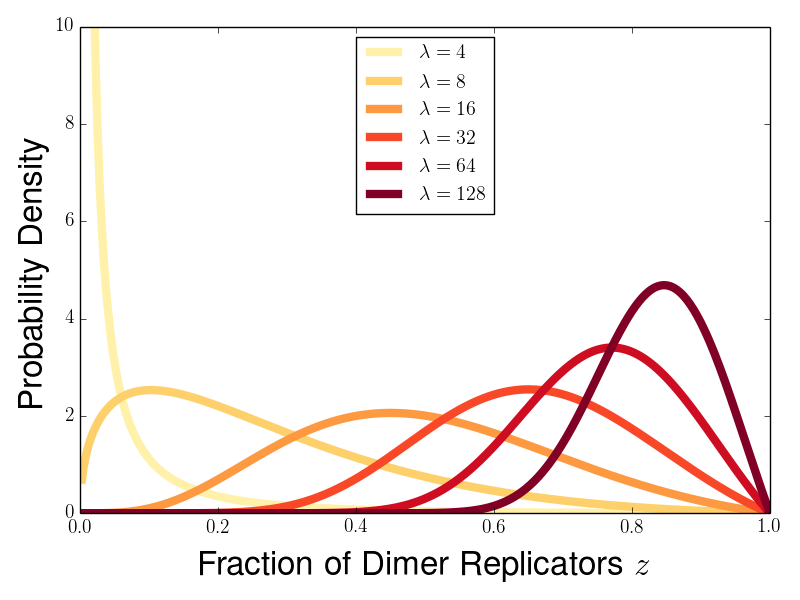}
    \caption[Steady state densities for competition on fast-dimer edge.]{Steady state densities for multilevel selection given by Equation \eqref{eq:FDsteadysimple} on fast-dimer edge of the simplex for $\eta = 1$ and various values of $\lambda$. We note that $\eta = 1$ is the parameter for which no coexistence of fast and slow genes is possible on the slow-fast edge, but that densities can reach all the way to the full-dimer group as $\lambda \to \infty$. Other parameters are fixed as $s = 1$ and $\theta = 2$. }
    \label{fig:eta1density}
\end{figure}

We formalize this intuition in Proposition \ref{prop:FDmodal}, in which we compute the modal composition of dimers $\hat{z}_{\lambda}:= \argmax_{z \in [0,1]} g^{\lambda}_{\theta}(z)$ for the steady state family of densities given by $g^{\lambda}_{\theta}(z)$. In the limit of infinite strength of between-protocell competition, we see that $\lim_{\lambda \to \infty} \hat{z}_{\lambda} = 1$, and therefore, for any complementarity parameter $\eta \in [0,1]$, multilevel selection on the fast-dimer edge of the simplex as many dimers as possible for sufficiently strong protocell-level selection.

\begin{proposition}[\bfseries Most Abundant Composition at Steady State Approaches All Dimers in the Limit of Infinite Intensity of Between-Protocell Competition] \label{prop:FDmodal}
Consider the steady state density $g^{\lambda}_{\theta}(z)$ and suppose that $\frac{\lambda}{2} \left(1 - \frac{\eta}{2}\right) > \left(b_F - b_D \right) \left(\theta + 1\right)$ and $\theta \geq 1$. Then, for $\eta \in (0,1]$, the most abundant composition at steady state $\hat{z}_{FD}^{\lambda} := \argmax_{z \in [0,1]} g^{\lambda}_{\theta}(z)$ is given by 
\begin{equation} \label{eq:modalFDlambda}
    \hat{z}^{\lambda}_{FD} = \frac{\lambda - 4(b_F - b_D)  - \sqrt{\left(\lambda - 4 (b_F - b_D) \right)^2 - 4 \lambda \eta \left[\frac{ \lambda}{2} \left(1-\frac{\eta}{2}\right) - \left(b_F - b_D\right) \left(\theta +  1 \right) \right]}}{ \lambda \eta}.
\end{equation}

Furthermore, in the limit of infinite intensity of between-protocell composition, the modal composition $\hat{z}^{\infty}_{FD} := \lim_{\lambda \to \infty} \hat{z}_{FD}^{\lambda}$ satisfies 
\begin{equation} \label{eq:hatzFDinf}
    \hat{z}^{\infty}_{FD} = \frac{1}{\eta} - \left( \frac{1}{\eta} - 1\right) = 1
\end{equation}
Finally, noting from Equation \eqref{eq:zstarFD} that $z^*_{FD} = 1$, we see that the optimal protocell composition is achieved by the modal protocell at steady state as $\lambda \to \infty$. 
\end{proposition}

\subsection{Comparison of Dynamics of Fast-Slow and Fast-Dimer Edges of the Simplex}
\label{sec:fdfscomparison}

From our analysis of multilevel competition on the fast-slow and fast-dimer edges of the simplex in Sections \ref{sec:protocellongtime} and \ref{sec:fastdimer}, we have shown that while dimers face an additional gene-level disadvantage relative to slow replicators in competition against fast replicators, an all-dimer protocell can obtain a greater collective advantage than an all-slow protocell in between-protocell competition. In this section, we will study how these costs and benefits of dimerization play out under our multilevel dynamics, characterizing the parameter space in which competition on the fast-dimer edge of the simplex can more easily facilitate coexistence of fast and slow genes or produce a higher average protocell-level fitness in comparison to the baseline protocell model on the fast-slow edge of the simplex. In particular, we find that the fast-slow edge outperforms the fast-dimer edge for any relative selection strength $\lambda$ when the all-fast protocell has a collective advantage over the all-dimer protocell $(\eta < \frac{2}{3}$), while the fast-dimer edge can do better for any relative selection strength if fast and slow genes are sufficiently complementary under between protocell competition ($\eta$ close enough to one). There also exist intermediate degrees of complementarity (intermediate values of $\eta$) for which fast-slow competition produces a greater collective outcome for weak between-protocell competition, while fast-dimer competition does better when protocell-level competition becomes sufficiently strong. Taken together, these different behaviors highlight the effects of the complementarity parameter $\eta$, relative selection strength $\lambda$ of protocell-level competition, and the gene-level advantage of fast replicators $b_F - b_S$ in determining whether fast or dimer replicator replicators are more conducive to producing coexistence of the fast and slow genes.

This difference between competition on the fast-slow and fast-dimer edges of the simplex is particularly stark in the case $\eta = 1$, where no coexistence of the fast and slow genes is possible on the fast-slow edge. As an illustration of this case, we present in Figure \ref{fig:fsfdtrajectorycompare} the trajectories of group compositions under finite volume numerical simulations for the dynamics on the two edges of the simplex for $\lambda = 10$, showing that the population converges to the all-fast equilibrium under fast-slow competition while dimer and fast replicators can coexist in the long-run for sufficiently strong between-protocell composition. Noting that the initial uniform distributions feature an overall composition of half slow genes on the fast-slow edge and one-third slow genes on the fast-dimer edge, we see that the two-level dynamics can produce coexistence of the fast and slow genes on the fast-dimer edge even when the initial population on the fast-slow edge of the simplex has more slow genes than fast genes. This highlights the insight from the threshold quantity $\lambda^*$ from Equation \eqref{eq:lambdastargeneral} showing that coexistence of two replicators depends on the collective ability for the monomorphic states of all-dimers or all-slows to outperform the all-slow composition under protocell-level competition. 

\begin{figure}[htp]
    \centering
    \includegraphics[width = 0.48\textwidth]{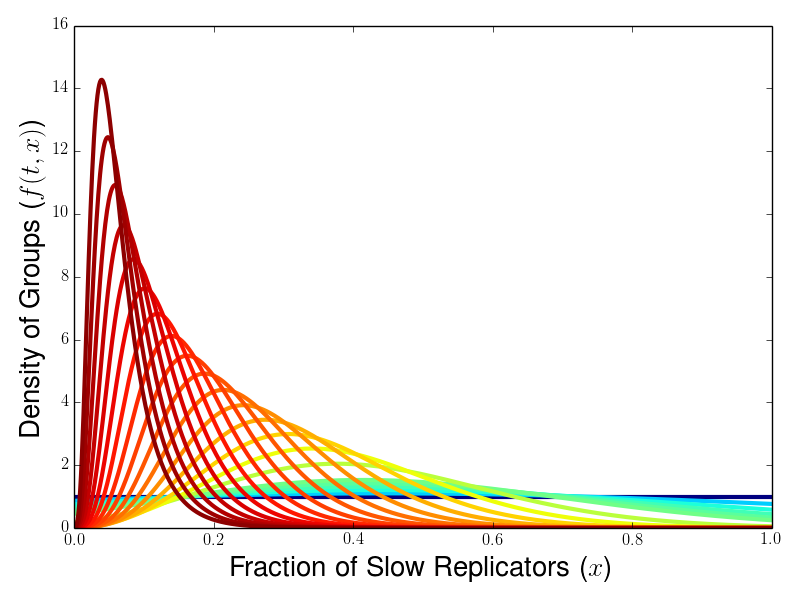}
     \includegraphics[width = 0.48\textwidth]{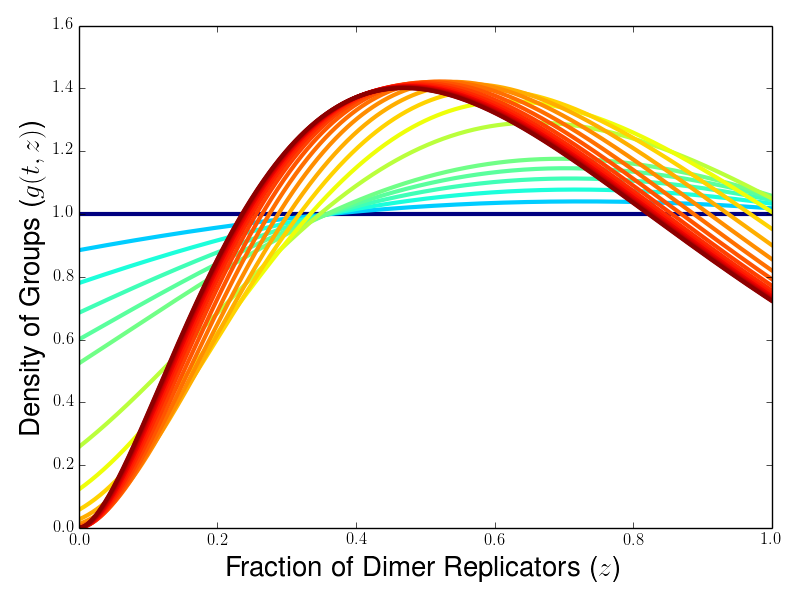}
    \caption{Comparison of numerical dynamics under finite volume discretization on the fast-slow (left) and fast-dimer (right) edges of the simplex for uniform initial protocell compositions and complementarity parameter $\eta = 1$. The color of the densities corresponds to the time at which the density is achieved in the numerical times, with early times represented by blue curves and later times represented by red curves. Between-protocell selection strength fixed as $\lambda = 10$ and gene-level birth rates given by $b_S = 1$, $b_F = 2$, and $b_D = \frac{2}{3}$ (corresponding to fast-replicator advantage of $s = 1$).}
    \label{fig:fsfdtrajectorycompare}
\end{figure}

We can now try to explore the parameter regimes in which the use of dimers helps or hurts establishment of the slow gene via multilevel selection relative to our baseline model of competition on the fast-slow edge of the dimer. One way to measure this is by comparing threshold levels $\lambda^*_{FS}$ and $\lambda^*_{FD}$ of the relative between-cell competition intensity at which slow replicators and dimers can coexist with fast replicators at steady state, respectively. Using our typical assumed gene birth rate parameters from Equation \eqref{eq:trimorphicspecialbirthrates} $b_S = 1$, $b_F = b_S$, and $b_{D} = 1 -  \frac{1}{2+s}$, we see that we can write these two thresholds as
\begin{subequations} \label{eq:lambdastarFSFD}
\begin{alignat}{2}
    \lambda^*_{FS} &= \frac{\left(b_F - b_S \right) \theta}{G_{FS}(1) - G_{FS}(0)} &&= \frac{s \theta}{1 - \eta} \\
    \lambda^*_{FD} &= \frac{\left(b_F - b_D \right) \theta}{G_{FD}(1) - G_{FD}(0)}  &&= \left( s + \frac{1}{2+s} \right) \left(\frac{2\theta}{1 - \frac{\eta}{2}}\right)
\end{alignat}
\end{subequations}

From the the expressions for threshold selection strength in terms of the generic gene-level birth rates $b_F$, $b_S$, and $b_D$, we see that the threshold $\lambda$ needed to achieve a density steady state is lower for the fast-dimer edge than the fast-slow edge ($\lambda^*_{FD} < \lambda^*_{SD}$) when the group reproduction complementarity parameter $\eta$ is above a critical level
\begin{equation} \label{eq:etacbvalues}
    \eta_c := \frac{2 \left( b_F + b_S \right) - 4 b_D}{3 b_F + b_D - 4 b_D} = \frac{4 \left(b_S - b_D\right) + 2 \left(b_F - b_S\right)}{4 \left(b_S - b_D\right) + 3 \left(b_F - b_S\right)} \geq \frac{2}{3}.
\end{equation}
We see that this critical complementarity parameter satisfies the properties that $\eta_c \to 1$ when $b_F \to b_S$ and that $\eta_c \to \frac{2}{3}$ when $b_F - b_S \to \infty$. In the case of the special birth rates parameterized in terms of the gene-level advantage $s$ for fast replicators, we can can see that $\eta_c$ takes the form 
\begin{equation} \label{eq:etacs}
 \eta_c^s = 
\frac{2 s^2 + 4s + 4}{3 s^2 +  6s + 4}.
\end{equation}

The fact that the relative rankings of the thresholds depends on the complementarity parameter highlights the fact that introducing dimers provides two countervailing effects on the threshold $\lambda^*_{FD}$ relative to $\lambda^*_{FS}$: dimerization increases the protocell-level advantage for dimorphic compositions over all-fast compositions in the denominator of Equation \eqref{eq:lambdastarFSFD} while increasing the gene-level advantage for fast replicators in the numerator of Equation \eqref{eq:lambdastarFSFD}. When $\eta \in [\frac{2}{3},\eta_c^s)$, the gene-level disadvantage hurts more than the group-level advantage helps, making $\lambda^*_{FD} > \lambda^*_{FS}$ for those tradeoff parameters. When $\eta \in (\eta_c^s,1]$, the protocell-level advantages outweigh the individual-level disadvantages, allowing $\lambda^*_{FD} < \lambda^*_{FS}$ in this regime. We further note from the observation that $\eta_c \geq \frac{2}{3}$ that the threshold to achieve coexistence is always lower on the fast-slow edge than on the fast-dimer edge for complementarity scenarios in which between-protocell competition favors all-slow compositions over all-dimer compositions.

To extend the comparison between the costs and benefits of dimerization, we can write the average protocell-level fitness achieved as steady state under competition on the fast-slow and fast-dimer edges of the simplex for our special family of gene-level birth rates. Using Equations \eqref{eq:Gaveragelambda}, the gene-level and protocell-level replication rates on the two edges, and noting that 
\[ b_F - b_D = 1 + s - \left(1 - \frac{1}{2+s}\right) = \frac{(s+1)^2}{2+s}, \] we see that
\begin{subequations} \label{eq:FSFDavgG}
    \begin{align}
        \langle G_{FS}(x) \rangle_{f^{\lambda}_{\theta}} &= \left\{
        \begin{array}{cr}
           0  &:  \lambda < \lambda^*_{FS}\\
           1 - \eta - \ds\frac{s \theta}{\lambda}  &: \lambda \geq \lambda^*_{FS} 
        \end{array} \right. \\
          \langle G_{FD}(z) \rangle_{g^{\lambda}_{\theta}} &= 
        \left\{ \begin{array}{cr}
           0  &:  \lambda < \lambda^*_{FD}\\
           \ds\frac{1}{2} \left( 1 - \ds\frac{\eta}{2} \right) - \ds\frac{\left(s-1\right)^2 \theta}{(2 + s) \lambda}  &: \lambda \geq \lambda^*_{FD}  
        \end{array} \right.
    \end{align}
\end{subequations}
In Figure \ref{fig:FSFDGcompare}, we compare the average protocell-level fitnesses from Equation \eqref{eq:FSFDavgG} as a function of the relative strength $\lambda$ of protocell-level competition for gene-level advantage $s = 1$ of fast replicators and for complementarity parameters $\eta = 0.705$ (left) and $\eta = 0.9$ (right). When $\eta = 0.705$ (Figure \ref{fig:FSFDGcompare},left), we see that the protocell-level fitness on the fast-slow edge of the simplex reaches a nonzero level at a lower value of $\lambda$ than on the fast-dimer edge of the simplex, but that, for sufficiently large $\lambda$, the protocell-level fitness on the fast-dimer edge surpasses that of the fast-slow edge. In this case, dimerization can make coexistence of the fast and slow gene more difficult for a range of lower $\lambda$ values, but confers a great collective benefit to the population at higher values of $\lambda$. When $\eta = 0.9$ (Figure \ref{fig:FSFDGcompare}, right), we see that the collective fitness on the fast-dimer edge of simplex first achieves a nonzero value at a lower $\lambda$ than on the fast-slow edge of the simplex, and then the fast-dimer competition produces a better collective outcome than the fast-slow edge of simplex for all higher relative selection strengths. In this regime, protocells composed of fast and dimer replicators will outperform protocells composed of fast and slow replicators given an equal relative selection strength and respective H{\"o}lder exponents $\theta$ near the all-dimer equilibrium (for fast-dimer competition) and the all-slow equilibrium (for fast-slow competition).

\begin{figure}[htp!]
    \centering
    \includegraphics[width = 0.48\textwidth]{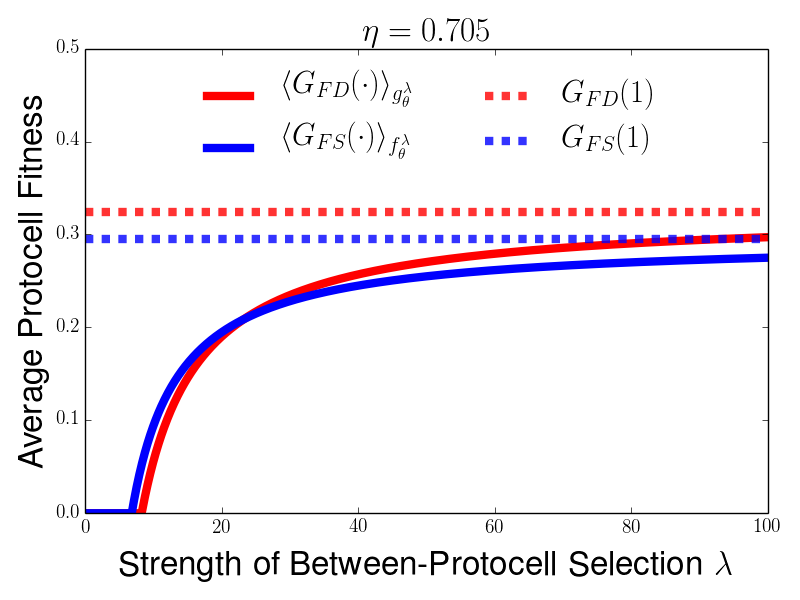}
    \includegraphics[width = 0.48\textwidth]{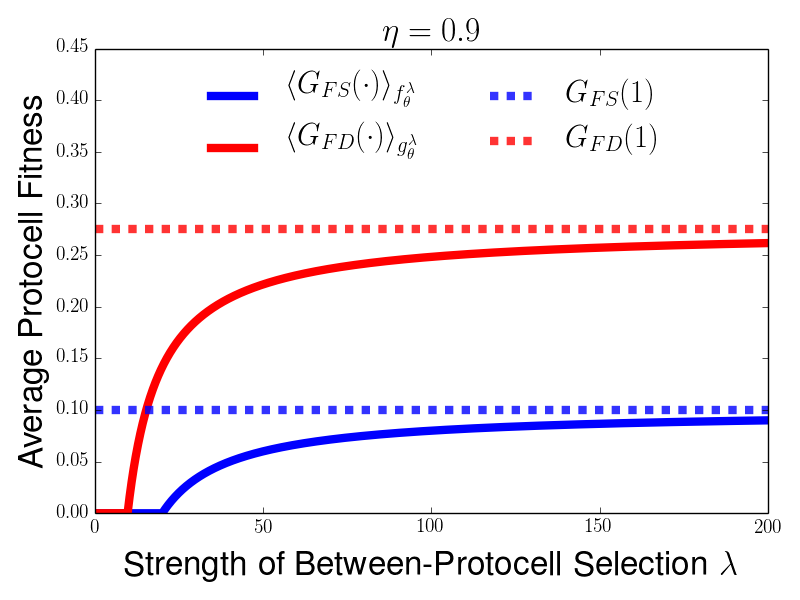}
    \caption{Comparison between the average protocell-level fitness at the long-time steady state for uniform initial distribution given by Equation \eqref{eq:FSFDavgG} for the fast-slow (solid blue line) and fast-dimer (solid red line) edges of the simplex, plotted as a function of the relative selection strength $\lambda$. Average protocell-level fitnesses are provided for complementarity parameter $\eta = 0.705$ (left) and $\eta = 0.9$ (right), and the gene-level advantage of fast replicators $s = 1$ and H{\"o}lder exponent $\theta = 1$ are held constant for the two panels. We also compare these collective fitnesses to the maximum protocell-level reproduction rates achieved as $\lambda \to \infty$ given by $G_{FS}(1) = 1 - \eta$ (dashed blue line) and $G_{FD}(1) = \frac{1}{2} \left(1 - \frac{\eta}{2}\right)$ (dashed red line) on the fast-slow and fast-dimer edges of the simplex.}
    \label{fig:FSFDGcompare}
\end{figure}

We can also study the impact of dimerization on the collective outcomes achieved at steady state in the limit of infinitely strong between-protocell competition. Using Equation \eqref{eq:FSFDavgG}, we see that the maximal average protocell-level fitnesses on the fast-slow and fast-dimer edges are given by $G_{FS}(1) = 1 - \eta$ and $G_{FD}(1) = \frac{1}{2}\left(1 - \frac{\eta}{2}\right)$, respectively. In particular, this tells us that the fast-dimer edge produces a better outcome when $\eta > \frac{2}{3}$. We can also compare collective outcomes based upon the modal compositions at steady state. These were found in Equation \eqref{eq:hatxFSinf} and $\eqref{eq:hatzFDinf}$, and are given by  $\hat{x}_{FS}^{\infty} = \frac{1}{\eta} - 1$ on the fast-slow edge and $\hat{z}_{FD}^{\infty} = 1$ on the fast-dimer edge. In Figure \ref{fig:fsfdcomparison}, we illustrate the average group reproduction rates (left) and most abundant composition of slow genes (right) in steady state in the limit that $\lambda \to \infty$ for multilevel competition on both the slow-dimer and slow-fast edges of the simplex.  For the comparison of peak composition of slow genes, we are plotting the value $\hat{x}^{\infty}_{FD} := \frac{\hat{z}^{\infty}_{FD}}{2}$ because we have assumed that the slow-fast dimers count as half of a slow gene and half of a fast gene from the perspective of the group reproductive function. We see that both the average payoff and the number of slow genes at steady state are greater on the fast-dimer edge than on the slow-dimer edge for large $\lambda$ when $\eta > \frac{2}{3}$, as seen to the right of the second vertical dashed line. 

\begin{figure}[htp]
    \centering
      \includegraphics[width=0.48\textwidth]{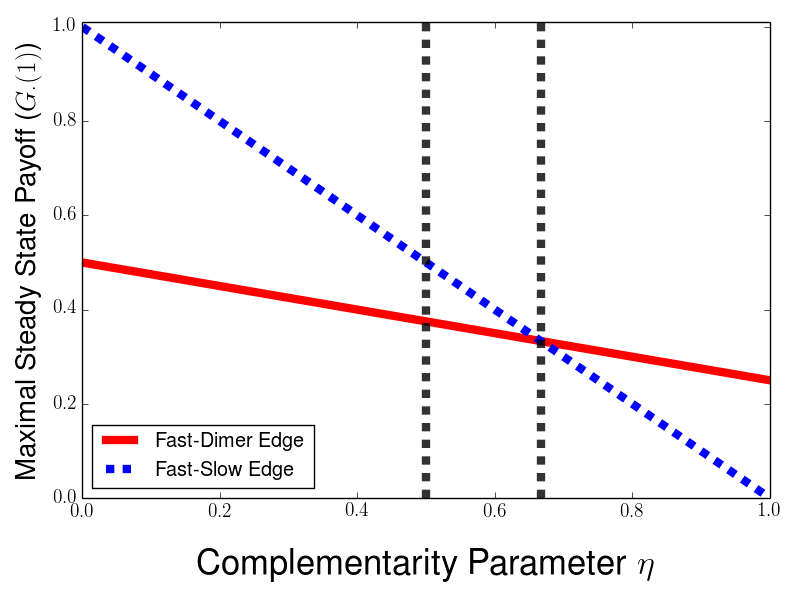}
    \includegraphics[width=0.48\textwidth]{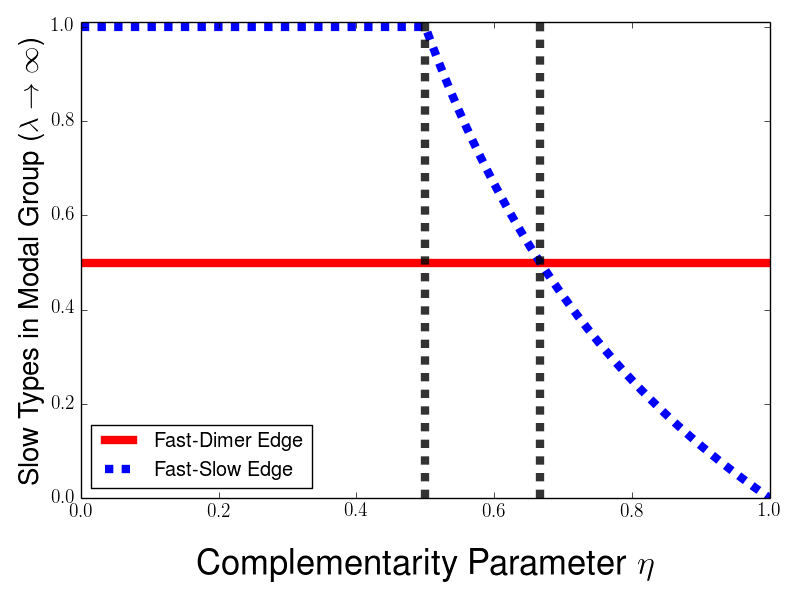}
    \caption[Comparison of slow gene composition in most abundant group and average cell reproductive rate at steady state for competition on fast-dimer and fast-slow edges for large $\lambda$.]{Comparison of average cell reproductive rate (left) and slow gene composition in most abundant group (right) at steady state for competition on fast-dimer and fast-slow edges for large $\lambda$. Red solid lines and blue dashed lines correspond to behavior on the fast-dimer edge and fast-slow edge, respectively. Left and right vertical dashed lines depict $\eta = \frac{1}{2}$ and $\eta = \frac{2}{3}$.  }
    \label{fig:fsfdcomparison}
\end{figure}

Specializing to the case of the family of gene-level birth rate from Equation \eqref{eq:trimorphicspecialbirthrates}, we can extend this comparison between the relative benefits and costs of dimerization by exploring how complementarity parameters $\eta$ and gene-level advantage of fast replicators $s$ impact the relative collective outcomes on the fast-slow and fast-dimer edges of the simplex. In particular, we would like to characterize the parameter regimes of $\eta$ and $s$ for which dimerization increases or decreases the threshold between-protocell selection strength needed to promote coexistence and for which dimerization produces a higher maximum possible collective fitness under pairwise multilevel competition.  In Figure \ref{fig:slowdimerregimes}, we illustrate the three regimes that are possible when one moves the dynamics from the fast-slow edge of the simplex to the fast-dimer edge of the simplex. From left to right, we see that the three regions of parameter space correspond to pairs of tradeoff parameter $\eta$ and individual-level advantage for fast replicators $s$ for which both dimerization hurts both the threshold relative selection strength and the maximal possible payoff (plotted in yellow), dimerization helps maximal possible payoff but hurts threshold selection strength (plotted in orange), or helps both threshold selection strength and maximal possible payoff (plotted in red). For the case in which slow and fast genes are perfect complements for between-protocell competition (when $\eta = 1$), any amount of gene-level advantage $s > 0$ for fast replicators will result in a better outcome for dimerization. For $\eta \in (\frac{2}{3},1)$ there exist sufficiently weak gene-level advantages $s$ for which fast-dimer competition that confer a collective disadvantage relative to fast-slow competition for low relative selection strengths $\lambda$, but confer a collective benefit for sufficiently high values of $\lambda$. 

\begin{figure}[ht]
    \centering
    \includegraphics[width = 0.75\textwidth]{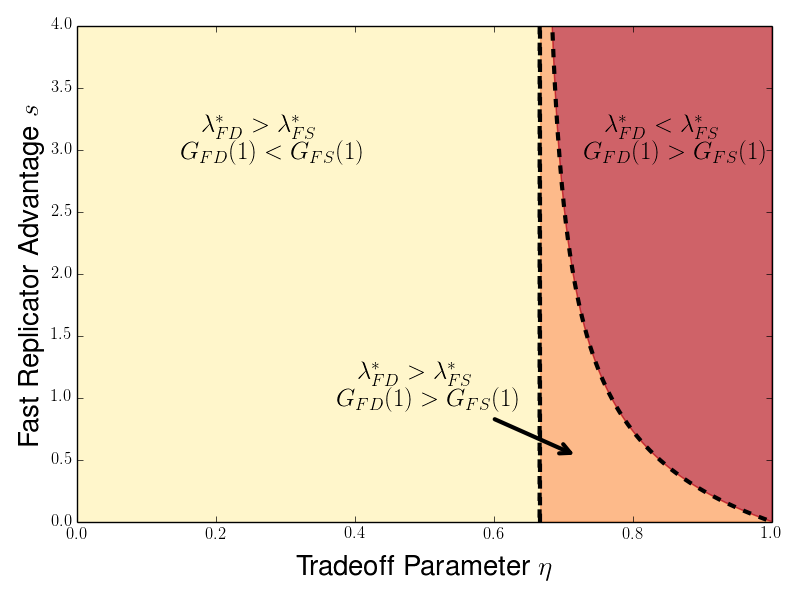}
    \caption[Illustration of possible impacts of replacing slow replicators with dimers.]{Illustration of the three possible impacts of replacing slow replicators with dimers as we vary $\eta$ and $s$. In leftmost region (light yellow), dimerization hurts the threshold needed to achieve coexistence and hurts maximal possible steady-state average cell reproduction rate. In the middle region (orange), dimerization increases threshold $\lambda$ needed to sustain coexistence, but helps the cell reproduction rate for large $\lambda$. In right region (red), dimerization helps both by lowering the threshold between-cell selection strength need to achieve coexistence and increasing the possible cell reproduction rate at steady state for strong between-cell selection. Boundary between first two regions given by $\eta = \frac{2}{3}$ and boundary between second and third region corresponds to $\eta_c^s$ given by Equation \eqref{eq:etacs}.}
    \label{fig:slowdimerregimes}
\end{figure}

\begin{remark}
This comparison we have made between the average protocell-level fitness achieved by competition on the fast-slow and fast-dimer edges of the simplex can be placed in a dynamical setting using a framework for nested birth-death models multilevel selection with multiple dynamics of groups \cite{cooney2021long}. In particular, we can think of a population consisting of protocells with compositions consisting of either a mix of fast and slow replicators or a mix of fast and dimer replicators, where gene-level competition follows the rules introduced in Section \ref{sec:trimorphicformulation}. To consider competition between fast-slow and fast-dimer protocells, we can model between-protocell competition by a process in which protocells on the fast-slow edge and fast-dimer edge respectively replicate at rates $G_{FS}(x)$ and $G_{FD}(z)$, with the offspring protocell replacing a random protocell coming from either the fast-dimer or fast-slow populations. Introducing the non-negative densities $f(t,x)$ and $g(t,z)$ which describe the distribution of protocells on the fast-slow and fast-dimer edges of the simplex, we can describe the evolution of these two densities under our nested birth-death process using the following system of PDEs
\begin{subequations}
     \begin{align}
    \dsdel{f(t,x)}{t} &= \dsdel{}{x} \left[\left(b_F - b_S\right)x \left(1-x\right) f(t,x) \right] \nonumber \\  &+ \lambda f(t,x) \left[ G_{FS}(x) - \int_0^1 G_{FS}(y) f(t,y) dy - \int_0^1 G_{FD}(w) g(t,w) dw \right]     \\ 
     \dsdel{g(t,z)}{t} &= \dsdel{}{z} \left[\left(b_F - b_D\right)z \left(1-z\right) g(t,z) \right]  \nonumber \\  &+ \lambda g(t,z) \left[ G_{FD}(z) - \int_0^1 G_{FS}(y) f(t,y) dy - \int_0^1 G_{FD}(w) g(t,w) dw \right].
     \end{align}
\end{subequations}
Starting with non-negative initial densities $f_0(x)$ and $g_0(z)$ satisfying $\int_0^1 f_0(x) dx + \int_0^1 g_0(z) dz = 1$, it can be shown that the long-time behavior of the population will concentrate entirely on either the fast-slow edge ($\int_0^1 g(t,w) dw \to 0$ as $t \to \infty$) or the fast-dimer edge ($\int_0^1 f(t,y) dy \to 0$ as $t \to \infty$) \cite{cooney2021long}. The edge upon which the population concentrates is the one which would produce a higher average protocell-level fitness at steady state under a multilevel competition on the two edges alone \cite{cooney2021long}.
\end{remark}

\subsection{Dynamics on Slow-Dimer Edge of the Simplex}\label{sec:slowdimer}

In this section, we consider the multilevel dynamics on the slow-dimer edge of the simplex, exploring how removing the fast replicator and its corresponding gene-level advantage can help to facilitate coexistence of the fast and slow genes. We see that when protocell-level competition favors all-slow compositions to all-dimer compositions, multilevel selection will promote concentration upon all-slow protocells. When all-dimer protocells have a collective advantage over all-slow protocells, we can show that sufficient levels of between-protocell competition can result in steady state coexistence of slow and dimer replicators. In addition, we see that a version of the shadow of lower-level selection holds on this edge of the simplex, as no level of between-protocell competition can allow for optimal collective fitness for complementarity scenarios in which a mix of slow replicators and dimers is most favored under protocell-level replication. 

On the slow-dimer edge of the simplex, we will describe the composition of a protocell by its fraction $z$ of dimer replicators. Restricting our trimorphic protocell-level replication rate $G(x,y,z)$ to the slow-dimer edge by plugging $x = 0$ and $y = 1-z$ into Equation \eqref{eq:Gxyzeta}, we see that the collective reproduction rate for protocells featuring only slow and dimer replicators $G_{SD}(z)$ is given by 
\begin{equation} \label{eq:groupslowdimer}
G_{SD}(z) := G(1-z,0,z) = 1 - \eta + \left( \eta - \frac{1}{2}\right)z - \frac{\eta}{4} z^2.
\end{equation}
Noting that $G'_{SD}(z) = \eta - \frac{1}{2} - \frac{\eta z }{2}$, we see that $G_{SD}(z)$ is a decreasing function of $z$ when $\eta \leq \frac{1}{2}$, so the protocell-level reproduction rate is maximized by the all-slow composition. When $\eta > \frac{1}{2}$, $G_{SD}$ has a unique maximizer featuring a mix of slow and dimer replicators. Across the possible complementarity parameters $\eta \in [0,1]$, we see that the fraction of dimers $z^*_{SD}$ that maximizes the protocell-level reproduction function $G_{SD}(z)$ is given by
 \begin{dmath} \label{eq:zstarSD}
   z^*_{SD}(\eta) = \left\{
     \begin{array}{cl}
       0 & : \eta \leq \frac{1}{2}\\
       2 - \frac{1}{\eta} & : \frac{1}{2} \leq \eta \leq 1
     \end{array}
   \right. .
\end{dmath} 
The all-slow composition is optimal for a protocell for any values of $\eta$ at which the all-slow protocell is most favored in the original slow-fast protocell model. The full-dimer composition is optimal for $\eta = 1$, which is the case in which an equal fraction of slow and fast genes is most favored for protocell-level reproduction. When $\frac{1}{2} < \eta < 1$, the optimal protocell composition on the slow-dimer edge features an interior composition $z$ with both slow and dimer replicators.  

We can also examine the collective replication rate $G_{SD}(z)$ to explore the values of $\eta$ for which the dynamics of Equation \eqref{eq:slowdimermultilevelPDE} satisfy the assumptions of Theorem \ref{thm:PDconvergencetosteady}, \ref{thm:PDconvergencetodelta}, or Proposition \ref{prop:PDelconvergencetodelta}.  By rewriting Equation \eqref{eq:groupslowdimer} in the following form
\[G_{SD}(z) = 1 - \eta + z \left[\frac{\eta}{4} \left( 1 - z \right) + \left( \frac{3 \eta}{2} - 1 \right) \right], \] 
we see that $G_{SD}(0) = 1 - \eta$ is the collective minimum when $\eta > \frac{2}{3}$, so the assumptions needed for the equality case of Theorem \ref{thm:PDconvergencetodelta} hold in this regime.
Furthermore, we see that the relative ranking of the protocell-level replication rates of the all-slow and all-dimer equilibrium can change depending on $\eta$. In particular, noting that $G_{SD}(1) = \frac{1}{2} \left(1 - \frac{\eta}{2} \right)$ and $G_{SD}(0) = 1 - \eta$, we see that $G_{SD}(1) < G_{SD}(0)$ for $\eta < \frac{2}{3}$. In that case, the all-slow composition is favored over the all-dimer composition under both gene-level and protocell-level competition. When $\eta > \frac{2}{3}$, $G_{SD}(1) > G_{SD}(1)$ and the evolutionary tension between gene-level competition favoring more slow replicators and protocell-level competition favoring all-dimer protocells over all-slow protocells is more similar to the scenarios observed on the fast-slow and fast-dimer edges. In Figure \ref{fig:Gzsdfunction.png}, we illustrate $G_{SD}(z)$ for various of $\eta$, illustrating the different cases of edge and interior optimal fractions of slow and dimer replicators and indicating the regions in which either the all-dimer or all-slow composition is favored under protocell-level competition.

\begin{figure}[htp!]
    \centering
    \includegraphics[width = 0.6\textwidth]{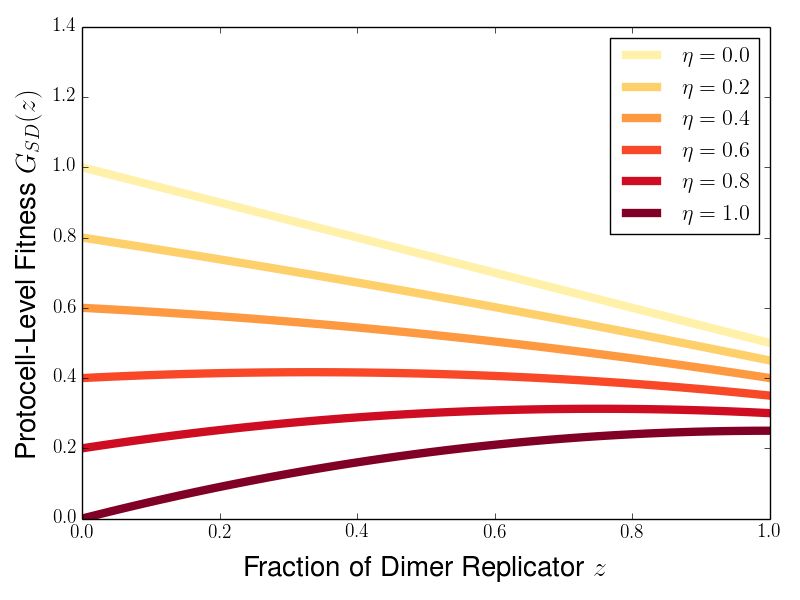}
    \caption{Protocell-level reproduction rates $G_{SD}(z)$ on slow-dimer edge of the simplex for various levels of the complementarity parameter $\eta$. For $\eta < \frac{1}{2}$, the group reproduction rate $G_{SD}(z)$ is a decreasing function of $z$ and protocell-level reproduction is maximized by all-slow protocells. When $\eta \in (\frac{1}{2},1)$, protocell-level reproduction is maximized by an intermediate mix of slow and dimer replicators, while, for $\eta = 1$, protocell-level reproduction is maximized by all-dimer protocells. For $\eta \in (\frac{1}{2},\frac{2}{3})$, $G_{SD}(0) > G_{SD}(1)$, so the all-slow composition has both a gene-level and protocell-level advantage over the all-dimer composition in this regime.}
    \label{fig:Gzsdfunction.png}
\end{figure}

To obtain a differential equation for the gene-level dynamics on the slow-dimer edge of the simplex, we can apply the restriction $y = 0$ and $x = 1-z$ to Equation \eqref{eq:withindimer}. This allows us to see that, in protocells featuring only slow and dimer replicators, the fraction of dimers evolves according to the follow gene-level replicator equation
\begin{equation} \label{eq:SDwithin}
\dsddt{z(t)} = z \left[b_D - \left(b_S x + b_F y + b_D z\right) \right] \bigg|_{\substack{yn = 0 \\ x = 1-z}} = 
- \left( b_{S} - b_{D} \right) z \left( 1 - z \right).
\end{equation}
Notably, this ODE is of the form of the characteristic curves given by Equation \eqref{eq:characteristicsgeneral} with the net gene-level replication function $\pi_{FD}(z) = b_S - b_D > 0$. 
This gene-level competition on the slow-dimer edge of the simplex always pushes to increase the fraction of slow replicators.

Now that we have characterized the gene-level dynamics and the protocell-level reproduction function for populations on the slow-dimer edge, we can introduce a density $h(t,z)$ describing the distribution of protocell compositions when the population is confined on this edge. The reduced dynamics of multilevel selection on the slow-dimer edge evolve according to the PDE
\begin{dmath} \label{eq:slowdimermultilevelPDE}
\dsdel{h(t,z)}{t} = \dsdel{}{z}\left[\left(b_S - b_D\right) z \left( 1 - z \right) h(t,z) \right] + \lambda h(t,z) \left[G_{SD}(z) - \int_0^1 G_{SD}(w) h(t,w) dw \right].
\end{dmath}

Using the same approach as in Sections \ref{sec:protocellongtime} and \ref{sec:fastdimer}, we can find that there is a family of steady state solutions $h^{\lambda}_{\theta}(z)$ to Equation \eqref{eq:slowdimermultilevelPDE} that are given by probability densities of the form
\begin{subequations} \label{eq:slowdimersteadynormalized}
     \begin{align}
     h^{\lambda}_{\theta}(z) &= Z_h^{-1} \: z^{\mathlarger{\left[\nicefrac{1}{2 \left(b_S - b_D\right)}\right] \lambda \left[ \frac{3 \eta}{2} - 1\right] - \theta - 1}} \left(1 - z\right)^{\mathlarger{\theta - 1}} \exp\left( - \frac{\lambda \eta z}{4 \left( b_S - b_D\right)} \right) \\
     Z_h &= \int_0^1 w^{\mathlarger{\left[\nicefrac{1}{2 \left(b_S - b_D\right)}\right] \lambda \left[ \frac{3 \eta}{2} - 1\right] - \theta - 1}} \left(1 - w\right)^{\mathlarger{\theta - 1}} \exp\left( - \frac{\lambda \eta w}{4 \left( b_S - b_D\right)} \right) dw.
     \end{align}
\end{subequations}

We note that densities of this form can only be integrable near $z = 0$ if $\eta > \frac{2}{3}$. If this condition is satisfied, the density will be integrable provided that the relative intensity of between-protocell competition $\lambda$ exceeds the following threshold value

\begin{equation} \label{eq:lambdastarSD}
\lambda^*_{SD} = \frac{(b_S - b_D) \theta}{G_{SD}(1) - G_{SD}(0)} = \frac{4(b_S - b_D) \theta}{3 \eta - 2}.
\end{equation}
Notably, we find that the threshold $\lambda^*_{SD} \to \infty$ as $\eta \to \frac{2}{3}$.

Recalling the form of the condition on $\lambda$ from Theorem \ref{thm:PDconvergencetosteady}, we see convergence to such a steady state density could only be possible on the slow-dimer edge if the following inequality in satisfied
 \[\lambda \left( G_{SD}(1) - G_{SD}(0) \right) > \underbrace{\left(b_S - b_D\right) \theta}_{> 0}. \]
 This can never be satisfied for any positive $\lambda$ if $G_{SD}(1) - G_{SD}(0) \leq 0$, which holds when $\eta \leq \frac{2}{3}$. Therefore no coexistence at a density steady state is expected when the complementarity parameter $\eta \in [0,\frac{2}{3}]$, or when the protocell-level replication rate is maximized by a composition of slow genes between 75 percent and 100 percent. This means that, on the slow-dimer edge of the simplex, it is only possible to obtain long-time coexistence of slow and dimer replicators if the group level reproduction function most favors compositions with between 50 and 75 percent slow genes, corresponding to the case in which the all-dimer protocell is closer to the collective optimum than the all-slow protocell.  
 
In Proposition \ref{prop:longtimeslowdimer}, we collect our results of the long-time behavior of the multilevel slow-dimer dynamics of Equation \eqref{eq:slowdimermultilevelPDE} given an initial condition with H{\"o}lder exponent $\theta > 0$ near $z = 1$. When $\lambda \left[3 \eta - 2 \right] > 4 \left(b_S - b_D\right) \theta$ (which occurs when $\eta > \tfrac{2}{3}$ and $\lambda > \lambda^*_{SD}$), we see from Theorem \ref{thm:PDconvergencetosteady} that the population of protocells will converge to a density steady state supported a coexistence of slow and dimer replicators. When $\lambda \left[3 \eta - 2 \right] \leq 4 \left(b_S - b_D\right) \theta$, the population concentrates upon a delta-function at the all fast composition. This result on convergence to the all-slow state follows from Theorem \ref{thm:PDconvergencetodelta} when $\eta > \frac{2}{3}$ (so $G_{SD}(1) > G_{SD}(0)$) and $\lambda \leq \lambda^*_{SD}$, and this convergence is confirmed by Proposition \ref{prop:PDelconvergencetodelta} when $\eta \leq \frac{2}{3}$ (and correspondingly $G_{SD}(1) \leq G_{SD}(0)$) for any positive level of between-protocell competition $\lambda > 0$.  

\begin{proposition}  \label{prop:longtimeslowdimer} Suppose the population of protocells composed of slow and dimer replicators has initial measure $\mu_0(dz)$ with H{\"o}lder exponent of $\theta$ near $z=1$. Then, in the limit as $t \to \infty$, the solution $\mu_t(dz)$ to Equation \eqref{eq:slowdimermultilevelPDE} will display the following behavior \begin{displaymath} \mu_t(dz)  \rightharpoonup \left\{
     \begin{array}{cr}
       \delta(z) & : \lambda \left[3 \eta - 2 \right] \leq 4 \left(b_S - b_D\right) \theta \\ 
       h^{\lambda}_{\theta}(z) & : \lambda \left[3 \eta - 2 \right] > 4 \left(b_S - b_D\right) \theta
     \end{array} \right.,
     \end{displaymath} 
where the steady state density $h^{\lambda}_{\theta}(z)$ is given by Equation \eqref{eq:slowdimersteadynormalized}.
\end{proposition}

We now combine the results of Proposition \ref{prop:longtimeslowdimer} with the expression for average protocell-level fitness from Equation \eqref{eq:Gaveragelambda} to see that, for a solution $\mu_t(dz)$ to Equation \eqref{eq:slowdimermultilevelPDE}, the collective outcome satisfies
\begin{equation} \label{eq:sdcollectivelongtime}
    \lim_{t \to \infty} \langle G_{SD}(\cdot) \rangle_{\mu_t} = 
    \left\{ 
    \begin{array}{cr}
    G_{SD}(0) &:  \lambda \left[3 \eta - 2 \right] \leq 4 \left(b_S - b_D\right) \theta  \\
   G_{SD}(1) - \frac{\left(b_S - b_D \right) \theta}{\lambda} &: \lambda \left[3 \eta - 2 \right] > 4 \left(b_S - b_D\right) \theta
    \end{array}
    \right\} \leq \max\left(G_{SD}(0), G_{SD}(1) \right).
\end{equation}

Therefore we see that the collective outcome at steady state on the slow-dimer edge is limited by the larger of the protocell-level reproduction rates of the all-slow and all-dimer protocell, even when intermediate fractions of slow and dimer replicators is optimal for collective reproduction among compositions on this edge of the simplex for $\eta \in (\frac{1}{2},1)$. Even on the slow-dimer edge, we see that the absence of fast replicators is not enough to always promote optimal compositions in the limit of infinite between-protocell competition, as the presence of the gene-level advantage of slow replicators over dimers can allow for a shadow of lower-level selection in the multilevel dynamics on this edge. Furthermore, the bound from Equation \eqref{eq:sdcollectivelongtime} on the collective outcome also holds for the fast-slow and fast-dimer edges of the simplex, which further motivates studying whether such a shadow of gene-level selection also arises when we consider dynamics on the full simplex and allow protocells featuring coexistence of fast, slow, and dimer replicators.

To further understand the shadow cast by gene-level selection on the fast-dimer edge of the simplex, we can study the modal protocell gene composition for the family of steady states $h^{\lambda}_{\theta}(z)$. In Proposition \ref{prop:SDmodal}, we characterize the most abundant cell composition at steady state $\hat{z}_{\tilde{\lambda}} := \argmax_{x \in [0,1]} h^{\tilde{\lambda}}_{\theta}(z)$. We see in the limit as $\lambda \to \infty$ that, whenever a mix of dimers and slow replicators is maximizes protocell-level replication rate (which occurs when $\eta \in (\tfrac{1}{2},1)$), the modal outcome features fewer dimers than the collectively optimal composition. When slow and fast genes are perfect complements for protocell-level reproduction (in the case of $\eta = 1$), the collective optimum and the modal outcome achieved under infinite between-protocell competition coincide upon the all-dimer composition.

\begin{proposition}[\bfseries Most Abundant Composition at Steady State Features Fewer Dimers than Optimal Unless Protocell-Level Competition Most Favors Fifty-Fifty Mix of Fast and Slow Replicators] \label{prop:SDmodal}
\sloppy{Consider the steady state density $h^{\lambda}_{\theta}(z)$ and suppose that $\eta > \frac{2}{3}$, $\lambda \left(\frac{3 \eta}{2} - 1\right) > 2 \left(b_S - b_D \right) \left(\theta + 1\right)$ and $\theta \geq 1$. Then the most abundant composition at steady state $\hat{z}_{SD}^{\lambda} := \argmax_{z \in [0,1]} h^{\lambda}_{\theta}(z)$ is given by }
\begin{dmath} \label{eq:modalSDlambda}
    \hat{z}^{\lambda}_{SD} = \frac{1}{\lambda \eta} \left(\lambda \left(2 \eta - 1\right) - 4(b_S - b_D)  - \sqrt{\left[\lambda \left( 2 \eta - 1 \right) - 4 (b_S - b_D) \right]^2 - 2 \lambda \eta \left[ \lambda \left( \frac{3 \eta}{2} - 1 \right) - 2 \left(b_S - b_D \right) \left( \theta + 1 \right) \right]} \right).
\end{dmath}
Furthermore, in the limit of infinite intensity of between-protocell composition, the modal composition $\hat{z}^{\infty}_{FD} := \lim_{\lambda \to \infty} \hat{z}_{FD}^{\lambda}$ is given by
\begin{equation} \label{eq:hatzSDinf}
    \hat{z}^{\infty}_{SD} = 3 - \frac{2}{\eta}
\end{equation}
Finally, noting from Equation \eqref{eq:zstarSD} that $z^*_{SD} = 2 - \frac{1}{\eta}$ when $\eta > \frac{2}{3}$, we see that $\hat{z}^{\infty}_{SD} = z^*_{SD} = 1$ when $\eta = 1$ and that, for $\eta \in (\frac{2}{3},1)$,
\begin{equation}
\hat{z}_{\infty} = 3 - \frac{2}{\eta} = \left( 2 - \frac{1}{\eta}\right) + \left( 1 - \frac{1}{\eta} \right) = z^*_{SD} + \underbrace{\left( 1 - \frac{1}{\eta} \right)}_{< 0} <  z^*_{SD}.
\end{equation}

\end{proposition}

In Figure \ref{fig:slowdimerpeakplot}, we illustrate the difference between the protocell composition $z^*_{SD}$ with maximal collective reproduction rate and the composition $\hat{z}^{\infty}_{SD}$ achieving maximal abundance at steady state for large relative selection strength $\lambda \to \infty$. We see that for $\eta < \frac{1}{2}$, both the optimal protocell composition and modal steady state composition agree upon the all-dimer protocell. For $\eta \in (\frac{1}{2},1)$, we see that the maximum possible level of dimers achieved at steady state is less than what is optimal for the protocell, while the modal composition as $\lambda \to \infty$ coincides with the collective optimum when $\eta = 1$ (and correspondingly the all-dimer composition is optimal for protocell-level competition). In particular, we see that for $\eta \in (\frac{1}{2},\frac{2}{3})$, the compositions with maximal protocell-level replication rate feature dimers, while no dimers are ever actually achieved at steady state. We see that the gap between the optimal composition $z^*_{SD}$ and the modal composition $\hat{z}^{\infty}_{SD}$ for large $\lambda$ is maximized at $\eta = \frac{2}{3}$, when a fifty-fifty mix of slow and dimer replicators is most favored until between-protocell competition. This also be understood by using Equations \eqref{eq:zstarSD} and \eqref{eq:hatzSDinf} to see that, for $\eta  \geq \frac{2}{3}$, the gap is given by $z^*_{SD} - \hat{z}^{\infty}_{SD} = \frac{1}{\eta} - 1$ and is a decreasing function of $\eta$. 

\begin{figure}[ht!]
    \centering
    \includegraphics[width = 0.7\textwidth]{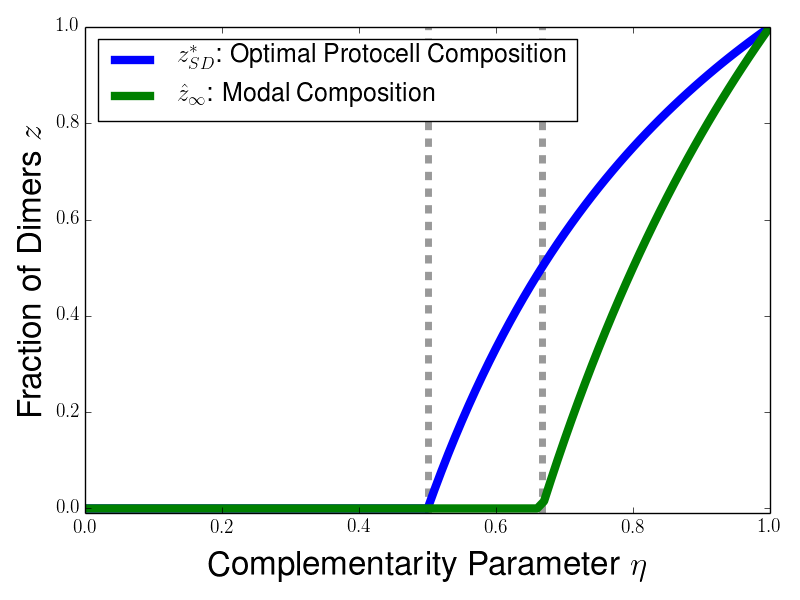}
    \caption{Comparison between protocell composition maximizing $G_{SD}(z)$ and protocell properties at steady state as $\lambda \to \infty$ for dynamics on the slow-dimer edge, plotted as a function of the complementarity parameter $\eta$. Blue line corresponds to group type with maximal group reproduction rate, while green line describes fraction of dimers achieved at steady state as $\lambda \to \infty$. Left dashed vertical line denotes $\eta = \frac{1}{2}$, the threshold above which dimers are present in  the protocell composition with maximal reproductive rate, and right dashed line corresponds to $\eta = \frac{2}{3}$, the threshold needed for dimers to be realized as steady state for any level of $\lambda$.}
    \label{fig:slowdimerpeakplot}
\end{figure}

While we saw in Section \ref{sec:fastdimer} that introducing dimers results in a pairwise fast-dimer multilevel competition in which sufficiently strong between-protocell competition achieves the optimal all-dimer outcome, we see from the dynamics on the slow-dimer edge of the simplex that multilevel protocell competition appears to still fail to achieve optimal collective outcomes unless $\eta = 1$. %
This failure of dimers to achieve optimal abundance on the slow-dimer edge motivates explorations into mechanisms that can further help to improve upon the maximal collective fitness achievable via multilevel selection.

\section{Numerical Approach for Three-Type Dynamics} \label{sec:trimorphicnumerics}

So far, we have only considered the long-time behavior of PDE models for multilevel selection with two types of individuals, focusing on the pairwise competition on the various edges of the fast-slow-dimer simplex. In this section, we look to obtain a preliminary understanding of the multilevel competition that takes places when we allow for protocells featuring a mix of fast, slow, and dimer replicators. As a first attempt to understand these trimorphic dynamics, we will make use of a finite volume schemes to study numerical solutions to an approximation of Equation \eqref{eq:multileveltrimorphic}. The details of this finite volume approach are presented in Section \ref{sec:fvtrimorphic}, and the main equation for our discretized dynamics is given by Equation \eqref{eq:trimorphicfinitevolume}. These numerical simulations serve as an initial suggestion of the ability of the population of protocells to maintain coexistence of fast and slow genes through multilevel competition featuring fast, slow, and dimer replicators, even when no coexistence is possible in pairwise competition between slow and fast replicators (when $\eta = 1$).

In this section, we will first explore how the density of fast, slow, and dimer replicators (Figure \ref{fig:trimorphicnumericssp4}) and the average protocell-level fitness (Figure \ref{fig:Gintime}) evolve over time for our finite volume numerical solutions. These numerical trajectories provide an initial suggestion that our trimorphic multilevel dynamics can support coexistence of fast and slow genes, even in the case in the two genes are perfect complements for protocell-level reproduction ($\eta = 1$) and coexistence cannot occur in pairwise multilevel competition on the fast-slow edge of the simplex. Then, we analyze numerical solutions achieved after a large number of time steps, studying how these long-time densities vary with changes in the relative strength $\lambda$ of between-protocell competition and the degree of complementarity $\eta$ of slow and fast genes for protocell-level replication (Figure \ref{fig:steadystatetrimorphicplots}). Further comparisons are provided between the average protocell-level fitness (Figures \ref{fig:Gplottrimorphic} and \ref{fig:Getatrimporphic}) and the modal and mean compositions of fast genes (\ref{fig:peakandmeantrimorphic}) between the long-time trimorphic numerical solutions and the analytical steady states of the fast-slow and fast-dimer edges of the simplex. The agreement found in these figures suggest the possibility of a tug-of-war between the gene-level and protocell-level reproductive advantages taking place between corners of the simplex similar to what is observed in the dimorphic case through the threshold condition on relative selection strength of Equation \eqref{eq:lambdastargeneral} and the average protocell-level fitness at steady state \eqref{eq:Gaveragethresh}. 

As a baseline scenario for studying the dynamics of our numerical scheme, we will consider a uniform initial distribution of protocell compositions on the fast-slow-dimer simplex. This choice is motivated by the agreement found between the long-time behavior of finite-volume numerical solutions for two-type multilevel dynamics and the analytical solutions for the long-time steady states achieved by solutions to the multilevel PDE of Equation \eqref{eq:generalPDEreplicator} given uniform initial densities \cite{cooney2020pde}. This motivation will be presented further in Section \ref{sec:fvdimorphic} in the context of multilevel dynamics of the fast-slow and fast-dimer edges of the simplex. An interesting question for future work is how the initial configuration of the population of protocells impacts the long-time support for fast and slow genes, and to understand whether there is a fine property of the initial distribution analogous to the H{\"o}lder exponent near the all-slow equilibrium that can be used to characterize the possible long-time behaviors of the trimorphic multilevel dynamics.

In Figure \ref{fig:trimorphicnumericssp4}, we display snapshots at various points in time of the approximate solutions for $\rho(t,x,y)$ as a heatmap on the simplex of possible cell compositions $x + y + z  \leq  1$, with the individual advantage of fast replicators given by $s = 0.4$ {Figure \ref{fig:trimorphicnumericssp4}).  We consider a large value of between-cell competition $\lambda = 300$ and a cell-level tradeoff parameter $\eta = 1$ corresponding to the slow and fast genes serving as perfect complements for between-cell competition. From the snapshots, we see that the density of cell types first concentrates towards cell compositions display close to fifty-fifty mixes of slow and fast genes (as seen in the top-left panels), ranging on the line from the all-dimer cells in the bottom-left to the cell composition with half slow replicators and half fast replicators on the diagonal of the simplex satisfying $1-x-y$. Then we see that the fast and slow replicators tend to take over due to within-cell replication, so the groups closer to the diagonal next increase in frequency relative to the compositions featuring many dimers (as seen in the top-right panels). Next the fast replicators begin to beat out the slow replicators within cells, and so the most frequent cell types tend to feature more fast replicators to slow replicators (as seen in the middle panels of Figure \ref{fig:trimorphicnumericssp4}). Finally, we see that the remaining many-dimer groups are now more competitice under protocell-level reproduction than the groups with a majority of fast replicators, and then the balance between within-cell and between-cell competition results in steady state densities featuring a coexistence of slow and fast monomer replicators with the slow-fast dimers.

\begin{figure}[ht]
    \centering
    \includegraphics[width = 0.48 \textwidth,height = 0.38\textwidth]{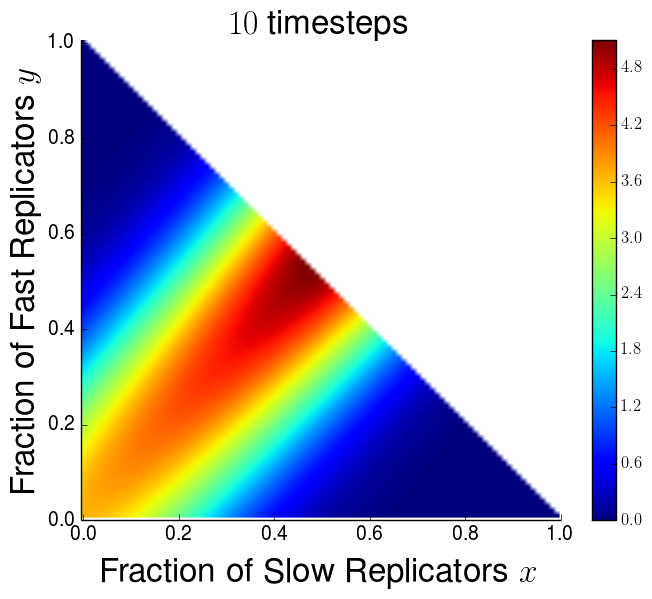}
     \includegraphics[width = 0.48 \textwidth,height = 0.38\textwidth]{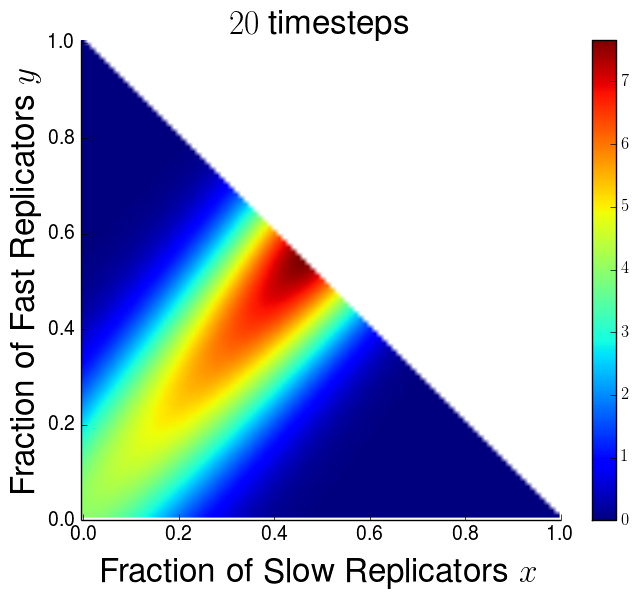}
      \includegraphics[width = 0.48 \textwidth,height = 0.38\textwidth]{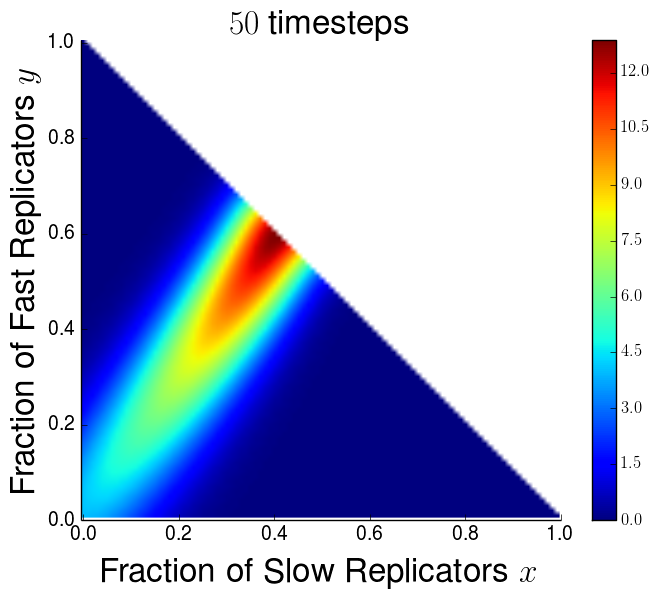}
       \includegraphics[width = 0.48 \textwidth,height = 0.38\textwidth]{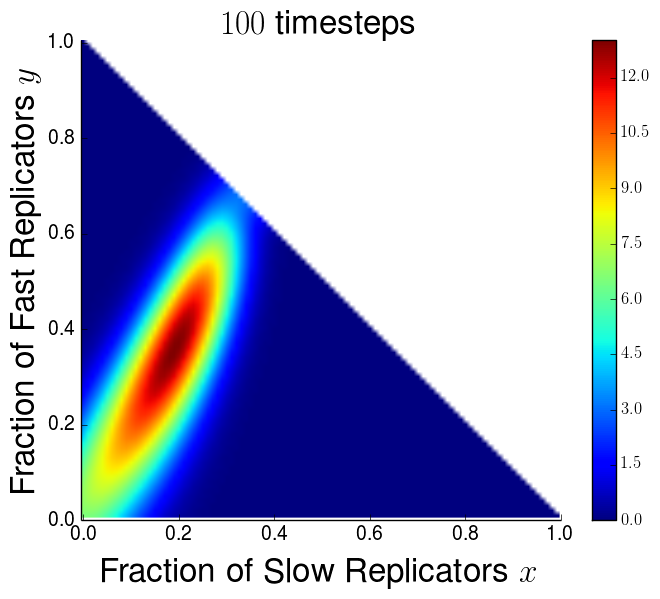}
        \includegraphics[width = 0.48 \textwidth,height = 0.38\textwidth]{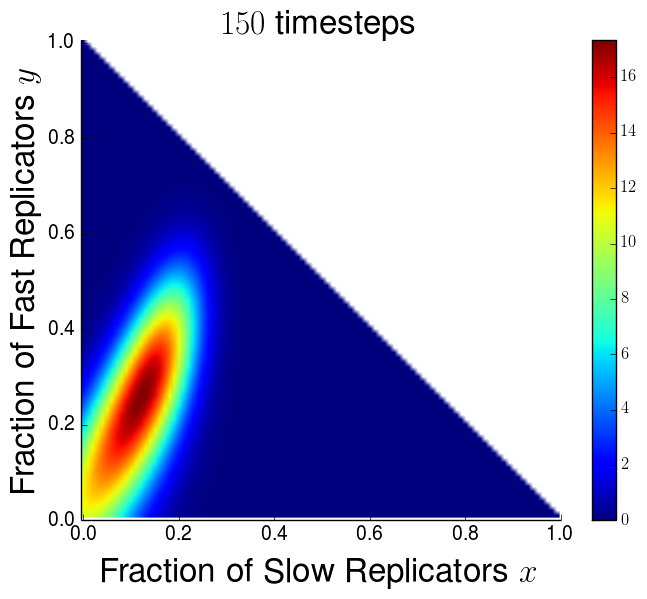}
         \includegraphics[width = 0.48 \textwidth,height = 0.38\textwidth]{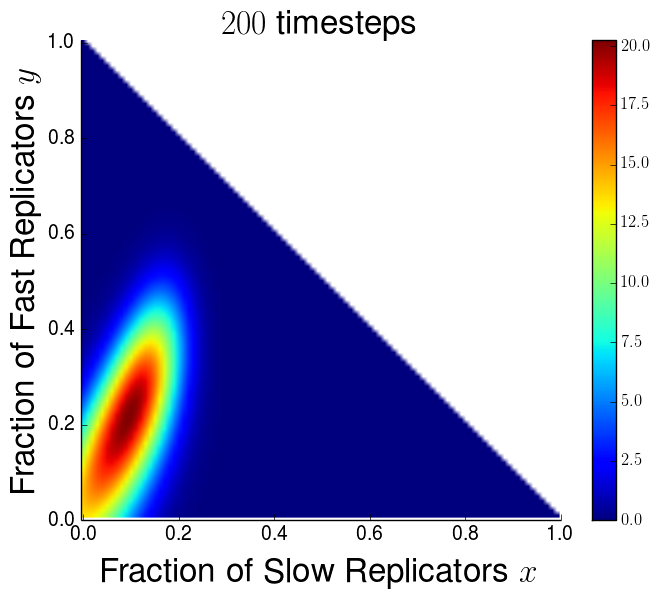}
    \caption[Numerical solutions for the finite volume approximation of the multilevel dynamics of slow-fast-dimer competition for $\eta = 1$ and $s = 0.4$.]{Numerical solutions for the finite volume approximation of the multilevel dynamics of slow-fast-dimer competition with relative selection strength $\lambda = 300$, fast replicator advantage $s = 0.4$ and tradeoff parameter $\eta = 1$. Each panel corresponds the the approximate density of $\rho(t,x,y)$ after 10 (top-left), 20 (top-right), 50 (middle-left), 100 (middle-right), 150 (bottom-left), and 200 (bottom-right) timsteps with a time increment $\Delta t = 0.03$. 
    }
    \label{fig:trimorphicnumericssp4}
\end{figure}

\clearpage

After seeing how the densities evolve in time for a given initial condition, we can also look to understand how quantities like the average protocell-level fitness of the population evolve over time. In Figure \ref{fig:Gintime}, we show time trajectories of the collective fitness $\langle G(x,y) \rangle_{\rho(t,x,y)}$ under the trimorphic multilevel dynamics for complementarities given by $\eta = 1$ (Figure \ref{fig:Gintime}, left) and  $\eta = 0.7$ (Figure \ref{fig:Gintime}, right) and for various values of the between-protocell selection strength $\lambda$. We see that the trajectories of the protocell-level fitness is nonmonotonic in time, and that average protocell-level reproduction rate appears to be an increasing function of $\lambda$. The values of collective reproduction rate $\langle G(x,y) \rangle_{\rho(t,x,y)}$ appear to converge towards fixed values after around several thousand time steps for time difference $\Delta t = 0.0015$ seconds (around $1000$ time steps for $\eta = 1$ and around $3500$ time steps for $\eta = 0.7$). This equilibriation towards a fixed protocell-level fitness suggests the possibility that the densities $\rho(t,x,y)$ solving the trimorphic PDE of Equation \eqref{eq:multileveltrimorphic} with an initial uniform distribution may converge to steady-state densities in a manner reminiscent of the dimorphic multilevel dynamics featured in Section \ref{sec:existingresults}. 

\begin{figure}[ht]
    \centering
    \includegraphics[width = 0.48\textwidth]{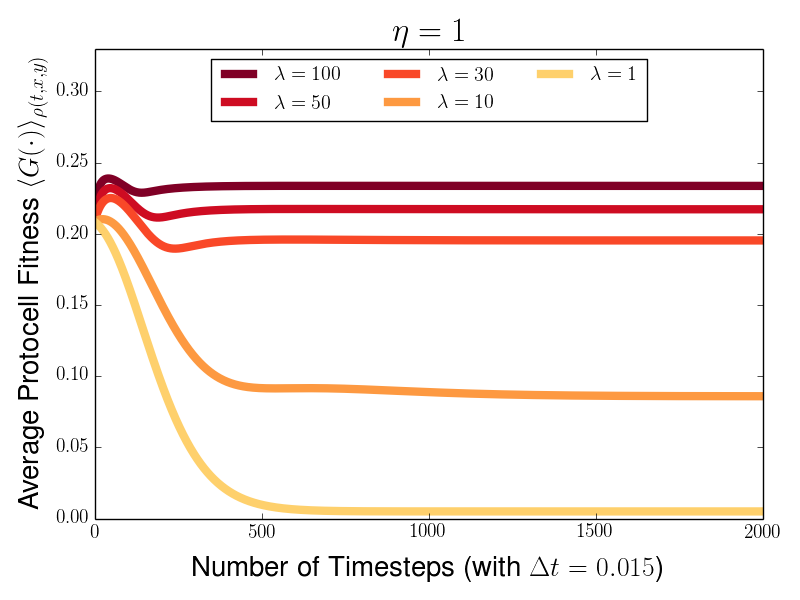}
      \includegraphics[width = 0.48\textwidth]{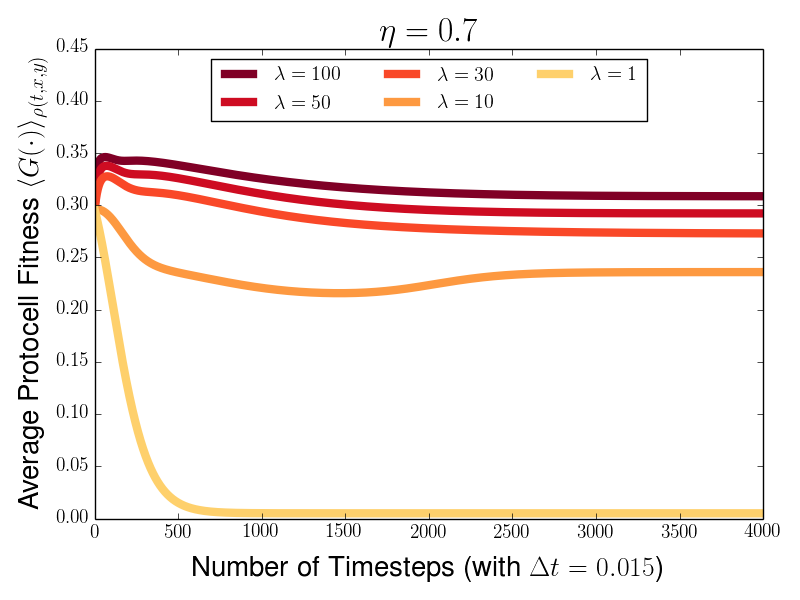}
    \caption{Average protocell-level fitness in the population $\langle G(x,y) \rangle_{\rho(t,x,y)}$ as a function of the number of time steps (with a step-length of $\Delta t = 0.0015$) for various intensities of between-protocell competition $\lambda$ and for the complementarity parameter $\eta = 1$ (left) or $\eta = 0.7$ (right). In both cases, the average protocell-level fitness is decreases in $\lambda$ for all time, and average fitness appears to equilibrate after several thousand time steps.}
    \label{fig:Gintime}
\end{figure}

With this possible convergence towards steady state densities for the trimorphic density $\rho(t,x,y)$, we can also consider the long-time behavior of our finite volume numerical solutions achieved for different parameters for the multilevel dynamics, exploring how changing the relative strength of between-protocell competition $\lambda$ or the complementarity parameter $\eta$ can impact the long-term support for fast, slow, and dimer replicators. In Figure \ref{fig:steadystatetrimorphicplots}, we provide the states achieved after 5000 time steps while varying $\lambda$ between the values $10$, $30$, $50$, and $100$ (from top row to bottom row) and varying $\eta$ between the values $1.0$, $0.9$, and $0.8$ (from left column to right column). For fixed $\eta$, we see that increasing $\lambda$ produces a greater proportion of the slow gene at steady state, with a greater weight of the slow gene carried through slow replicators for $\eta = 0.7$ and fractions carried through dimer replicators for $\eta = 0.9$ and even more for $\eta = 1.0$. For the lowest value of $\lambda$ considered, we see that population is concentrated very close to the all-fast state for $\eta = 1.0$, with slightly more representation of the fast gene for $\eta = 0.9$ and substantially more coexistence of the types for $\eta = 0.7$. This corresponds to the intuition gleaned from the threshold selection strengths $\lambda^*_{FS}$ (Equation \eqref{eq:lambdastarFS}) and $\lambda^*_{FD}$ (Equation \eqref{eq:lambdastarFDbvalues}) for the edges of the simplex, in which increasing the complementarity of the fast and slow genes results in the requirement of stronger between-protocell competition to sustain coexistence of both genes via multilevel selection.

\begin{figure}[ht]
    \centering
    \includegraphics[width = 0.32\textwidth]{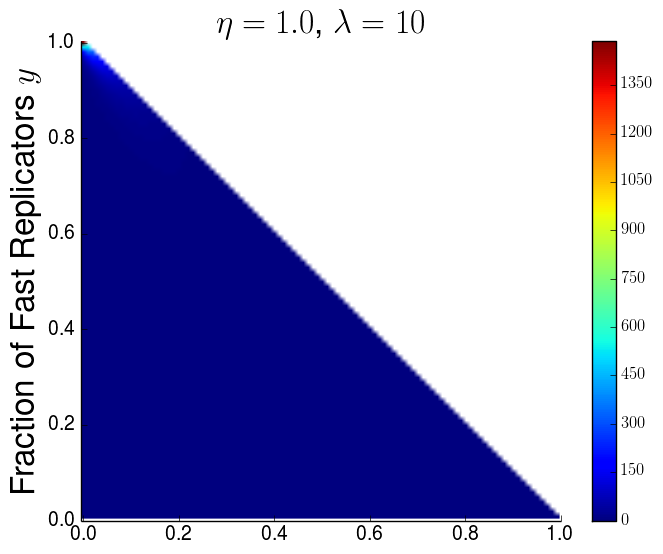}
     \includegraphics[width = 0.31\textwidth]{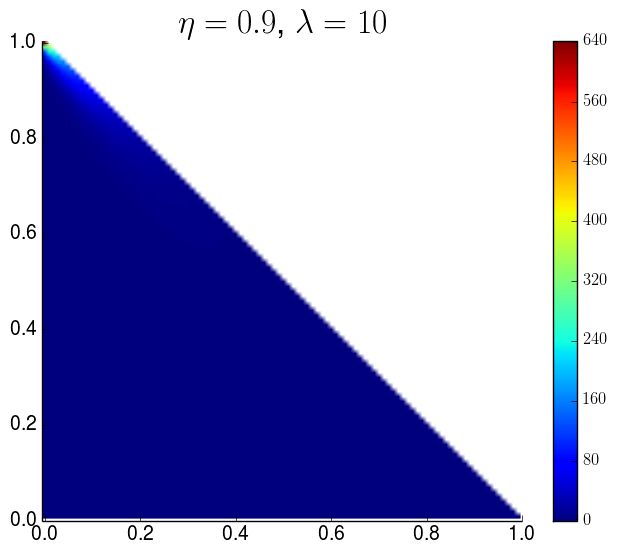}
         \includegraphics[width = 0.31\textwidth]{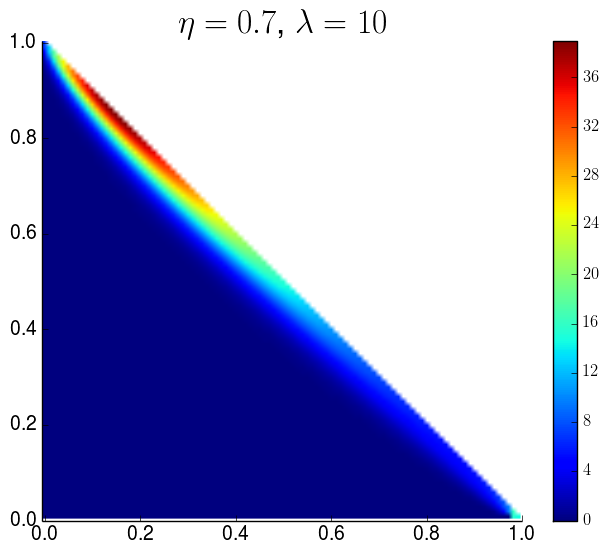}
         
       \includegraphics[width = 0.32\textwidth]{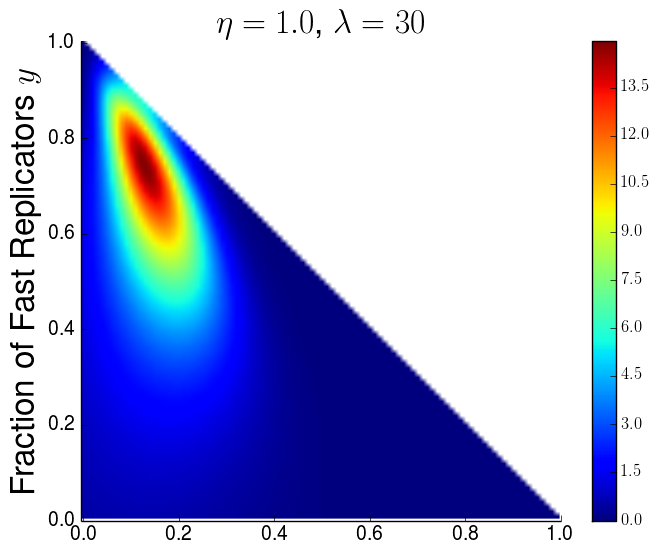}
     \includegraphics[width = 0.31\textwidth]{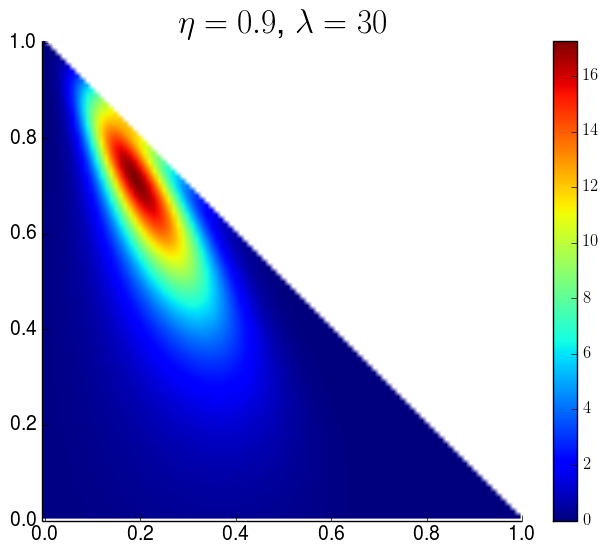}
         \includegraphics[width = 0.31\textwidth]{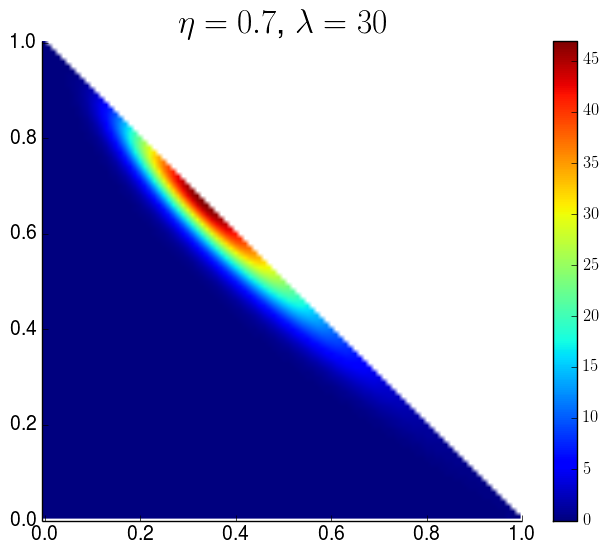} 
         
     \includegraphics[width = 0.32\textwidth]{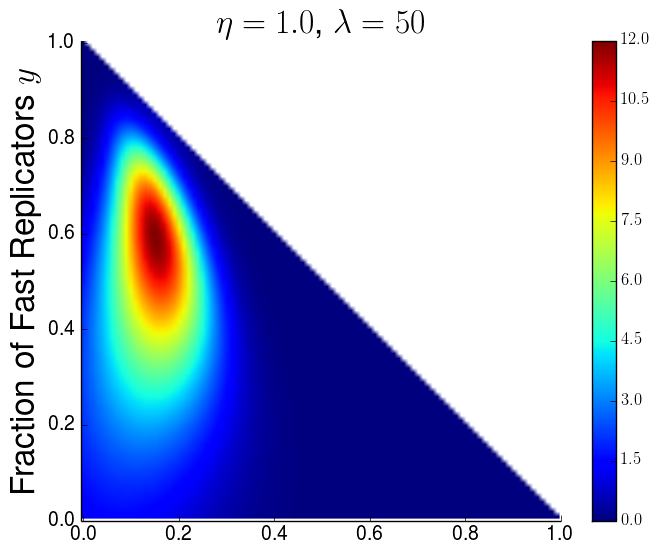}
       \includegraphics[width = 0.31\textwidth]{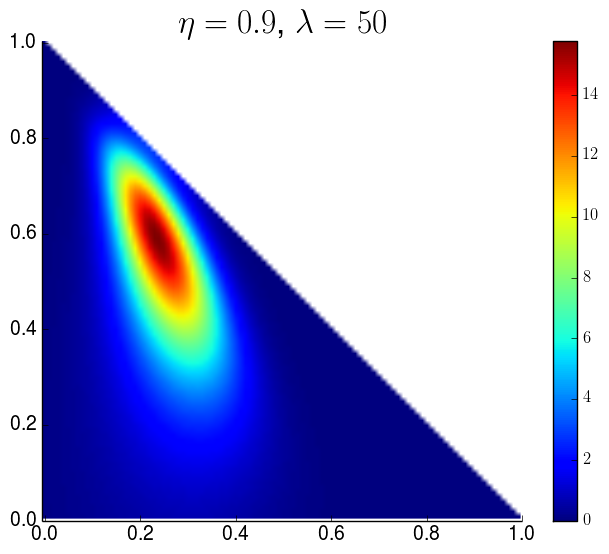}
      \includegraphics[width = 0.31\textwidth]{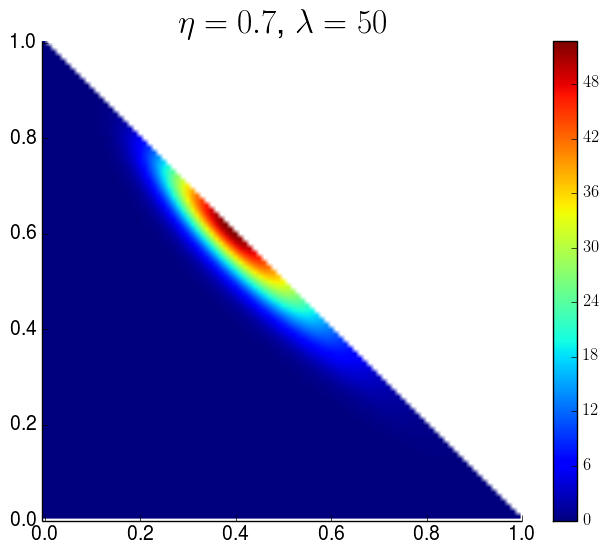}

      \includegraphics[width = 0.32\textwidth]{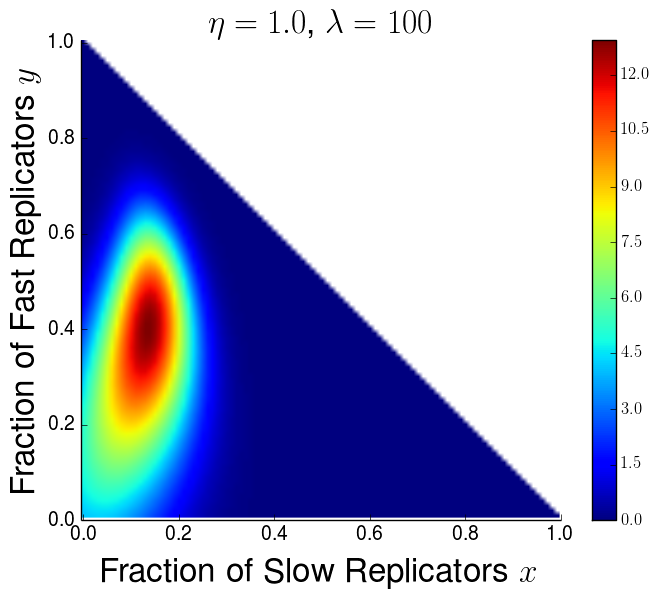}
      \includegraphics[width = 0.31\textwidth]{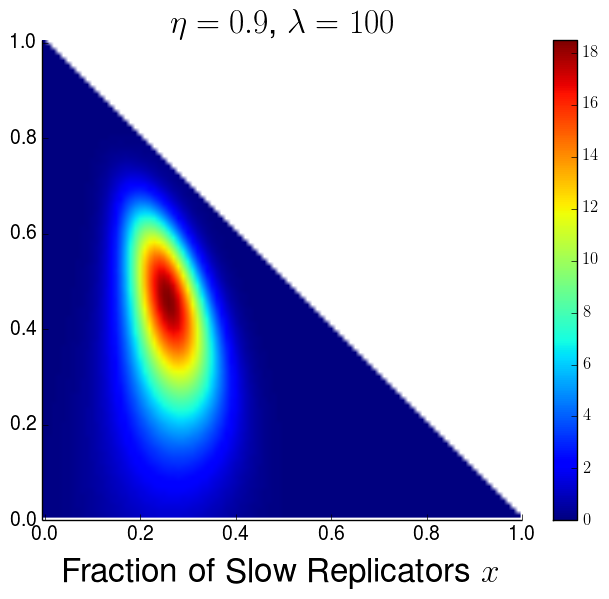}
      \includegraphics[width = 0.31\textwidth]{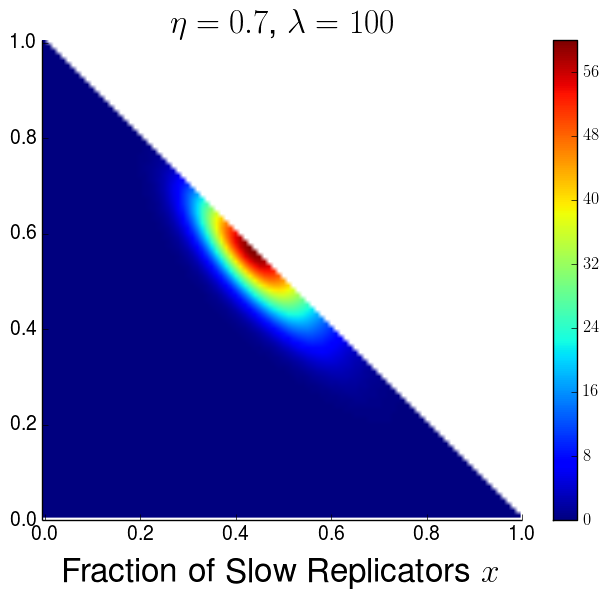}
    \caption{Numerically computed densities after 5000 time-steps with step-length of $\Delta t = 0.0015$ seconds for trimorphic multilevel dynamics with various values of between-protocell selection strength $\lambda$ and complementarity parameter $\eta$. From top to bottom, each row presents a different value of $\lambda$: $10$ (first row), $30$ (second row), $50$ (third row), and $100$ (fourth row). Each column corresponds to a different value of $\eta$: $1.0$ (left column), $0.9$ (center column), and $0.7$ (right column).}
    \label{fig:steadystatetrimorphicplots}
\end{figure}

\clearpage

For the case of $\eta = 1$, we see from Figure \ref{fig:steadystatetrimorphicplots} that a substantial fraction of the protocells are composed primarily of dimers at steady state. Dimers help to ensure coexistence of fast and slow genes when they are perfect complements for a protocell, and within-cell competition would otherwise eliminate slow genes. Having studied the ability for the multilevel dynamics to promote coexistence of the fast and slow genes, we can now quantify the impact of this coexistence upon the overall protocell-level fitness of the population. In Figure \ref{fig:Gplottrimorphic}, we plot the average protocell-level fitness $\langle G(x,y) \rangle_{\rho(x,y)}$ for the numerically computed state for the trimorphic dynamics after $5000$ time steps, and compare this to the protocell-level fitness $\langle G(z) \rangle_{g(z)}$ achieved on the fast-dimer edge of the simplex for a uniform initial condition and the same strength of between-protocell competition. We see that the behavior of the protocell-level fitness is similar for the trimporphic dynamics and for the fast-dimer dimorphic multilevel competition for the cases in which $\eta = 1$ (Figure \ref{fig:Gplottrimorphic}, left) and for $\eta = 0.7$ (Figure \ref{fig:Gplottrimorphic}, right). In both cases, we see that the protocell-level fitness for the trimporphic dynamics appears to tend to $G(0,0)$, the protocell-level fitness of the all-dimer protocell, as $\lambda$ becomes large. Therefore it appears that the long-time collective outcome cannot outperform---achieve higher average protocell-level fitness than---the fitness of an all-dimer protocell. This is true even though protocells featuring a majority of slow genes maximize the protocell-level reproduction rate for the case in which $\eta = 0.7$. This suggests that an analogous shadow of lower-level selection may hold in the case of trimorphic multilevel competition, and that the best outcomes that can be achieved by multilevel selection are the collective replication rates of compositions that are equilibria under within-protocell competition. 

\begin{figure}[ht]
    \centering
    \includegraphics[width = 0.48\textwidth]{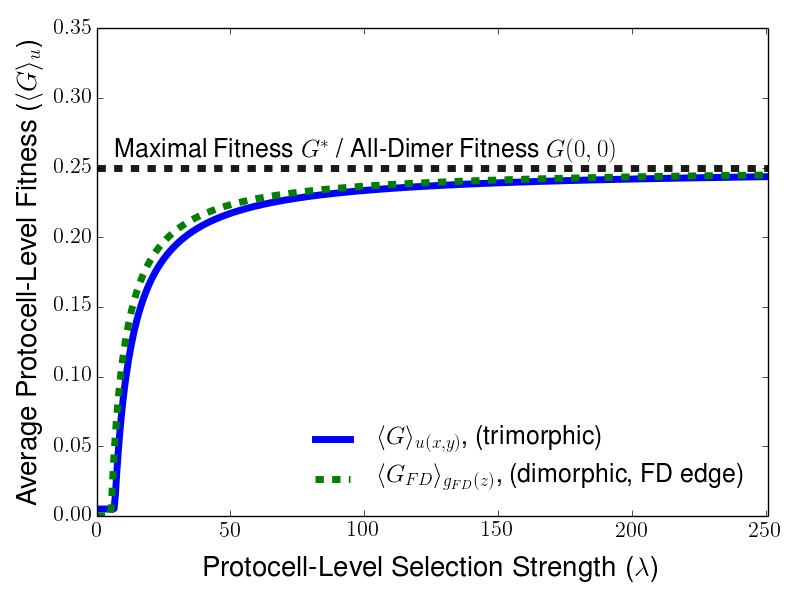}
    \includegraphics[width = 0.48\textwidth]{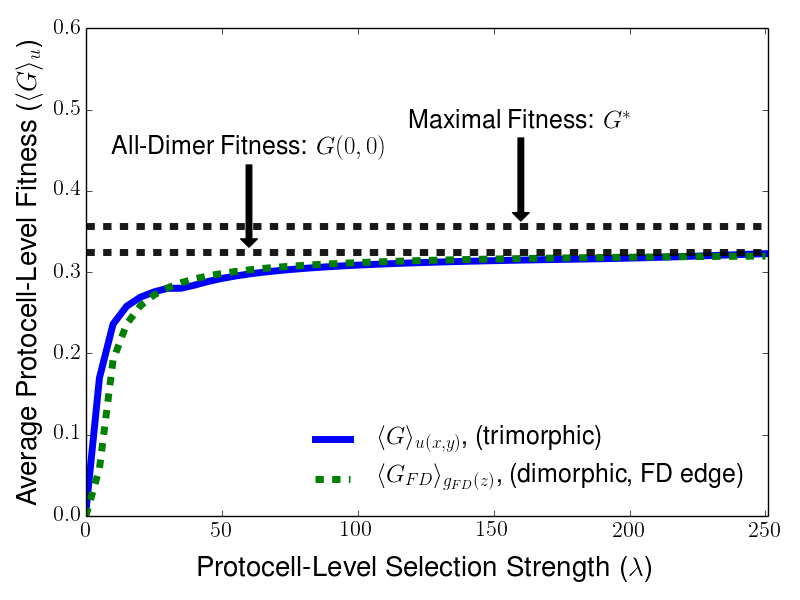}
    \caption{Numerically computed average protocell-level fitness (solid blue line) after 5000 time steps with step-length $\Delta t = 0.0015$ for complementarity parameter $\eta = 1.0$ (left) and $\eta = 0.7$ (right), plotted as a function of the relative strength of between-protocell competition $\lambda$. A comparison is provided with the average protocell-level fitness achieved at steady state (dashed green line) given by Equation \eqref{eq:FSFDavgG} for dimorphic multilevel competition on the fast-dimer edge of the simplex. For both values of $\eta$, the average protocell-level fitness for the trimorphic competition tends to the reproduction rate of the all-dimer protocell as $\lambda \to \infty$. For the case of $\eta = 0.7$, the lower horizontal dashed black line indicates the collective fitness $G(0,0)$ of an all-dimer protocell, while the higher dashed black line indicates the maximal possible collective fitness $G^* := \max_{x + y \leq 1} G(x,y)$ among possible compositions on the simplex. These two dashed lines coincide for the case in which $\eta = 1$ (and the all-dimer composition achieves the maximal protocell-level reproduction rate). }
    \label{fig:Gplottrimorphic}
\end{figure}

We also explore other quantities characterizing the support for coexistence of the fast and slow gene under the trimorphic dynamics. In Figure \ref{fig:peakandmeantrimorphic}, we display the fraction of slow gene $\rm{\%Slow} = x + \frac{z}{2}$ present in numerically computed states after $5000$ time steps of the trimorphic dynamics with $\eta = 1$ in both the most abundant protocell composition (Figure \ref{fig:peakandmeantrimorphic}, left) and averaged across all of the protocells in the population (Figure \ref{fig:peakandmeantrimorphic}, right). We see that the mean and modal fraction of slow genes increases with $\lambda$, and that the modal fraction of slow genes has good agreement with the modal composition of slow genes realized for the dimorphic dynamics on the fast-dimer edge of the simplex for an initial uniform density and the same relative selection strength $\lambda$. We see that there is less agreement between the mean fraction of slow genes under the trimorphic and fast-dimer dynamics, but a similar qualitative picture of increasing mean fraction of slow genes with increasing between-protocell competition in the direction of the optimal fifty-fifty mix of fast and slow is still present for both the dimorphic and trimorphic models.

\begin{figure}[htp!]
    \centering
    \includegraphics[width = 0.48\textwidth]{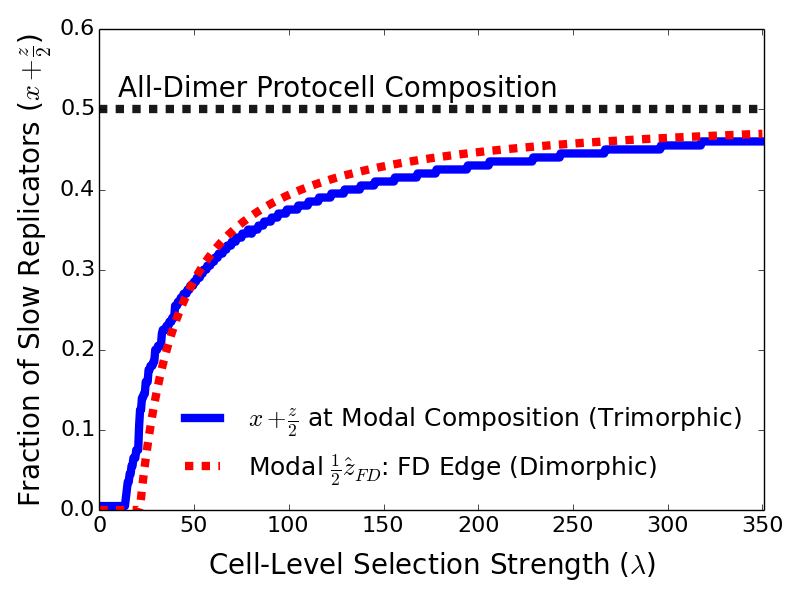}
     \includegraphics[width = 0.48\textwidth]{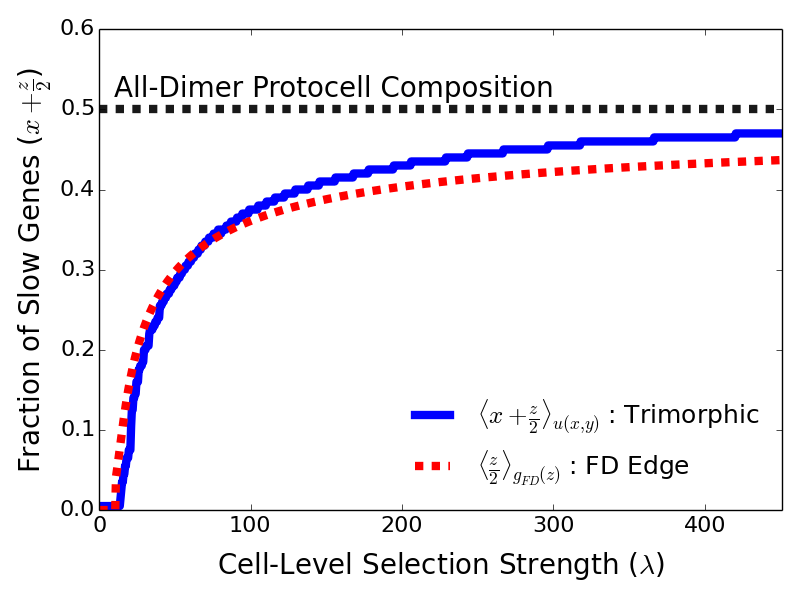}
    \caption{Mean (left) and most abundant (right) fraction of slow replicators in numerically computed trimorphic population after 5000 time steps with step-length $\Delta t = 0.0015$ for complementarity parameter $\eta = 1.0$, plotted as a function of the relative strength of between-protocell competition $\lambda$. A comparison is provided between the numerical trimorphic solutions (solid blue lines) and the mean and modal fractions of the slow gene $\frac{z}{2}$ for the steady state $g^{\lambda}_{\theta}(z)$  on the fast-dimer edge of the simplex (dashed red lines) corresponding to an initial uniform distribution. For fast-dimer dimorphic dynamics, mean fraction of slow genes is calculated numerically using Equation \eqref{eq:glambdatheta} and the modal fraction comes from the analytical formula provided by Equation \eqref{eq:modalFDlambda}. The horizontal dashed line corresponds to the portion $0.5$ of slow genes in an all-dimer protocell.}
    \label{fig:peakandmeantrimorphic}
\end{figure}

Finally, we can consider the impact on varying the complementarity parameter $\eta$ on the trimorphic dynamics for a fixed strength of between-protocell competition. In Figure \ref{fig:Getatrimporphic}, we plot, as a function of $\eta$, the average protocell-level fitness $\langle G(x,y) \rangle_{\rho(x,y)}$ for the trimorphic dynamics after $5000$ time steps (solid blue line) and the equivalent protocell-level fitnesses $\langle G_{FS}(x) \rangle_{f(x)}$ (dash-dot black line) and $\langle G_{FD}(z) \rangle_{g(z)}$ (dash-dot red line)  for dimorphic multilevel competition on the fast-slow and fast-dimer edges of the simplex. For the case with $\lambda = 10$, we see that the average protocell-level fitness for the long-time numerical trimorphic solutions roughly agrees with the average protocell level fitness for the dimorphic steady states on the fast-slow edge (for $\eta < 0.8$0) and on the fast-dimer edge (for $\eta > 0.8$). For $\lambda = 0.75$, we see a similar agreement between the average protocell-level fitness for the trimorphic long-time state and the maximal possible dimorphic collective fitnesses on the fast-slow edge (for $\eta < \frac{2}{3}$) and on the fast-dimer edge (for $\eta > \frac{2}{3}$). In addition, we see, for both values of $\lambda$ and for any complementarity parameter $\eta$, that the protocell-level fitness of the trimorphic populations does not exceed the larger of the collective reproduction rates between the all-dimer protocell and the all-slow protocell, suggesting that there may be a trimorphic analogue of the shadow of lower-level selection in which the protocell-level fitness is limited by the maximal collective fitness among the within-protocell equilibria. In addition, these numerical results suggest some connection between the collective outcomes in the trimorphic dynamics and the collective outcomes of the multilevel dynamics on edges of the simplex, which suggests the possibility that the long-term behavior of our trimorphic PDE model may be determined by a tug-of-war between the gene-level advantage of fast replicators and the collective advantage of all-dimer or all-slow protocells over protocells with an all-fast composition.     

\begin{figure}[ht]
    \centering
    \includegraphics[width = 0.48\textwidth]{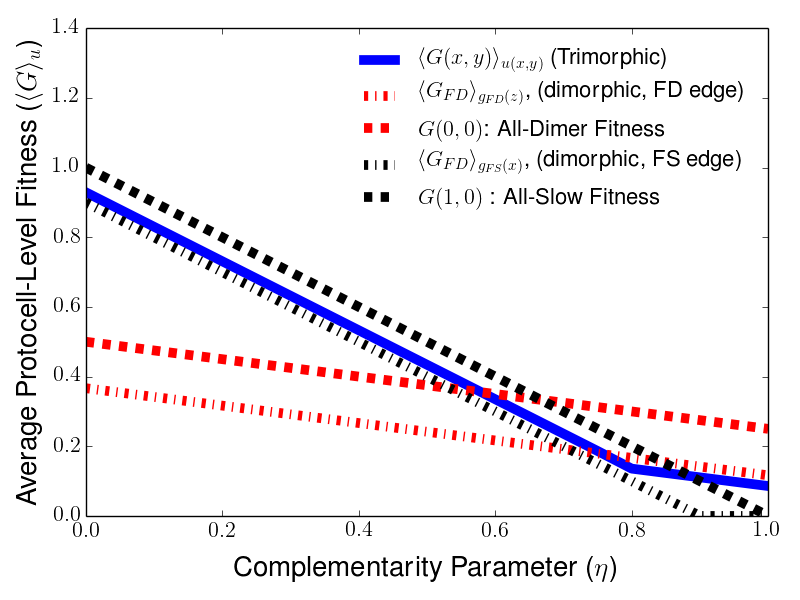}
    \includegraphics[width = 0.48\textwidth]{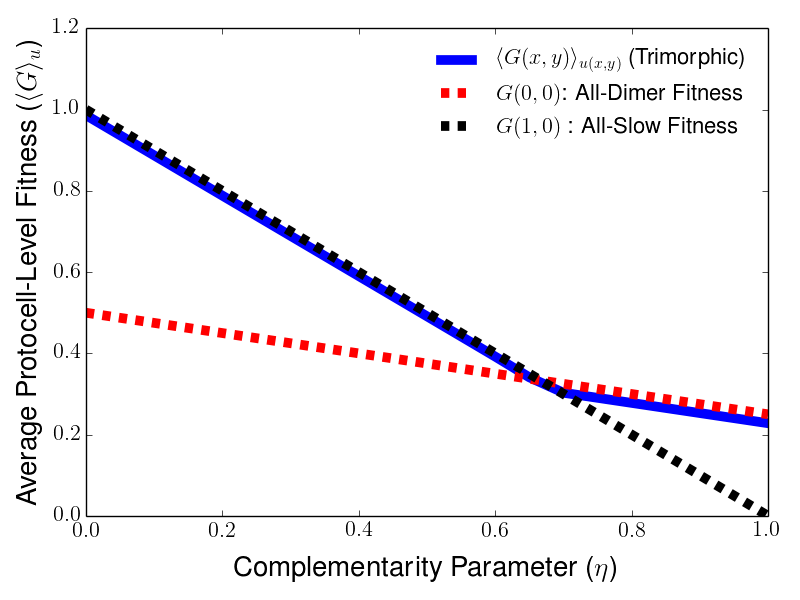}
    \caption{Numerically computed average protocell-level fitness after 5000 time steps with step-length $\Delta t = 0.0015$ (solid blue line) and relative selection strengths $\lambda = 10$ (left) and $\lambda = 75$ (right), plotted as a function of the complementarity parameter $\eta$. Comparison provided with analytical solution for average protocell fitness at steady state under dimorphic model on fast-slow (black dash-dot line) and fast-dimer (red dash-dot line) edges of the simplex for same selection strength $\lambda$ and uniform initial density, as well as the collective reproduction rate of the all-slow protocell (black dashed line) and the all-dimer protocell (red dashed line).}
    \label{fig:Getatrimporphic}
\end{figure}

\section{Discussion} \label{sec:discussion}
In this paper, we introduced a a PDE model for the evolution of protocells under multilevel selection in which there an evolutionary tension between within-protocell competition for replication among genes and the between-protocells competition for replication among the cells themselves. In particular, when there is a genetic template (a ``fast replicator'' that has a selective advantage over another (a ``slow replicator'') for replication within protocells, then no level of between-protocell competition can produce the mix of the two types of replicators that is optimal for protocell-level replication. This shadow of gene-level selection is particularly extreme for the case in which the two genetic templates are perfect complements for protocell fitness: any slight gene-level advantage for one of the two templates prevents coexistence. We then introduce the possibility of linking the two genes together as a dimer (or proto-chromosome) template, formulating a PDE model for the multilevel competition of slow, fast, and dimer replicators. In this extended model, we show through a simplified analytical approach and through numerical simulations of the full model that the presence of dimers can allow for the long-time coexistence of the slow and fast genes. Dimerization can help to overcome the shadow of lower-level selection;  chromosomes can thus play a key role in the major evolutionary transition to cellular life.

In previous work on the evolution of chromosomes, dimerization was primarily presented as a means for overcoming the possibility of stochastic loss of a necessary genetic template \cite{gabriel1960primitive,smith1993origin}. In particular, the simulations of Maynard Smith and Szathmary have shown that dimerization was most effective in protocells with low copy number, but that dimers cannot persist when their are many genetic templates per protocell \cite{smith1993origin}. The stochastic corrector mechanism \cite{szathmary1987group,grey1995re} and the package model \cite{bresch1980hypercycles,niesert1981origin} similarly fail to support the persistence of a dimer replicator with an individual-level disadvantage in the limit of many genes per protocell, and therefore different mechanisms for multilevel selection are required in the case of high copy number. For the kind of nested birth-death model introduced by Luo and coauthors \cite{luo2014unifying,van2014simple}, we have shown that dimers / proto-chromosomes can help to facilitate the evolutionary coexistence of complementary genetic templates in a deterministic, large-population PDE limit and in the presence of a fast replicator with an individual-level selective advantage. As a result, our analysis complements the work of Maynard Smith and Szathmary, illustrating that genetic linkage into proto-chromosomes can help to promote coexistence of complementary genes either as a means to overcome stochastic effects at low copy number or to overcome competitive effects at high copy number.  

\sloppy{Dimerization and further chromosome formation are a way to eliminate the individual-level competition between complementary replicators within a protocell, as sufficiently strong between-cell competition selecting for linked slow-fast dimers can eventually eliminate the need for separate slow and fast dimers whose head-to-head competition casts a long shadow on the needs cell-level viability. The benefits of dimerization bring to mind the mechanisms identified to allow for the cooperative coexistence needed to establish forms of multicellular life \cite{aktipis2015cancer}. %
Furthermore, we note that dimerization has a different impact on promoting persistence of beneficial traits via multilevel selection when compared to the mechanisms of assortment and reciprocity studied to promote cooperation in evolutionary game theory \cite{taylor2007transforming,nowak2006five,cooney2019assortment}. In particular, the game-theoretic mechanisms only decrease the individual-level advantage of defectors over cooperators, but do not improve the collective payoff of the full-cooperator group or the maximal possible long-time collective payoff \cite{cooney2019assortment}. As a result, assortment and reciprocity can decrease the level of between-group competition required for the evolution of cooperation, but cannot actually eliminate the shadow of lower-level selection and promote a collective benefit exceeding that of a full-cooperator group. By contrast, dimerization actually induces a greater individual-level disadvantage of dimer replicators in exchange for an improved collective advantage of all-dimer protocells, and therefore the linkage mechanism provided by dimerization can help to improve the best possible outcome that can be achieved by protocells under multilevel selection. This motivates further work on modeling within-group mechanisms that can work synergistically with multilevel selection to promote beneficial collective outcomes, finding more general approaches for overcoming the shadow of lower-level selection that may apply to a wide variety of biological settings. }   

The focus of this paper was on the evolution of template coexistence in protocells and on the origin of chromosomes, two major evolutionary transitions that took place early on in the evolution of complex life. However, similar problems regarding levels of selection and the coexistence and integration of genetic templates arise in modern biological systems as well. For example, it has been shown theoretically that a costly microbrial trait can evolve via multilevel selection through the benefit conferred to its host and transmission via a mix of vertical and horizontal transmission \cite{van2019role}. Multilevel selection has been attributed as a factor in evolutionary dynamics featuring transposable genetic elements such as plasmids, bacteriophages, transposons, and viruses \cite{brunet2015multilevel,iranzo2017disentangling}. In the context of viruses, public goods dilemmas can arise due to defecting interfering viral particles that ``cheat'' off the replicative mechanisms provided by full viral genomes and thus can modulate the dynamics of viral infections \cite{huang1970defective,manzoni2018defective,szathmary1992natural}, but collective infection of cells by a cohort of viral particles \cite{diaz2017sociovirology} result in either competitive exclusion \cite{turner1999prisoner} or heterotypic cooperation between complementary virus strains \cite{xue2016cooperation,turner2003escape}. The presence of higher levels of selection may also help to solve dilemmas faced by plasmids, from establishing replication control mechanisms to regulate plasmid copy number \cite{paulsson2002multileveled} to limiting horizontal gene transfer to the tragedy of the commons imposed by runaway invasion of parasitic plasmids \cite{smith2012tragedy,lopez2021modeling}.

The themes of multilevel competition and dimerization are particularly present in studying the evolution of antibiotic multi-resistance \cite{summers2006genetic}. Experiment work has shown that low levels of treatment with antibiotics and heavy metals can select for the evolution of plasmids conferring resistance to multiple drugs \cite{gullberg2014selection}. An experiment by Sachs and Bull explored the of coinfection of bacterial cells with phages carrying different antibiotic-resistance genes can result in the evolution of copackaging the two phage genomes into into a single protein coat, mediating conflict between complementary resistance genes \cite{sachs2005experimental,velicer2005benefits}. In both of these cases, treatment with multiple antibiotics created a scenario in which the presence of multiple antibiotic resistance genes is beneficial for between-bacterium competition, while maintaining the resistance gene imposes a cost upon the bacterium. There has been recent simulation work on multilevel models for antibiotic resistance and multi-resistance using a membrane computing approach \cite{campos2019simulating,campos2020simulating}, but PDE models of multilevel selection acting on complementary but competing resistance genes could provide analytical insight.

Our numerical results from Section \ref{sec:trimorphicnumerics} for the trimorphic fast-slow-dimer competition raise potentially interesting mathematical questions about the dynamics of multilevel selection with three types of individuals. In particular, for the case in which fast and slow genes are perfect complements for protocell replication, reasonable agreement was found between the average protocell-level fitness found in the numerically computed trimorphic population after many time steps and the analytically calculated steady state for the dimorphic competition on the fast-dimer edge of the simplex. This agreement suggests that perhaps some of the main characteristics of the long-time behavior of the trimorphic dynamics will resemble the main results for two-type multilevel competition illustrated in Section \ref{sec:existingresults}, such as a threshold level of between-protocell competition required to establish the long-time coexistence of fast and slow genes and a protocell-level fitness tending to the collective replication rate at the all-dimer equilibrium in the limit of infinitely strong between-protocell competition. We also saw that the numerically computed protocell-level fitness was always limited by the maximum of the reproduction rates at the all-dimer or all-slow equilibrium, suggesting perhaps that an intermediate collective optimum may not be achieved by the trimorphic dynamics. In future work, we hope to analytically explore these sorts of claims, and to understand how the long-time behavior of the trimorphic PDE depends on the initial distribution of protocell compositions. 

Another direction for future research would be to further explore and generalize the formulation of PDE models for multilevel selection with more than two types of individuals. In particular, another potential application for this type of PDE would be to study the multilevel competition between individuals playing two-strategy cooperative dilemmas who follow one of three strategies: always cooperate, always defect, or a form of conditional cooperation in which the strategy played depends on either the strategy of one's opponent or the strategic composition of the group. Such a three-type model would allow for the study of the synergy of multilevel selection and mechanisms like direct or indirect reciprocity, showing how strategies such as tit-for-tat (under direct reciprocity) \cite{imhof2005evolutionary} or stern judging (under indirect reciprocity) \cite{nowak2005indirect,pacheco2006stern} can help promote the persistence of cooperation via multilevel selection. Additional generalization of the within-group and between-group replication rates to a broader class of sufficiently regular functions could allow for a flexible model for trimorphic multilevel competition, potentially leading to insight into the tug-of-war between individual incentives and collective benefits of different traits or strategies. The finite volume numerical approach described in Section \ref{sec:finitevolume} could also be extended to incorporate a larger number of types of individuals, or to incorporate ecological constraints to allow exploration of multilevel selection in populations with variable group size \cite{simon2016group,janssen2006dynamic}. 

Between the range of possible mathematical and biological extensions of our modeling approach for deterministic multilevel selection, we see that are a variety of directions for future analytical and numerical work. In addition, we have seen that our protocell model and the mechanism of dimerization provides both solutions and further puzzles for evolutionary competition across levels of biological selection. By further examining the conflict between the interests of the individual and the group, we can further explore the necessary evolutionary mechanisms and design principles required to achieve complex, multilevel population structure, ranging from the evolution of protocells and chromosomes to collective behavior of animal groups and cooperative management of the global commons.     

\renewcommand{\abstractname}{Acknowledgments}
\begin{abstract} 
DBC received support from the National Science Foundation through grants DMS-1514606 and GEO-1211972 and from the Simons Foundation through the Math + X grant awarded to University of Pennsylvania. DBC and SAL received support from the Army Research Office through grant W911NF-18-1-0325. The authors would like to thank Erol Ak\c{c}ay, George Constable, Louis Fan, Yoichiro Mori, Joshua Plotkin, Corina Tarnita, and Carl Veller for helpful discussion. 
\end{abstract}

\bibliography{multilevelselection}
\bibliographystyle{ieeetr}

\appendix

\section{Derivation of PDE Models for Multilevel Selection from Individual-Based Nested Birth-Death Process}
\label{sec:derivation}

In this section, we present derivations of the PDEs describing the dynamics of our baseline protocell model and our model of trimorphic protocell dynamics. Our starting point is a nested birth-death process describing gene-level and protocell-level replication events in a population of $m$ protocells each composed on $n$ replicators. From this finite population stochastic process, we first derive systems of ODEs describing how the composition of protocells evolves in the limit of infinitely protocells ($m \to \infty$) each consisting of finitely many genes. Then, taking the limit of infinitely many genes per protocell ($n \to \infty$), we obtain PDE descriptions of how the density of protocell descriptions evolves in time.

In Section \ref{sec:derivationbaseline}, we derive the baseline protocell model of Equation \eqref{eq:protocellPDEG} for multilevel competition featuring fast and slow replicators. We follow the approach taken by Luo and coauthors \cite{luo2014unifying,van2014simple} and by Cooney \cite{cooney2019replicator} to derive the limiting PDE description. In Section \ref{sec:derivationtrimorphic}, we provide a similar derivation for the trimorphic dynamics described by Equation \eqref{eq:multileveltrimorphic}, extending the prior approach to account for the presence of a third type of replicator under gene-level and protocell-level competition. We note that the derivation of trimorphic dynamics is somewhat more involved than the derivation of the baseline model, particularly due to the increased number of possible gene-level events that can occur in the trimorphic setting.

\subsection{Derivation of Baseline Protocell Model}
\label{sec:derivationbaseline}

For a population of $m$ protocells each composed of $n$ genes, we denote by $\rho_{i}^{m,n}(t)$ the fraction of protocells featuring $i$ slow replicators and $n-i$ fast replicators. For the derivation of our PDE limit, we will first focus on compositions satisfying $i, n-i \geq 1$, which corresponds to the protocell compositions featuring a nontrivial mix of fast and slow replicators. For such compositions, we will now formulate how the probability $\rho_{i}^{m,n}(t)$ evolves under gene-level dynamics. In our finite population model, we assume that gene-level dynamics resemble a continuous-time Moran process. Fast and slow replicators produce copies of themselves with rates $1 + w_I b_F$ and $1 + w_I b_S$, and that these copies replace a randomly chosen replicator in the same group. 

Under these rules, the fraction of protocells $\rho_{i}^{m,n}(t)$ with $i$ slow replicators increases by $\frac{1}{m}$ due to gene-level competition if one of the following two events happens

\begin{itemize}
    \item A slow replicator is born and a fast replicator is replaced in an $(i-1)$-protocell, which occurs with rate
    \[ m \rho_{i-1}^{m,n}(t) \left( i - 1 \right) \left( \frac{n - i + 1}{n} \right) \left( 1 + w_I b_S \right) \]
    
    \item A fast replicator is born and a slow replicator is replaced in an $(i+1)$-protocell, which occurs with rate
    \[m \rho_{i+1}^{m,n}(t) \left(n - i - 1 \right) \left( \frac{i+1}{n} \right) \left( 1 + w_I b_F \right). \]
\end{itemize}
Similarly, the fraction of protocells $rho_{i}^{m,n}(t)$ with $i$ slow replicators decreases by $\frac{1}{m}$ due to gene-level competition if one of the following two events happens
\begin{itemize}
    \item A slow replicator is born and a fast replicator is replaced in an $i$-protocell, which occurs with rate
    \[ m \rho_{i}^{m,n}(t) \left( i \right) \left( \frac{n - i}{n} \right) \left( 1 + w_I b_S \right) \]
    
    \item A fast replicator is born and a slow replicator is replaced in an $i$-protocell, which occurs with rate
    \[m \rho_{i}^{m,n}(t) \left(n - i \right) \left( \frac{i}{n} \right) \left( 1 + w_I b_F \right). \]
\end{itemize}
Finally, the fraction of protocells $\rho_{i}^{m,n}(t)$ is unchanged to gene-level events in which the offspring replicator replaces a replicator of the same type. 

Turning to between-protocell compeition, we now consider a group-level birth-death process that resembles a continuous-time process featuring the $n+1$ possible group competitions. We assume that a protocell featuring a fraction $x$ slow replicators and $1-x$ fast replicators produces a copy of itself with rate $\Lambda \left( 1 + w_G G_{FS}(x) \right)$, and that the resulting offspring protocell replaces a randomly chosen protocell in the population. Under this rule for protocell-level replication events, we see that the fraction of $i$-protocells $\rho_{i}^{m,n}(t)$ increases by $\frac{1}{m}$ due to between-protocell competition when an $i$-protocell is chosen to reproduce and a protocell with a number of slow replicators other than $i$ is chosen to be replaced. This event occurs with rate
\[ \Lambda m \rho_{i}^{m,n}(t) \left[1 + w_G G_{FS}\left(\frac{i}{n}\right)   \right]  \left[ 1 - \rho_{i}^{m,n}(t) \right].\] 
Similar, the fraction $\rho_{i}^{m,n}(t)$ decreases by $\frac{1}{m}$ due to protocell-level competition when a protocell with featuring a number of slow replicators other than $i$ is chosen to reproduce and an $i$-protocell is chosen to be replaced. Such an event occurs with rate
\[\Lambda m \rho_{i}^{m,n}(t) \ds\sum_{\substack{j=1 \\ j \ne i}}^n \rho_{j}^{m,n}(t) \left[ 1 + w_G G\left( \frac{j}{n} \right)\right] .\]

Following the approach of Luo and coauthors \cite{luo2014unifying,van2014simple}, we can use these rates to calculate the infinitessimal mean of the two-level birth-death process, which yields
\begin{align} \label{eq:infmeandimorphic}
\begin{split}
M_{\Delta t} &:= E [f_{i}^{m,n}(t + \Delta t) - f_{i}^{m,n}(t)] \\
 &= \frac{1}{m} P\left(f_{i}^{m,n}(t + \Delta t) - f_{j}^{m,n}(t) = \frac{1}{m} \right) - \frac{1}{m} P\left(f_{i}^{m,n}(t + \Delta t) - f_{i}^{m,n}(t) = -\frac{1}{m} \right) + o\left(\Delta t\right)\\
 &= \frac{1}{m} \left[ m f_{i-1}^{m,n}(t) \left( i - 1 \right) \left( \frac{n - i + 1}{n} \right) \left( 1 + w_I b_S \right) \right] \Delta t \\
 &+ \frac{1}{m} \left[m \rho_{i+1}^{m,n}(t) \left(n - i - 1 \right) \left( \frac{i+1}{n} \right) \left( 1 + w_I b_F \right) \right] \Delta t \\
 &- \frac{1}{m} \left[ m f_{i}^{m,n}(t) \left( i \right) \left( \frac{n - i}{n} \right) \left( 1 + w_I b_S \right)\right] \Delta t \\ 
 &- \frac{1}{m} \left[ m f_{i}^{m,n}(t) \left(n - i \right) \left( \frac{i}{n} \right) \left( 1 + w_I b_F \right) \right] \Delta t \\
 &+ \frac{1}{m} \left\{ \Lambda m \rho_{i}^{m,n}(t) \left[1 + w_G G_{FS}\left(\frac{i}{n}\right)   \right]  \left[ 1 - \rho_{i}^{m,n}(t) \right] \right\} \\ &- \frac{1}{m} \left\{ \Lambda m \rho_{i}^{m,n}(t) \ds\sum_{\substack{j=1 \\ j \ne i}}^n \rho_{j}^{m,n}(t) \left[ 1 + w_G G_{FS}\left( \frac{j}{n} \right)\right] \right\} + o(\Delta t)
 \end{split}
 \end{align}
 Using the forward and backward first-order difference quotients
 \begin{equation} \label{eq:dimorphicfirst} D_1^+(u(\frac{i}{n})) = \frac{u(\frac{i+1}{n}) - u(\frac{i}{n})}{\frac{1}{n}} \: \:, \: \: D_1^-(u(\frac{i}{n})) = \frac{u(\frac{i}{n}) - u(\frac{i-1}{n})}{\frac{1}{n}} \end{equation}
 and the central second-order difference quotient
\begin{equation} \label{eq:dimorphicsecond}
    D_2\left(u(\frac{i}{n})\right) = \frac{u(\frac{i+1}{n}) - 2 u(\frac{i}{n}) + u(\frac{i-1}{n})}{\frac{1}{n^2}},
\end{equation}
we can rearrange Equation \eqref{eq:infmeandimorphic} to write the infinitessimal mean as
\begin{dmath} \label{eq:infmeanrearrangeddimorphic}
\frac{E [f_{i}^{m,n}(t + \Delta t) - f_{i}^{m,n}(t)]}{\Delta t} = \frac{1}{n} D_2\left( \frac{i}{n} \left(1 - \frac{i}{n}\right) f_{i}(t) \right) + \Lambda w_G f_i(t) \left( G_{FS}\left(\frac{i}{n}\right)  - \ds\sum_{j=0}^N G_{FS}\left(\frac{j}{n}\right) f_j(t) \right) + w_I \left[ b_F D_1^+ \left( \frac{i}{n} \left(1 - \frac{i}{n} \right) f_i(t) \right) - b_S D_1^- \left( \frac{i}{n}  \left(1 - \frac{i}{n} \right) f_{i}(t) \right)  \right] + \frac{o(\Delta t)}{\Delta t}.
\end{dmath}
We can also calculate the infinitessimal variance of the two-level birth death process. Noting that all of the transition rates are linear in $m$, we see that 
\begin{align}
\begin{split}
V_{\Delta t} := & E \left[ \left( \rho_{i}^{m,n}(t + \Delta t) - \rho_{i}^{m,n}(t) \right)^2 \right] \\ 
 = & \frac{1}{m^2} P\left(\rho_{i}^{m,n}(t + \Delta t) - \rho_{i}^{m,n}(t) = \frac{1}{m} \right) - \frac{1}{m^2} P\left(\rho_{i}^{m,n}(t + \Delta t) - \rho_{i}^{m,n}(t) = \frac{1}{m} \right) \\ 
 = & \frac{1}{m^2} \left[O(m) \right] + o(\Delta t).
\end{split}
\end{align}
Therefore we can calculate that
\begin{equation}
    \ds\lim_{m \to \infty} E \left[ \left( \rho_{i,j}^{m,n}(t + \Delta t) - \rho_{i,j}^{m,n}(t) \right)^2 \right]  = 0,
\end{equation}
and the infinitessimal variance vanishes as $m \to \infty$. Therefore, in this limit, the distribution of $f_{i}^n(t) := \lim_{m \to \infty} f_{i}^{m,n}(t)$ is a constant equal to its mean $E[f_{i}^n(t)]$. This means that taking the limit of infinitely many protocells removes the randomness from the multilevel dynamics, and, we can look to describe the deterministic evolution of the fraction of protocells $f_{i}^n(t)$ in this limit. Taking the limit of both sides of Equation \eqref{eq:infmeanrearrangeddimorphic} as $m \to \infty$ and using the linearity of expectation, we obtain
\begin{dmath}
    \frac{f^n_i(t+\Delta t) - f^n_i(t)}{ \Delta t} = \frac{1}{n} D_2\left( \frac{i}{n} \left(1 - \frac{i}{n}\right) f_{i}(t) \right) + \Lambda w_G f_i(t) \left( G{FS}\left(\frac{i}{n}\right)  - \ds\sum_{j=0}^N G_{FS}\left(\frac{j}{n}\right) f_j(t) \right) + w_I \left[ b_F D_1^+ \left( \frac{i}{n} \left(1 - \frac{i}{n} \right) f_i(t) \right) - b_S D_1^- \left( \frac{i}{n}  \left(1 - \frac{i}{n} \right) f_{i}(t) \right)  \right] + \frac{o(\Delta t)}{\Delta t}.
\end{dmath}
Furthermore, in a limit as $\Delta t \to 0$, we obtain the following ODE for the evolution of $f^n_i(t)$
\begin{dmath} \label{eq:fniODE}
    \dsddt{f^n_i(t)} = \frac{1}{n} D_2\left( \frac{i}{n} \left(1 - \frac{i}{n}\right) f_{i}(t) \right) + \Lambda w_G f_i(t) \left( G_{FS}\left(\frac{i}{n}\right)  - \ds\sum_{j=0}^N G_{FS}\left(\frac{j}{n}\right) f_j(t) \right) + w_I \left[ b_F D_1^+ \left( \frac{i}{n} \left(1 - \frac{i}{n} \right) f_i(t) \right) - b_S D_1^- \left( \frac{i}{n}  \left(1 - \frac{i}{n} \right) f_{i}(t) \right)  \right]. 
\end{dmath}
The system of ODEs characterized by Equation \eqref{eq:fniODE} for $i \in \{1,\cdots,n-1\}$ can be paired with the following differential equations derived for $f^n_0(t)$ and $f^n_n(t)$, take into effect the absorbing nature of the boundary at the all-slow and all-fast compositions,
\begin{subequations} \label{eq:fniboundary}
\begin{align}
    \dsddt{f^n_0(t)} &= \left(\frac{ n - 1}{n} \right) \left( 1 + w_I b_F \right) f^n_1(t) + \Lambda w_G f^n_0(t) \left(G_{FS}(0)- \sum_{j=0}^n G_{FS} \left(\frac{j}{n} \right) f^n_j(t) \right) \\
    \dsddt{f^n_n(t)} &= \left(\frac{ n - 1}{n} \right) \left( 1 + w_I b_S \right) f^n_{n-1}(t) + \Lambda w_G f^n_n(t) \left(G_{FS}(1) - \sum_{j=0}^n G_{FS} \left(\frac{j}{n} \right)  f^n_j(t) \right).
\end{align}
\end{subequations}

Solutions to Equation \eqref{eq:fniODE} and \eqref{eq:fniboundary} can be used to understand the dynamics of multilevel selection for the our baseline model when there are infinitely many protocells with a finite number $n$ of genes per protocell. 

Finally, we can further take the limit as the number of genes per protocell tends to infinity. In this limit, we will describe the fraction of slow replicators in a protocell by $x$ and the distribution of groups with $x$ slow replicators (and $1-x$ fast replicators) at time $t$ by the probability density $f(t,x)$. Taking the limit on both sides of Equation \eqref{eq:fniODE} as $n \to \infty$ (and correspondingly $\frac{i}{n} \to x$ and $f_i^n(t) \to f(t,x)$), we can use the difference quotients from Equation \eqref{eq:dimorphicfirst} and \eqref{eq:dimorphicsecond} to see the distribution of protocell compositions evolves according to the following PDE
\begin{dmath} \label{eq:protocellPDEprelim}
\dsdel{f(t,x}{t} = w_I \left(b_F - b_S\right) \dsdel{}{x} \left[x (1-x) f(t,x)\right]  + \Lambda w_G f(t,x) \left[G_{FS}(x) - \int_0^1 G_{FS}(y) f(t,y) dy \right].
\end{dmath}
In particular, we notice that the impact of neutral birth-death events encoded by the second-order difference quotient $D_2(\cdot)$ present in Equation \eqref{eq:fniODE} vanishes in the limit as $n \to \infty$. Using Equation \eqref{eq:dimorphicsecond}, we can see that the large-$n$ limit of the second order difference quotient corresponds to the Kimura diffusion operator \cite{kimura1955solution,kimura1984evolution} given by
\begin{equation}
    \ds\lim_{n \to \infty} D_2\left(u(t,\frac{i}{n})\right) = \ds\lim_{n \to \infty} \frac{u(t,\frac{i+1}{n}) - 2 u(t,\frac{i}{n}) + u (t,\frac{i-1}{n})}{\frac{1}{n^2}} = \frac{\partial^2}{\partial x^2} \left[ x(1-x) u(t,x) \right].
\end{equation}
However, the factor of $\frac{1}{n}$ multiplying the second-order difference quotient in Equation \eqref{eq:fniODE} guarantees that $\frac{1}{n} D_2(\frac{i}{n} (1 - \frac{i}{n}) f_i^n(t)) \to 0$ as $n \to \infty$. For other possible scalings of the number of protocells $m$ or the number of genes per protocell $n$ as $m,n \to \infty$ can allow for either deterministic diffusive effects \cite{velleret2019two} or a Fleming-Viot stochastic process \cite{luo2017scaling} as other possible large population limits of our two-level stochastic process. 

Returning to the PDE description of Equation \eqref{eq:protocellPDEprelim}, we can further clarify the relative strength of competition at the protocell levels by dividing both sides of the equation by $w_I$ and rescaling time as $\tau := \frac{t}{w_I}$. Introducing the parameter $\lambda := \frac{\Lambda w_G}{w_I}$, we can that the density $\rho(\tau,x)$ evolves according to the PDE
\begin{equation} \label{eq:baselineprotocellappendix}
    \dsdel{f(\tau,x)}{\tau} = \left(b_F - b_S \right) \dsdel{}{x} \left[ x (1-x) f(\tau,x) \right] + \lambda f(\tau,x) \left[G_{FS}(x) - \int_0^1 G_{FS}(y) f(\tau,y) dy \right].
\end{equation}
By plugging in the gene-level replication rates $b_S = 1$ and $b_F = 1 + s$ and replacing $\tau$ with $t$, we see that Equation \eqref{eq:baselineprotocellappendix} takes the form of the baseline protocell PDE of Equation \eqref{eq:protocellPDEG} studied in Section \ref{sec:protocell}. 

\begin{remark}
The form of the relative selection strength $\lambda = \frac{\Lambda w_G}{w_I}$ suggests that there are multiple routes to having gene-level or protocell-level competition that is dominant under the multilevel dynamics. In particular, having relatively weak between-protocell competition (small $\lambda$) can occur because of strong gene-level selection (large $w_I$) or because between-group competition replication events are rare (small $\Lambda$) or have weak dependence on protocell-level competition (small $w_I$). Similarly, this means that the limit of strong relative between-protocell competition ($\lambda \to \infty$) could arise as the result of extremely weak gene-level selection ($w_I \to 0$), and so the large-$\lambda$ can also be thought of as describing the dynamics of multilevel selection in the limit as gene-level competition becomes neutral. 
\end{remark}

\subsection{Derivation of Trimorphic (Fast-Slow-Dimer) Protocell Model}
\label{sec:derivationtrimorphic}

For a protocell consisting of $i$ slow replicators, $j$ fast replicators, and $n - i -j$ dimer replicators, we denote the state of the protocell by $(i,j)$ and the fraction of the $m$ total protocells with this composition at time $t$ by $\rho_{i,j}^{m,n}(t)$. To derive the behavior in our PDE limit, we will focus on group compositions satisfying $i,j, n-i-j \geq 1$, which describe the dynamics for compositions featuring all three types of replicators. %
For such compositions, we now study the gene-level dynamics that resembles a continuous-time Moran process with three types (fast, slow, and dimer replicators). We assume that slow, fast, and dimer replicators produce copies of themselves with rates $1 + w_I b_S$, $1 + w_I b_F$, and $1 + w_I b_D$, respectively, and that the copy replaces a randomly chosen gene within the same protocell. In particular, we see that $\rho_{i,j}^{m,n}(t)$ increases by $\frac{1}{m}$ if one of the following six events occurs
\begin{itemize}
    \item A slow is born and replaces a dimer in an $(i-1,j)$-protocell, which occurs with rate \vspace{-2mm}
\[m \rho_{i-1,j}^{m,n}(t) \left(i - 1\right) \left( 1 + w_I b_S\right) \left( \frac{n - i -j + 1}{n} \right)\]
    \item  \vspace{-3mm} A slow is born and replaces a fast in an $(i-1,j+1)$-protocell, which occurs with rate \vspace{-2mm} 
    \[m \rho_{i-1,j+1}^{m,n}(t) \left(i-1\right) \left( 1 + w_I b_S \right) \left( \frac{j+1}{n} \right)\]
    \item \vspace{-3mm} A dimer is born and replaces a slow in an $(i+1,j)$-protocell, which occurs with rate \vspace{-2mm} 
    \[ m \rho_{i+1,j}^{m,n}(t) \left( n - i - j - 1\right) \left( 1 + w_I b_D \right) \left( \frac{i+1}{n} \right) \]
     \item \vspace{-3mm} A fast is born and replaces a slow in an $(i+1,j-1)$-protocell, which occurs with rate \vspace{-2mm} 
    \[ m \rho_{i+1,j-1}^{m,n}(t) \left( j - 1\right) \left( 1 + w_I b_F \right) \left( \frac{i+1}{n} \right) \]
    \item \vspace{-3mm} A fast is born and replaces a slow in an $(i,j-1)$-protocell, which occurs with rate \vspace{-2mm} 
    \[ m \rho_{i,j-1}^{m,n}(t) \left( j - 1\right) \left( 1 + w_I b_F \right) \left( \frac{n-i-j+1}{n} \right) \]
    \item \vspace{-3mm} A dimer is born and replaces a fast in an $(i,j+1)$-protocell, which occurs with rate \vspace{-2mm} 
    \[ m \rho_{i,j+1}^{m,n}(t) \left( n - i - j + 1\right) \left( 1 + w_I b_D \right) \left( \frac{j+1}{n} \right). \]
\end{itemize}
For interior protocell compositions, $\rho_{i,j}^{m,n}(t)$ decreases by $\frac{1}{m}$ if one of the following six events happens
\begin{itemize}
    \item \vspace{-3mm} A slow is born and replaces a fast in an $(i,j)$-protocell, which occurs with rate \vspace{-2mm} 
    \[ m \rho_{i,j}^{m,n}(t) i \left(1 + w_I b_S \right) \left(\frac{j}{n} \right)  \]
     \item \vspace{-3mm} A slow is born and replaces a dimer in an $(i,j)$-protocell, which occurs with rate \vspace{-2mm} 
    \[ m \rho_{i,j}^{m,n}(t) i \left(1 + w_I b_S \right) \left(\frac{n-i-j}{n} \right)  \]
     \item \vspace{-3mm} A fast is born and replaces a slow in an $(i,j)$-protocell, which occurs with rate \vspace{-2mm} 
    \[ m \rho_{i,j}^{m,n}(t) j \left(1 + w_I b_F \right) \left(\frac{i}{n} \right)  \]
     \item \vspace{-3mm} A fast is born and replaces a dimer in an $(i,j)$-protocell, which occurs with rate \vspace{-2mm} 
    \[ m \rho_{i,j}^{m,n}(t) j \left(1 + w_I b_F \right) \left(\frac{n-i-j}{n} \right)  \]
    \item \vspace{-3mm} A dimer is born and replaces a slow in an $(i,j)$-protocell, which occurs with rate \vspace{-2mm} 
    \[ m \rho_{i,j}^{m,n}(t) \left( n - i - j \right) \left(1 + w_I b_D \right) \left(\frac{i}{n} \right)  \]
    \item \vspace{-3mm} A dimer is born and replaces a fast in an $(i,j)$-protocell, which occurs with rate \vspace{-2mm} 
    \[ m \rho_{i,j}^{m,n}(t) \left(n-i-j\right) \left(1 + w_I b_D \right) \left(\frac{j}{n} \right).  \]
\end{itemize}
Finally, the fraction of protocells $\rho_{i,j}^{m,n}(t)$ with state $(i,j)$ is unchanged by gene-level birth-death events when the replicator that is born replaces a replicator with the same type.

Now we turn to between-protocell competition, which is modeled as a continuous-time Moran process featuring all of the different possible protocell compositions. We assume that a protocell featuring fractions $x$ slow replicators, $y$ fast replicators, and $1-x-y$ dimer replicators produces a copy of itself with rate $\Lambda \left( 1 + w_I G(x,y) \right)$, with this offspring protocell replacing a randomly chosen protocell. 
Under this rule for protocell-level replication, we see that the fraction of $(i,j)$-protocells $\rho_{i,j}^{m,n}(t)$ increases by $\frac{1}{m}$ due to between-protocell competition when an $(i,j)$-protocell is chosen to reproduce and a protocell with a composition other than $(i,j)$ is chosen to be replaces. This event occurs with rate 
\[ \Lambda m \rho_{i,j}^{m,n}(t) \left[ 1 +  w_G  G_{FS}\left(\tfrac{i}{n},\frac{j}{n}\right) \right] \left[1 - \rho_{i,j}^{m,n}(t) \right].  \]
Similarly, $\rho_{i,j}^{m,n}(t)$ decreases by $\frac{1}{m}$ when a protocell with composition other than $(i,j)$ is chosen to reproduce and an $(i,j)$-protocell is chosen to be replaced. Such an event occurs with rate
\[ \Lambda m \rho_{i,j}^{m,n}(t)  \ds\sum_{\substack{k=1 \\ k \ne i}}^{n} \ds\sum_{\substack{l = 1 \\ l \ne j}}^{n - i}  \rho_{k,l}^{m,n}(t) \left[ 1 + w_G G\left(\tfrac{k}{n},\frac{l}{n} \right) \right]  \]

As in the derivation from Section \ref{sec:derivationbaseline} , we can use these rates to calculate the infinitessimal mean of the two-level birth-death process, which yields
\begin{align} \label{eq:infmean}
\begin{split}
M_{\Delta t} &:= E [\rho_{i,j}^{m,n}(t + \Delta t) - \rho_{i,j}^{m,n}(t)] \\
 &= \frac{1}{m} P\left(\rho_{i,j}^{m,n}(t + \Delta t) - \rho_{i,j}^{m,n}(t) = \frac{1}{m} \right) - \frac{1}{m} P\left(\rho_{i,j}^{m,n}(t + \Delta t) - \rho_{i,j}^{m,n}(t) = -\frac{1}{m} \right) + o\left(\Delta t\right)\\
 &= \frac{1}{m} \left[m \rho_{i-1,j}^{m,n}(t) \left(i - 1\right) \left( 1 + w_I b_S\right) \left( \frac{n - i -j + 1}{n} \right) \right] \Delta t \\ 
 &+ \frac{1}{m} \left[ m \rho_{i-1,j+1}^{m,n}(t) \left(i-1\right) \left( 1 + w_I b_S \right) \left( \frac{j+1}{n} \right) \right] \Delta t \\
 &+ \frac{1}{m} \left[m \rho_{i+1,j}^{m,n}(t) \left( n - i - j - 1\right) \left( 1 + w_I b_D \right) \left( \frac{i+1}{n} \right) \right] \Delta t \\
  &+ \frac{1}{m} \left[m \rho_{i+1,j-1}^{m,n}(t) \left( j - 1\right) \left( 1 + w_I b_F \right) \left( \frac{i+1}{n} \right) \right] \Delta t \\
   &+ \frac{1}{m} \left[ m \rho_{i,j-1}^{m,n}(t) \left( j - 1\right) \left( 1 + w_I b_F \right) \left( \frac{n-i-j+1}{n} \right)\right] \Delta t \\
    &+ \frac{1}{m} \left[m \rho_{i,j+1}^{m,n}(t) \left( n - i - j + 1\right) \left( 1 + w_I b_D \right) \left( \frac{j+1}{n} \right) \right] \Delta t \\
    &- \frac{1}{m} \left[m \rho_{i,j}^{m,n}(t) i \left(1 + w_I b_S \right) \left(\frac{j}{n} \right) \right] \Delta t 
      - \frac{1}{m} \left[m \rho_{i,j}^{m,n}(t) i \left(1 + w_I b_S \right) \left(\frac{n-i-j}{n} \right) \right] \Delta t \\
     &- \frac{1}{m} \left[m \rho_{i,j}^{m,n}(t) j \left(1 + w_I b_F \right) \left(\frac{i}{n} \right) \right] \Delta t 
      - \frac{1}{m} \left[m \rho_{i,j}^{m,n}(t) j \left(1 + w_I b_F \right) \left(\frac{n-i-j}{n} \right) \right] \Delta t \\
       &- \frac{1}{m} \left[m \rho_{i,j}^{m,n}(t) \left( n - i - j \right) \left(1 + w_I b_D \right) \left(\frac{i}{n} \right) \right] \Delta t - \frac{1}{m} \left[m \rho_{i,j}^{m,n}(t) \left(n-i-j\right) \left(1 + w_I b_D \right) \left(\frac{j}{n} \right) \right] \Delta t \\
       &+ \frac{1}{m} \left\{\Lambda m \rho_{i,j}^{m,n}(t) \left[ 1 +  w_G  G\left(\frac{i}{n},\frac{j}{n}\right) \right] \left[1 - \rho_{i,j}^{m,n}(t) \right] \right\} \Delta t \\ 
       &-\frac{1}{m} \left\{ \Lambda m \rho_{i,j}^{m,n}(t)  \ds\sum_{\substack{k=1 \\ k \ne i}}^{n} \ds\sum_{\substack{l = 1 \\ l \ne j}}^{n - i}  \rho_{k,l}^{m,n}(t) \left[ 1 + w_G G\left(\frac{k}{n},\frac{l}{n} \right) \right] \right\} \Delta t + o\left( \Delta t \right). 
       \end{split}
\end{align}
To further study the infinitessimal mean, we can rearrange Equation \eqref{eq:infmean} to group terms based upon the type of birth-death events. Simplifying the terms describing between-protocell competition and denoting by $C_{N}(m,n)$, $C_{S}(m,n)$, $C_{F}(m,n)$, and $C_D(m,n)$ the contributions to infinitessimal mean due to within-protocell events driven by neutral births and selective births of slow, fast, and dimer replicators, we can rewrite Equation \eqref{eq:infmean} as
\begin{dmath} \label{eq:infmeanA}
\frac{E [\rho_{i,j}^{m,n}(t + \Delta t) - \rho_{i,j}^{m,n}(t)]}{\Delta t} = C_N(m,n) + w_I \left[b_S C_S(m,n) + b_F C_F(m,n) + b_D  C_D(m,n) \right] + \Lambda w_G \rho_{i,j}^{m,n}(t) \left[ G\left(\tfrac{i}{n},\tfrac{j,}{n} \right) - \ds\sum_{\substack{k=1 \\ k \ne i}}^{n} \ds\sum_{\substack{l = 1 \\ l \ne j}}^{n - i}  \rho_{k,l}^{m,n}(t) G\left(\tfrac{k}{n},\tfrac{l}{n} \right)  \right] + \frac{o\left(\Delta t\right)}{\Delta t},
\end{dmath}
where 
\begin{subequations} \label{eq:Cexpressions}
     \begin{align} \label{eq:CN}
     \begin{split}
      \frac{C_N(m,n)}{n} := &  \left(\frac{i-1}{n} \right) \left(\frac{n-i-j}{n} \right) \rho_{i-1,j}^{m,n}(t) +  \left(\frac{i-1}{n}\right) \left(\frac{j+1}{n} \right) \rho_{i-1,j+1}^{m,n}(t) \\ 
      + & \left( \frac{i+1}{n} \right) \left(\frac{j-1}{n} \right) \rho_{i+1,j-1}^{m,n}(t) + \left( \frac{n-i-j+1}{n} \right) \left( \frac{j-1}{n} \right) \rho_{i,j-1}^{m,n}(t) \\
      + & \left(\frac{n-i-j-1}{n} \right) \left( \frac{i+1}{n} \right) \rho_{i+1,j}^{m,n}(t) + \left(\frac{n-i-j-1}{n} \right) \left(\frac{j+1}{n} \right) \rho_{i,j+1}^{m,n}(t) \\
      -& 2 \left[\left( \frac{i}{n} \right) \left(\frac{j}{n} \right) + \left( \frac{i}{n} \right) \left(\frac{n-i-j}{n} \right) + \left( \frac{j}{n} \right) \left(\frac{n-i-j}{n} \right) \right] \rho_{i,j}^{m,n}(t)
      \end{split} 
     \end{align}
     \begin{align} \label{eq:CS}
     \begin{split}
         \frac{C_S(m,n)}{n} := & \left(\frac{i-1}{n} \right) \left(\frac{n-i-j+1}{n} \right) \rho_{i-1,j}^{m,n}(t) + \left( \frac{i-1}{n} \right) \left( \frac{j+1}{n} \right) \rho_{i-1,j+1}^{m,n}(t) \\
         -&  \left( \frac{j}{n} \right) \left(\frac{n-i-j}{n} \right)\rho_{i,j}^{m,n}(t) - \left( \frac{n-i-j}{n}\right) \left( \frac{i}{n} \right) \rho_{i,j}^{m,n}(t)
         \end{split}
     \end{align}
    \begin{align} \label{eq:CF}
        \begin{split}
            \frac{C_F(m,n)}{n} := & \left( \frac{i+1}{n} \right) \left(\frac{j-1}{n} \right) \rho_{i+1,j-1}^{m,n}(t) + \left( \frac{n-i-j+1}{n} \right) \left( \frac{j-1}{n} \right) \rho_{i,j-1}^{m,n}(t) \\
            -&  \left( \frac{j}{n} \right) \left(\frac{i}{n} \right) \rho_{i,j}^{m,n}(t) - \left( \frac{j}{n} \right) \left(\frac{n-i-j}{n} \right) \rho_{i,j}^{m,n}(t)
        \end{split}
    \end{align}
    \begin{align} \label{eq:CD}
        \begin{split}
        \frac{C_D(m,n)}{n} := & \left(\frac{n-i-j-1}{n} \right) \left( \frac{i+1}{n} \right) \rho_{i+1,j}^{m,n}(t) + \left(\frac{n-i-j-1}{n} \right) \left(\frac{j+1}{n} \right) \rho_{i,j+1}^{m,n}(t) \\
        - & \left( \frac{n-i-j}{n}\right) \left( \frac{i}{n} \right) \rho_{i,j}^{m,n}(t) - \left(\frac{n-i-j}{n} \right) \left( \frac{j}{n} \right) \rho_{i,j}^{m,n}(t).
        \end{split}
    \end{align}
\end{subequations}
Furthermore, we can take the limit of both sides of Equation \eqref{eq:infmeanA} as $\Delta t \to 0$ to obtain
\begin{dmath} \label{eq:infmeanAdeltat}
\ds\lim_{\Delta t \to 0} \frac{E [\rho_{i,j}^{m,n}(t + \Delta t) - \rho_{i,j}^{m,n}(t)]}{\Delta t} = C_N(m,n) + w_I \left[b_S C_S(m,n) + b_F C_F(m,n) + b_D C_D(m,n) \right] + \Lambda w_G \rho_{i,j}^{m,n}(t) \Bigg[ G\left(\tfrac{i}{n},\tfrac{j,}{n} \right) - \ds\sum_{\substack{k=1 \\ k \ne i}}^{n} \ds\sum_{\substack{l = 1 \\ l \ne j}}^{n - i}  \rho_{k,l}^{m,n}(t) G\left(\tfrac{k}{n},\tfrac{l}{n} \right)  \Bigg].
\end{dmath}
As in Section \ref{sec:derivationbaseline}, we can calculate the infinitessimal variance of the two-level birth death process. Noting that all of the transition rates are linear in $m$, we see that 
\begin{align}
\begin{split}
V_{\Delta t} := & E \left[ \left( \rho_{i,j}^{m,n}(t + \Delta t) - \rho_{i,j}^{m,n}(t) \right)^2 \right] \\ 
 = & \frac{1}{m^2} P\left(\rho_{i,j}^{m,n}(t + \Delta t) - \rho_{i,j}^{m,n}(t) = \frac{1}{m} \right) - \frac{1}{m^2} P\left(\rho_{i,j}^{m,n}(t + \Delta t) - \rho_{i,j}^{m,n}(t) = \frac{1}{m} \right) \\ 
 = & \frac{1}{m^2} \left[O(m) \right] + o(\Delta t).
\end{split}
\end{align}
Therefore we can deduce that 
\begin{equation}
    \ds\lim_{m \to \infty} E \left[ \left( \rho_{i,j}^{m,n}(t + \Delta t) - \rho_{i,j}^{m,n}(t) \right)^2 \right]  = 0,
\end{equation}
and the infinitessimal variance vanishes as $m \to \infty$. Therefore, in this limit, the distribution of $\rho_{i,j}^n(t) := \lim_{m \to \infty} \rho_{i,j}^{m,n}(t)$ is a constant equal to its mean $E[\rho_{i,j}^n(t)]$. This means that taking the limit of infinitely many protocells removes the randomness from the multilevel dynamics, and we can look to describe the evolution of the fraction of protocells $\rho_{i,j}^n(t)$ whose $n$ replicators have the composition $(i,j)$ through a system of ODEs. 

Noting that the righthand side of Equation \eqref{eq:infmeanAdeltat} is independent of $m$, we can take the limit of both sides of the equation and use the fact that $E[\rho_{i,j}^n(t+\Delta t) - \rho_{i,j}^n(t)] = \rho_{i,j}^n(t+\Delta t) - \rho_{i,j}^n(t)$ to obtain the following system of ODEs 
\begin{dmath} \label{eq:ODEA}
\dsddt{\rho_{i,j}^n(t)} = C_N(n) + w_I \left[b_S C_S(n) + b_F C_F(n) + b_D C_D(n) \right] + \Lambda w_G \rho_{i,j}^n(t) \Bigg[ G\left(\tfrac{i}{n},\tfrac{j,}{n} \right) - \ds\sum_{\substack{k=1 \\ k \ne i}}^{n} \ds\sum_{\substack{l = 1 \\ l \ne j}}^{n - i}  \rho_{k,l}^n(t) G\left(\tfrac{k}{n},\tfrac{l}{n} \right)  \Bigg],
\end{dmath}
\sloppy{where the expressions for $C_N(n) = \lim_{m \to \infty} C_N(m,n)$, $C_S(n) = \lim_{m \to \infty} C_S(m,n)$, $C_F(n) = \lim_{m \to \infty} C_F(m,n)$, $C_D(n) = \lim_{m \to \infty} C_D(m,n)$ can be obtained by taking the limit of Equation \eqref{eq:Cexpressions} as $m \to \infty$.}

To derive the large population limits of our nested birth-death process, we look to express the infinitessimal mean in terms of difference quotients taken in both the $x$ and $y$ directions. In particular, we may find forward first-order difference quotients in the $x$ and $y$ direction
\begin{equation} \label{eq:forwarddiff} D_{1,x}^+\left(u(\tfrac{i}{n},\tfrac{j}{n})\right) := \frac{u(\tfrac{i+1}{n},\tfrac{j}{n}) - u(\tfrac{i}{n},\tfrac{j}{n})}{\frac{1}{n}} \: \:, \: \: D_{1,y}^+\left(u(\tfrac{i}{n},\tfrac{j}{n})\right) := \frac{u(\tfrac{i}{n},\tfrac{j+1}{n}) - u(\tfrac{i}{n},\tfrac{j}{n})}{\frac{1}{n}}, \end{equation}
backward first-order difference quotients in the $x$ and $y$ directions
\begin{equation} \label{eq:backwarddiff} D_{1,x}^-\left(u(\tfrac{i}{n},\tfrac{j}{n})\right) := \frac{u(\tfrac{i}{n},\tfrac{j}{n}) - u(\tfrac{i-1}{n},\tfrac{j}{n})}{\frac{1}{n}} \: \:, \: \:  D_{1,y}^-\left(u(\tfrac{i}{n},\tfrac{j}{n})\right) := \frac{u(\tfrac{i}{n},\tfrac{j}{n}) - u(\tfrac{i}{n},\tfrac{j-1}{n})}{\frac{1}{n}}. \end{equation}
We can now use Equations \eqref{eq:CS}, \eqref{eq:forwarddiff}, and \eqref{eq:backwarddiff} to see that $C_S(n)$ can be written as
\begin{subequations} \label{eq:Cfirstorderdiff}
\begin{align} \label{eq:CSdiff}
\begin{split}
  C_S(n) &= - n \left[ \left( \frac{n-i-j}{n} \right) \left(\frac{i}{n} \right) \rho_{i,j}^n(t) - \left( \frac{n-i-j+1}{n} \right) \left( \frac{i-1}{n} \right) \rho_{i-1,j}^n(t) \right] \\ &+  n \left(\frac{i-1}{n} \right) \left[\left(\frac{j+1}{n} \right) \rho_{i-1,j+1}^n(t) - \left(\frac{j}{n} \right) \rho_{i-1,j}^n \right] \\
  &- n \left(\frac{j}{n} \right) \left[\left(\frac{i}{n}\right) \rho_{i,j}^n(t) -\left( \frac{i-1}{n} \right) \rho_{i-1,j}^n(t) \right] \\ 
  &= - D_{1,x}^{-}\left(\left[ \frac{n-i-j}{n} \right]\frac{i}{n} \rho_{i,j}^n(t) \right)  + \left(\frac{i-1}{n} \right) D_{1,y}^+\left(\frac{j}{n} \rho_{i-1,j}^n(t) \right) - \left(\frac{j}{n} \right) D_{1,x}^{-}\left( \frac{i}{n} \rho_{i,j}^n(t)\right) \\ 
  &= -  D_{1,x}^{-}\left(\left[\frac{n-i}{n} \right] \frac{i}{n} \rho_{i,j}^n(t)\right) + \left(\frac{i-1}{n} \right) D_{1,y}^+\left(\frac{j}{n} \rho_{i-1,j}^n(t) \right).
 \end{split}
\end{align}

Using a similar approach, we can obtain the following expressions for $C_F(n)$ and $C_D(n)$ using first-order difference quotients
\begin{align} \label{eq:CFdiff}
\begin{split}
C_F(n) = \left(\frac{j-1}{n} \right) D_{1,x}^+\left(\frac{i}{n} \rho_{i,j-1}^n(t) \right) - D_{1,y}^-\left(\left[\frac{n-j}{n}\right]\frac{j}{n} \rho_{i,j}^n(t) \right) %
\end{split}
\end{align}
\begin{align} \label{eq:CDdiff}
\begin{split}
C_D(n) = D_{1,x}^+\left(\left[\frac{n-i-j}{n} \right] \left[\frac{i}{n}\right]  \rho_{i,j}^n(t)\right) + D_{1,y}^+\left(\left[\frac{n-i-j}{n} \right] \left[\frac{j}{n}\right]  \rho_{i,j}^n(t)\right).
\end{split}
\end{align}
\end{subequations}

To further understand the impact of neutral within-protocell birth events (as described by $C_N$), we will also have to introduce second-order difference quotients. We use central second-order difference quotients in the $x$ and $y$ directions
\begin{subequations} \label{eq:secondorderdifference}
\begin{align} D_{2,xx}^{c}\left(u(\tfrac{i}{n},\tfrac{j}{n})\right) &:= \frac{u(\tfrac{i-1}{n},\tfrac{j}{n}) - 2 u(\tfrac{i}{n},\tfrac{j}{n}) + u(\tfrac{i+1}{n},\tfrac{j}{n})}{\frac{1}{n^2}} \\ D_{2,yy}^{c}\left(u(\tfrac{i}{n},\tfrac{j}{n})\right) &:= \frac{u(\tfrac{i}{n},\tfrac{j-1}{n}) - 2 u(\tfrac{i}{n},\tfrac{j}{n}) + u(\tfrac{i}{n},\tfrac{j+1}{n})}{\frac{1}{n^2}},
\end{align}
\end{subequations}
and also consider a mixed partial derivative taking the form
\begin{equation} \label{eq:mixedpartialdifference}
D_{1,x}^c \left( D_{1,y}^- \left(u(\tfrac{i}{n},\tfrac{j}{n}) \right) \right) := \frac{u(\tfrac{i+1}{n},\tfrac{j}{n}) - u(\tfrac{i-1}{n},\tfrac{j}{n}) - u(\tfrac{i+1}{n},\tfrac{j-1}{n}) + u(\tfrac{i-1}{n},\tfrac{j-1}{n})}{\frac{2}{n^2}}.  
\end{equation}
The choice of neutral difference in the $x$-direction and backward difference in the $y$-direction arises from the terms of $C_N$. Using Equations \eqref{eq:CN}, \eqref{eq:secondorderdifference}, and \eqref{eq:mixedpartialdifference}, we can write $n C_N(n)$ as 
\begin{align} \label{eq:CNdiff}
    \begin{split}
        n C_N(n) &=  D_{2,xx}^c \left( \left[\frac{i}{n} \right] \left[\frac{n-i-j}{n} \right] \rho_{i,j}^n(t) \right) + D_{2,yy}^c \left( \left[\frac{j}{n} \right] \left[\frac{n-i-j}{n} \right] \rho_{i,j}^n(t) \right) \\ 
        &+  \left(\frac{i-1}{n} \right) \left(\frac{j+1}{n} \right) \rho_{i-1,j+1}^n(t) + \left( \frac{i+1}{n} \right) \left( \frac{j-1}{n}\right) \rho_{i+1,j-1}^n(t) - 2 \left( \frac{i}{n} \right) \left( \frac{j}{n} \right) \rho_{i,j}^n(t)  \\
        &=    D_{2,xx}^c \left( \left[\frac{i}{n} \right] \left[\frac{n-i-j}{n} \right] \rho_{i,j}^n(t) \right) + D_{2,xx}^c \left( \left[\frac{i}{n} \right] \left[\frac{j}{n} \right] \rho_{i,j}^n(t) \right) \\
       &+ D_{2,yy}^c \left( \left[\frac{j}{n} \right] \left[\frac{n-i-j}{n} \right] \rho_{i,j}^n(t) \right) + D_{2,yy}^c \left( \left[\frac{j}{n} \right] \left[\frac{i-1}{n} \right] \rho_{i,j}^n(t) \right) + \left(\frac{i-1}{n} \right) \left(\frac{j}{n} \right) \rho_{i-1,j}^n(t) \\ 
       &- \left(\frac{i-1}{n} \right) \left(\frac{j-1}{n} \right) \rho_{i-1,j-1}^n(t) - \left(\frac{i+1}{n} \right) \left(\frac{j}{n} \right) \rho_{i+1,j}(t) + \left(\frac{i+1}{n} \right) \left(\frac{j-1}{n} \right) \rho_{i+1,j-1}^n(t) \\
       &= D_{2,yy}^c \left( \left[\frac{j}{n} \right] \left[\frac{n-i}{n} \right] \rho_{i,j}^n(t) \right) - 2 D_{1,x}^c \left( D_{1.y}^- \left( \left[\frac{i}{n}\right] \left[ \frac{j}{n}\right] \rho_{i,j}^n(t)  \right) \right) \\
       &+  D_{2,yy}^c \left( \left[\frac{j}{n} \right] \left[\frac{n-i-j}{n} \right] \rho_{i,j}^n(t) \right) + D_{2,yy}^c \left( \left[\frac{j}{n} \right] \left[\frac{i-1}{n} \right] \rho_{i,j}^n(t) \right)
    \end{split}
\end{align}
In the limit as $n \to \infty$ (and correspondingly $\frac{i}{n} \to x$, $\frac{j}{n} \to y$, and $\rho_{i,j}^n(t) \to \rho(t,x,y)$) and we obtain the following limiting expression for $n C_N(n)$
\begin{dmath}  \label{eq:CNlimit}
 \lim_{n \to \infty} n C_N(n) = \frac{\partial^2}{\partial x^2} \left(x(1-x) \rho(t,x,y) \right)  + \frac{\partial^2}{\partial y^2} \left(y(1-y) \rho(t,x,y) \right)  - 2 \frac{\partial^2}{\partial x \partial y} \left(xy \rho(t,x,y) \right),
\end{dmath}
which is the trimorphic version of the Kimura diffusion operator that arises to describe the role of individual-level noise in models of population genetics and evolutionary game theory  \cite{kimura1955solution,kimura1964diffusion,epstein2010wright,epstein2020some,souza2009evolution}. As a result, we can think of the term $C_N(n)$ as describing the diffusive effects caused by the background birth rate of $1$ for each type of replicator under within-protocell competition. 

Having obtained expressions for $C_N(n)$, $C_S(n)$, $C_F(n)$, and $C_D(n)$ in terms of first- and second-order difference quotients, we can now apply the formulas from Equations \eqref{eq:CSdiff}, \eqref{eq:CFdiff}, \eqref{eq:CDdiff}, and \eqref{eq:CNdiff} to Equation \eqref{eq:ODEA} to write our ODE for $\rho_{i,j}^n(t)$ as
\begin{align} \label{eq:infmeandiff}
\begin{split}
 \dsddt{\rho_{i,j}(t)}&= 
  \frac{1}{n} \left[ D_{2,yy}^c \left( \left[\frac{j}{n} \right] \left[\frac{n-i}{n} \right] \rho_{i,j}^n(t) \right) - 2 D_{1,x}^c \left( D_{1.y}^- \left( \left[\frac{i}{n}\right] \left[ \frac{j}{n}\right] \rho_{i,j}^n(t)  \right) \right) \right] \\
  &+ \frac{1}{n} \left[ D_{2,yy}^c \left( \left[\frac{j}{n} \right] \left[\frac{n-i-j}{n} \right] \rho_{i,j}^n(t) \right) + D_{2,yy}^c \left( \left[\frac{j}{n} \right] \left[\frac{i-1}{n} \right] \rho_{i,j}^n(t) \right) \right] \\
  &+ w_I b_S \left[\left(\frac{i-1}{n} \right) D_{1,y}^+\left(\frac{j}{n} \rho_{i-1,j}^n(t) \right) -  D_{1,x}^{-}\left( \left[\frac{n-i}{n} \right]\frac{i}{n} \rho_{i,j}^n(t)\right)  \right] \\
  &+ w_I b_F \left[  \left(\frac{j-1}{n} \right) D_{1,x}^+\left(\frac{i}{n} \rho_{i,j-1}^n(t) \right) - D_{1,y}^-\left[\left[\frac{n-i-j}{n}\right]\frac{j}{n} \rho_{i,j}^n(t) \right) \right] \\
  &+ w_I b_D \left[ D_{1,x}^+\left(\left[\frac{n-i-j}{n} \right] \left[\frac{i}{n}\right]  \rho_{i,j}^n(t)\right) + D_{1,y}^+\left(\left[\frac{n-i-j}{n} \right] \left[\frac{j}{n}\right]  \rho_{i,j}^n(t)\right) \right] \\
  &+ \Lambda w_G \rho_{i,j}^n(t) \Bigg[ G\left(\tfrac{i}{n},\tfrac{j,}{n} \right) - \ds\sum_{\substack{k=1 \\ k \ne i}}^{n} \ds\sum_{\substack{l = 1 \\ l \ne j}}^{n - i}  \rho_{k,l}^n(t) G\left(\tfrac{k}{n},\tfrac{l}{n} \right)   \Bigg]
  \end{split}
\end{align}
We can use this system of ODEs, along with equations derived in a similar manner to describing the evolution of $\rho_{i,j}(t)$ for compositions $(i,j)$ location on the boundary of the three-type simplex to describe the evolution of slow-fast-dimer replicator compositions in a population with $n$ protocells that each contain infinitely many replicators. In future work, we will explore the dynamics of this system to understand how the diffusive effects of the term $C_N(n)$ can impact the possibility of coexistence of fast and slow genes in populations with finitely many protocells. 

Next we look to describe the limit as $n \to \infty$, exploring the multilevel dynamics in an infinite population of protocells each having infinitely many replicators. In this limit, we will describe the fraction of replicators within a given protocell of slow, fast, and dimer, and will describe the distribution of replicator compositions within the population of protocells through the density $\rho(t,x,y)$. Then, taking the limit as $n \to \infty$ on both sides as $n \to \infty$ (and correspondingly $\frac{i}{n} \to x$, $\frac{j}{n} \to y$, and $\rho_{i,j}^n(t) \to \rho(t,x,y)$), we can replace our difference quotients with partial derivatives to obtain the following PDE for the evolution of the composition of the population of protocells
\begin{dmath} \label{eq:triPDEfirst}
\dsdel{\rho(t,x,y)}{t} = w_I b_S \left\{ x \dsdel{}{y} \left[ y \rho(t,x,y)\right] -  \dsdel{}{x} \left[ \left( 1 - x \right) x \rho(t,x,y) \right] \right\} 
+ w_I b_F \left\{ y \dsdel{}{x}\left[x \rho(t,x,y)\right] -  \dsdel{}{y} \left[\left(1 - x - y \right) y \rho(t,x,y) \right] \right\}
+ w_I b_D \left\{ \dsdel{}{x} \left[\left(1-x-y \right) x \rho(t,x,y) \right] + \dsdel{}{y}\left[ \left( 1 - y \right) y \rho(t,x,y) \right] \right\}
+ \lambda \rho(t,x,y) \left[ G(x,y) - \int_0^1 \int_0^{1-x} G(u,v) \rho(t,u,v) dv du \right]
\end{dmath}
We note from Equation \eqref{eq:CNlimit} that the terms with second-order difference quotients arising from the neutral birth-death events vanish in the limit as $n \to \infty$, so the limiting hyperbolic PDE of Equation \eqref{eq:triPDEfirst} describes the deterministic multilevel birth-death dynamics from the protocell model obtained by taking the limit as the number of protocells $m \to \infty$ and then taking the limit as the number of replicators per protocell $n \to \infty$. In future work, it may be of interest to consider alternate scaling limits of $m$ and $n$ that either retains random effects (such as a Fleming-Viot process \cite{luo2017scaling}) or retains a Kimura diffusion operator \cite{velleret2019two,velleret2020individual} in the large-population limit. 

We can further divide both sides by $w_I$, rescale time as $\tau := \frac{t}{w_I}$, introduce the new parameter $\lambda := \frac{w_G}{w_I} \Lambda$, and rearrange terms on the righthand side of Equation \eqref{eq:triPDEfirst} to obtain the following PDE for $\rho(\tau,x,y)$
\begin{dmath} \label{eq:triPDEsecond}
\dsdel{\rho(\tau,x,y)}{\tau} = - \dsdel{}{x} \left[x \left\{ b_S - b_D + \left(b_D - b_S \right) x + \left( b_D - b_F\right) y \right\}  \rho(\tau,x,y) \right] \\
- \dsdel{}{y} \left[ y \left\{ b_F - b_D + \left( b_D - b_S \right) x + \left( b_D - b_F \right) y \right\} \rho(\tau,x,y)\right]  \\
+ \lambda  \rho(\tau,x,y) \left[ G(x,y) - \int_0^1 \int_0^{1-x} G(u,v) \rho(\tau,u,v) dv du \right].
\end{dmath}
By replacing our rescaled time variable $\tau$ with $t$, we see that this becomes Equation \eqref{eq:multileveltrimorphic}, the PDE replicator equation for multilevel selection in the trimorphic protocell introduced in Section \ref{sec:trimorphicformulation}.

\section{Derivation of Finite Volume Discretization} \label{sec:finitevolume}

In this section, we present the finite volumes schemes that are used to compute numerical solutions to the multilevel protocell dynamics both for pairwise competition on the edges of the simplex and for the full fast-slow-dimer trimorphic competition on the simplex. In Section \ref{sec:fvdimorphic}, we present the one-dimensional upwind finite volume scheme used to generate the time-dependent solutions illustrated in Figure \ref{fig:fsfdtrajectorycompare} for the fast-slow and fast-dimer edges of the simplex. In that sub-section, we also show that that the long-time behavior for numerical solutions under this scheme starting from an initial uniform distribution feature good agreement with a family of steady state densities presented in Section \ref{sec:existingresults}. In Section \ref{sec:fvtrimorphic}, we present the derivation of the finite volume scheme for the full trimorphic multilevel dynamics, using upwinding and properties of the gene-level replicator dynamics to solve the dynamics using a relatively simple spatial discretization of the three-type simplex. 

\subsection{Finite Volume Scheme for Two-Type Dynamics}
\label{sec:fvdimorphic}

In this section, we present the finite volume scheme used to generate the trajectories presented in Figure \ref{fig:fsfdtrajectorycompare} comparing the dynamics on the fast-slow and fast-dimer edges of the simplex. Such schemes are used to study hyperbolic PDEs \cite{leveque2002finite}, and have been derived for models of multilevel selection in the case in which within-group and between-group birth rates depend on the payoff of a two-strategy evolutionary game \cite[Section 5.5]{cooney2020pde}. We will compute the states achieved by numerical solutions of our finite volume approximation after a large number of time steps, showing that the states achieved from an initial uniform distribution have good qualitative agreement with the density steady states with H{\"o}lder exponent $\theta = 1$ achieved by the long-time behavior of the corresponding PDE models studied in Sections \ref{sec:protocellongtime} and \ref{sec:fastdimer}. In particular, the ability of the finite volume scheme to reproduce the analytically calculated PDE steady states for initial uniform distributions (with corresponding H{\"o}lder exponent $\theta = 1$ near the all-slow or all-dimer equilibrium) for two-type dynamics provides some heuristic motivation for choosing initial uniform densities as the baseline numerical scenario for the trimorphic fast-slow-dimer dynamics studied in Section \ref{sec:trimorphicnumerics}.

For the dynamics of our baseline protocell model on the fast-slow edge of the simplex, we describe the composition of protocells using the discretized density $\{f_{j}(t)\}_{j \in \{0,\cdot,N-1\}}$, where $f_j(t)$ describes the volume-average $f_j(t) = \left(\frac{1}{x_{j+1} - x_j}\right) \int_0^1 f(t,x) dx$ of the density $f(t,x)$ on the volume $[x_j,x_{j+1}] = [\frac{j}{n},\frac{j+1}{n}]$. Using an upwind finite volume scheme, the discretized density $f_j(t)$ evolves according to the following ODE
\begin{dmath} \label{eq:PDEwithgridFS}
   \dsddt{f_j(t)} = s N \left[  x_{j+1} (1 - x_{j+1}) f_{j+1} - x_{j} (1 - x_{j}) f_j \right] + \lambda f_j \left[G^j_{FS} - \ds\sum_{k=0}^{N-1} G^k_{FS} f_k \right], 
\end{dmath}
where the discretized protocell-level reproduction rates $G^j_{FS}$ correspond to the average of the reproduction rate $G_{FS}(x)$ on thevolume $[x_j,x_{j+1}] = [\frac{j}{N},\frac{j+1}{N}]$, which is given by 
\begin{dmath}
    {G^j_{FS} := \left(\frac{1}{x_{j+1} - x_j}\right) \int_{x_j}^{x_{j+1}} G_{FS}(x) dx = N \int_{\frac{j}{N}}^{\frac{j+1}{N}} x \left( 1 - \eta x \right) dx } = %
    \frac{1}{2} \left( \frac{2j + 1}{N} \right) -  \frac{\eta}{3} \left( \frac{3 j^2 + 3 j + 1}{N^2} \right).
\end{dmath}

In Figure \ref{fig:numericalfvfsdensities}, we provide a comparison between the numerical solutions of Equation \eqref{eq:PDEwithgridFS} after 5000 time steps with step-size $\Delta t = 0.003$ seconds and the analytically computed steady states $p^{\lambda}_{\theta}(x)$ from Equation \eqref{eq:plambdathetaFS}. We see good agreement between the long-time numerical solutions achieved from an initial uniform solution with the PDE steady states with H{\"o}lder exponent $\theta = 1$ near $x=1$, both for the case in which between-protocell competition most favors the all-slow composition ($\eta = \frac{1}{3}$, left) and in which between-protocell competition most favors an interior mix featuring a fraction $x^*_{FS} = \frac{3}{4}$ slow replicators and $1 - x^*_{FS} = \frac{1}{4}$ slow replicators. The choice of H{\"o}lder exponent $\theta = 1$ is motivated by the fact that our finite volume method approximates the density $f(t,x)$ by the piecewise contant density characterized by the values $\{f_j(t)\}_{ j \in \{0,1,\cdots,N-1\}}$, and therefore discretized densities with positive weight $f_{N-1}(t) > 0$ on the volume $[\frac{N-1}{N},1]$ behave like a uniform density in terms of H{\"o}lder exponent near $x=1$ \cite[Section 5.5]{cooney2021long}.

\begin{figure}[htp!]
    \centering
    \includegraphics[width=0.48\textwidth]{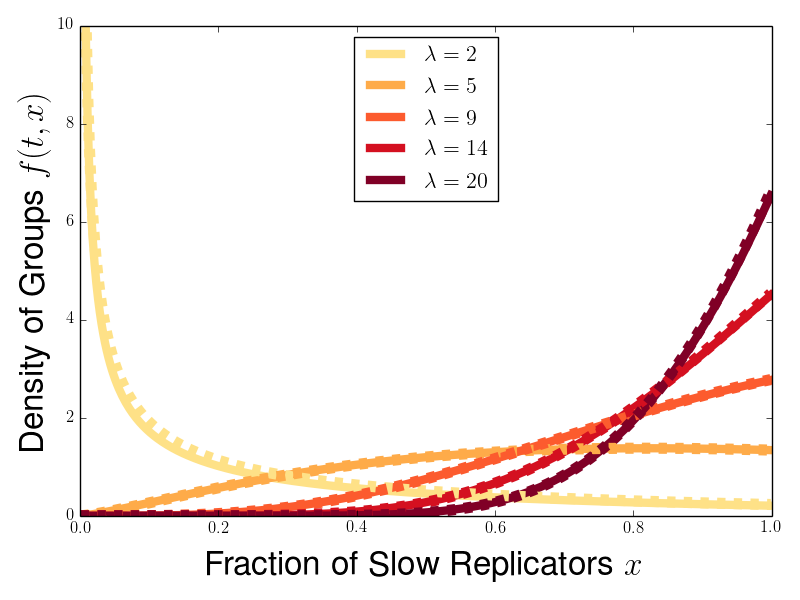}
    \includegraphics[width=0.48\textwidth]{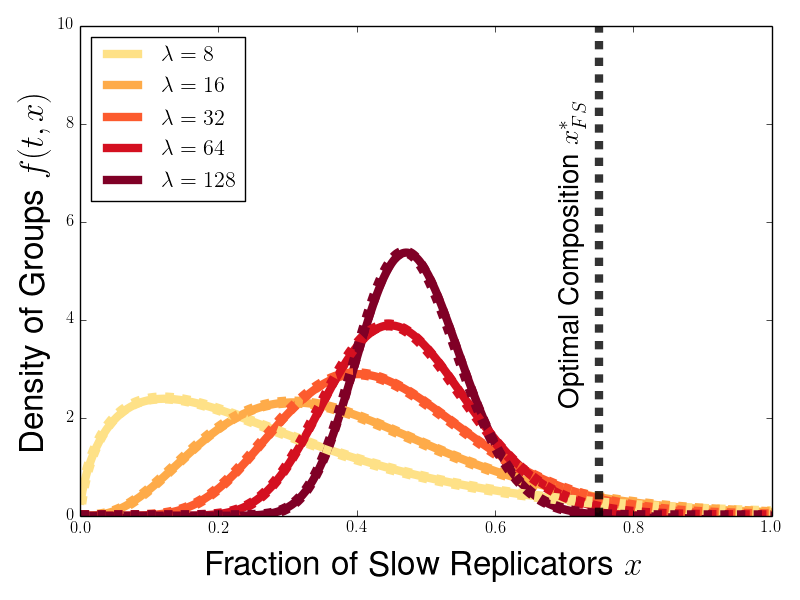}
    \caption{Comparison of numerical solutions for finite volume scheme from Equation \eqref{eq:PDEwithgridFS} after many time steps (solid lines) with steady state densities from Equation \eqref{eq:plambdathetaFS} for the baseline fast-slow protocell model (dashed lines), plotted for $s=1$ and various values of $\lambda$. Numerical solutions $\{f_j(t)\}_{\{j \in 0,\cdots,N-1\}}$ displayed after 5000 time steps of forward Euler method with step-size $\Delta t = 0.003$ starting from uniform initial distribution, while analytical steady states $p^{\lambda}_{\theta}(x)$ are provided for H{\"o}lder exponent $\theta = 1$  near $x=1$. Comparison is provided for complementarity parameters $\eta = \frac{1}{3}$ (left), in which collective reproduction is maximized by all-slow protocell, and $\eta = \frac{2}{3}$ (right), in which protocell-level reproduction is maximized by a mix of 75 percent slow replicators and 25 percent fast replicators. Dashed black vertical line in right panel corresponds to the interior fraction of slow replicators $x^*{FS} = \frac{3}{4}$ which maximizes protocell-level reproduction when $\eta = \frac{2}{3}$.}
    \label{fig:numericalfvfsdensities}
\end{figure}
For the dynamics on the fast-dimer edge of the simplex, we similarly describe the composition of protocells using the discretized density $\{g_{j}(t)\}_{j \in \{0,\cdot,N-1\}}$, through the volume-average $g_j(t) = \left(\frac{1}{z_{j+1} - z_j}\right) \int_0^1 g(t,z) dz$ of the density $g(t,z)$ on the volume $[z_j,z_{j+1}] = [\frac{j}{n},\frac{j+1}{n}]$. In our upwind finite volume method, the $g_j(t)$ evolves according to the following ODE 
\begin{dmath} \label{eq:PDEwithgridFD}
   \dsddt{g_j(t)} = \left(b_F - b_D\right)  N \left[  z_{j+1} (1 - z_{j+1}) g_{j+1} - z_{j} (1 - z_{j}) g_j \right] + \lambda g_j \left[G^j_{FD} - \ds\sum_{k=0}^{N-1} G^k_{FD} g_k \right], 
\end{dmath}
where the discretized protocell-level reproduction rate $G^j_{FD}$ on the fast-dimer edge corresponds to the volume-average of $G_{FD}(z)$ on $[z_j,z_{j+1}] = [\frac{j}{n},\frac{j+1}{n}]$, which is given by
\begin{equation}
    {G^j_{FD} := \left(\frac{1}{z_{j+1} - z_j}\right) \int_{z_j}^{z_{j+1}} G_{FD}(z) dz = N \int_{\frac{j}{N}}^{\frac{j+1}{N}} \frac{1}{2} \left( 1 - \frac{\eta z}{2} \right) dz } = %
    \frac{1}{4} \left( \frac{2j + 1}{N} \right) -  \frac{\eta}{12} \left( \frac{3 j^2 + 3 j + 1}{N^2} \right).
\end{equation}

In Figure \ref{fig:numericalfvfddensities}, we compare the numerical solutions to our finite volume scheme for the fast-dimer competition after a large number of time steps with the $g^{\lambda}_{\theta}(z)$ achieved as the long-time behavior of the fast-dimer PDE dynamics from Equation \eqref{eq:FDmultilevelPDE}. We again see good agreement between the numerical solutions computed from an initial uniform density and the analytically computed steady state with H{\"o}lder exponent $\theta = 1$ near the all-dimer composition $z = 1$. In this figure, we only consider the case $\eta = 1$, because protocell-level competition most favors the all-dimer composition on the fast-dimer edge under any complementarity scenario (corresponding to any $\eta \in [0,1$).

\begin{figure}[ht!]
    \centering
    \includegraphics[width=0.7\textwidth]{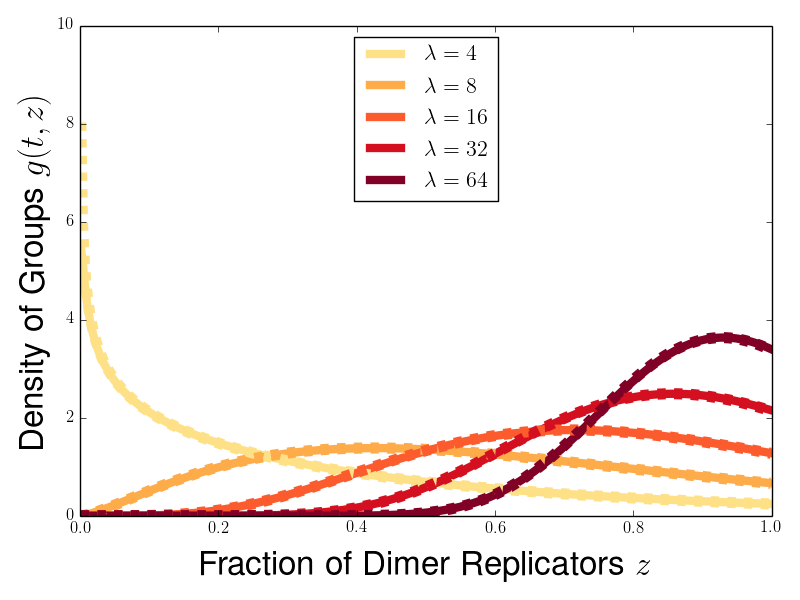}
    \caption{Comparison of numerical solutions for finite volume scheme from Equation \eqref{eq:PDEwithgridFD} after many time steps (solid lines) with steady state densities from Equation \eqref{eq:plambdathetaFS} for the baseline fast-slow protocell model (dashed lines), plotted for $\eta = 1$, $s=1$ and various values of $\lambda$. Numerical solutions $\{g_j(t)\}_{\{j \in 0,\cdots,N-1\}}$ displayed after 5000 time steps of forward Euler method with step-size $\Delta t = 0.003$ starting from uniform initial distribution, while analytical steady states $g^{\lambda}_{\theta}(z)$ are provided for H{\"o}lder exponent $\theta = 1$  near $z=1$.}
    \label{fig:numericalfvfddensities}
\end{figure}

\subsection{Finite Volume Scheme for Three-Type Dynamics}
\label{sec:fvtrimorphic}

As previous work has found qualitative agreement between a two-type finite volume numerical method and analytical predictions for the multilevel dynamics for evolutionary games \cite{cooney2020pde}, we will now extend this finite volume approach to accommodate three types of individuals in groups of fixed size. Because the state space for protocell composition in the slow-fast-dimer system is a three-type simplex, we will adapt our approach to handle within-cell replicator equations on the simplex. 

For a finite volume discretization for the three-type multilevel selection dynamics, we can divide the simplex into cells allowing us to compute flux in the x and y directions. While general two-dimensional domains can require the use of complicated spatialvolumes, we can use a special discretization of the simplex in order to break down the numerical scheme cleanly into fluxes in the x and y direction as is typically implemented for finite-volumes schemes on a 2D rectangular domain (i.e. Levesque \cite{leveque2002finite}). For a given integer $N$, this discretization consists of dividing the simplex into $\frac{N(N-1)}{2}$ squares of side length $\frac{1}{N}$ and $N$ isosceles right triangles with legs of length $\frac{1}{N}$. We illustrate this spatial grid in Figure \ref{fig:simplexvolumes} for side lengths of $\frac{1}{N} = \frac{1}{10}$ (left) and $\frac{1}{N} = \frac{1}{20}$ (right).

\begin{figure}[ht]
\centering
\includegraphics[width = 0.48\linewidth]{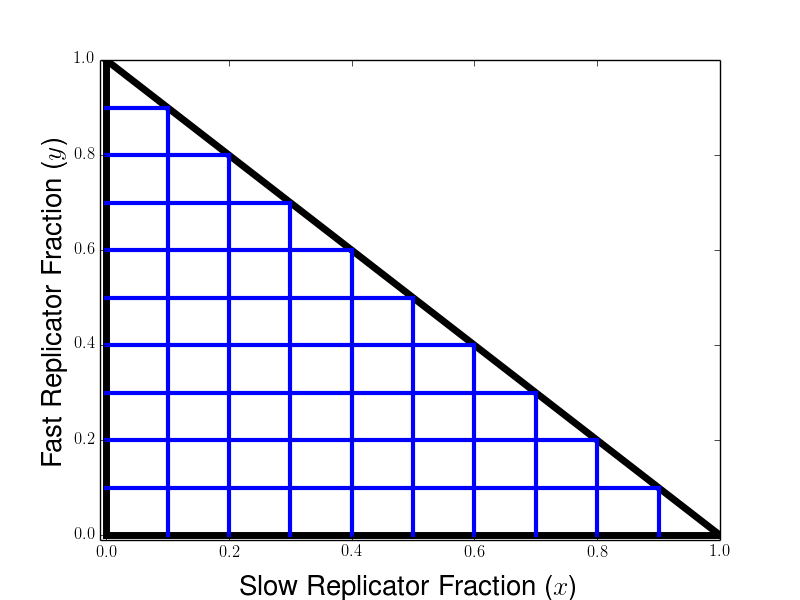}
\includegraphics[width = 0.48\linewidth]{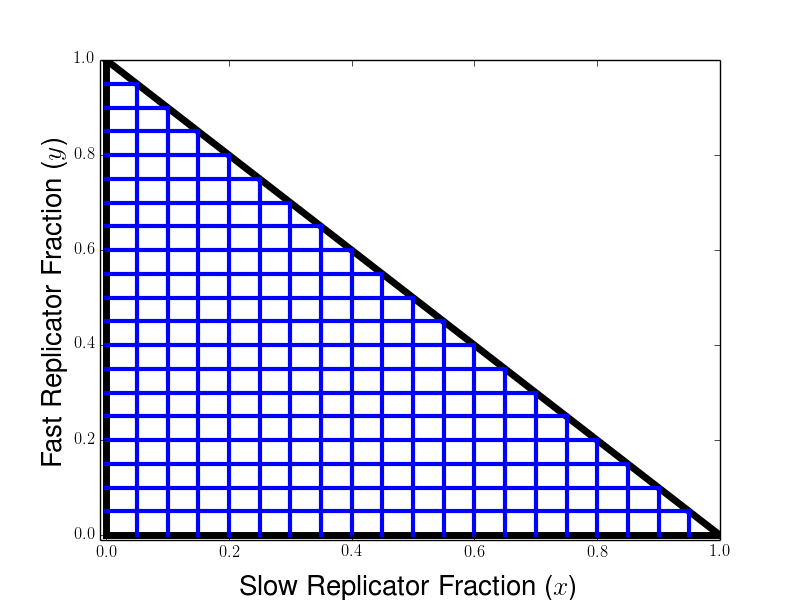}
\caption[Spatial discretization for three-type simplex.]{Spatial discretization for three-type simplex. For a given number $N$, we can divide the unit square into squares of sidelength $1/N$. If we only include the portion of these rectangles that actually fall within the simplex, then there are $N$ triangular volumes with area $\frac{1}{2N^2}$ and $\frac{N^2 - N}{2}$ square volumes with area $\frac{1}{N^2}$. Example grids depicted for $N = 10$ (left) and $N = 20$ (right).}
\label{fig:simplexvolumes}
\end{figure}

Counting from $0$, we label grid points that are $j$th in the $x$-direction and $k$th in the $y$-direction by $(x_j,y_k)$. Then, denoting by $C_{j,k}$ the volume with $(x_j,y_k)$ as its bottom-left corner, we label the edges of square volumes with $u_1$ through $u_4$ and label the edges of $v_1$ through $v_3$ for triangular volumes by starting at the bottom-left and proceeding clockwise. We also note that volumes with $(x_j,y_k)$ as its bottom-left corner are squares when $i + j \leq N-1$, while the volumes are triangles when $i+j = N-1$. 

We can proceed to derive the finite volume discretization for the square volumes by integrating Equation \eqref{eq:multileveltrimorphic}, yielding 
\begin{dmath*}
\int_{x_j}^{x_{j+1}} \int_{y_k}^{y_{k+1}} \dsdel{\rho(t,x,y)}{t} dy dx  \\ = - \int_{x_j}^{x_{j+1}} \int_{y_k}^{y_{k+1}} \left(\dsdel{}{x} \left[x \left( b_S - b_D  + \left( b_D - b_S \right) x + \left( b_D - b_F \right) y \right) \rho(t,x,y) \right]\right) dy dx  - \int_{x_j}^{x_{j+1}} \int_{y_k}^{y_{k+1}} \left(\dsdel{}{y} \left[ y \left( b_F - b_D + \left(b_D - b_S \right)x +  \left( b_D - b_F \right) y \right) \rho(t,x,y) \right]\right) dy dx
+ \int_{x_j}^{x_{j+1}} \int_{y_k}^{y_{k+1}} \lambda \rho(t,x,y) \left[G(x,y,1-x-y) - \int_{0}^1 \int_0^{1-x} G(u,v,1-u-v) \rho(t,u,v) dv du \right] dy dx
\end{dmath*}
For convenience, we will now abbreviate the characteristic ODEs corresponding to the  gene-level dynamics as $\ddt{x} = F_1(x,y)$ and $\ddt{y} = F_2(x,y)$. We can rewrite the integrals for our advection terms using the 2D divergence theorem, which lets us write the above equation as 
\begin{dmath*}
\int_{x_j}^{x_{j+1}} \int_{y_k}^{y_{k+1}} \dsdel{\rho(t,x,y)}{t} dy dx = \oint_{u_1} \left(\begin{smallmatrix} F_1(x,y) \\ F_2(x,y) \end{smallmatrix}\right) \cdot \left(\begin{smallmatrix} 1 \\ 0 \end{smallmatrix}\right) \rho(t,x,y) d(u_1) \\
- \oint_{u_2} \left(\begin{smallmatrix} F_1(x,y) \\ F_2(x,y) \end{smallmatrix}\right) \cdot \left(\begin{smallmatrix} 0 \\ 1 \end{smallmatrix}\right) \rho(t,x,y) d(u_2) \\
- \oint_{u_3} \left(\begin{smallmatrix} F_1(x,y) \\ F_2(x,y) \end{smallmatrix}\right) \cdot \left(\begin{smallmatrix} 1 \\ 0 \end{smallmatrix}\right) \rho(t,x,y) d(u_3) \\
+ \oint_{u_4} \left(\begin{smallmatrix} F_1(x,y) \\ F_2(x,y) \end{smallmatrix}\right) \cdot \left(\begin{smallmatrix} 0 \\ 1 \end{smallmatrix}\right) \rho(t,x,y) d(u_4)
+ \int_{x_j}^{x_{j+1}} \int_{y_k}^{y_{k+1}} \lambda \rho(t,x,y) \left[G(x,y,1-x-y) - \int_{0}^1 \int_0^{1-x} G(u,v,1-u-v) \rho(t,u,v) dv du \right] dy dx,
\end{dmath*}
where we chose the appropriate outward normal vectors for each edge $u_i$. Evaluating the dot products between the characteristic curves and the unit normals, we can now rewrite our equation as 
\begin{dmath*}
\int_{x_j}^{x_{j+1}} \int_{y_k}^{y_{k+1}} \dsdel{\rho(t,x,y)}{t} dy dx  = \int_{y_k}^{y_{k+1}} F_1(x_j,y) \rho(t,x,y) dy - \int_{x_k}^{x_{k+1}} F_2(x,y_{k+1}) \rho(t,x,y) dx 
- \int_{y_k}^{y_{k+1}} F_1(x_{k+1},y) \rho(t,x,y) dy 
+  \int_{x_j}^{x_{k+1}} F_2(x,y_{k}) \rho(t,x,y) dx
+ \int_{x_j}^{x_{j+1}} \int_{y_k}^{y_{k+1}} \lambda \rho(t,x,y) \left[G(x,y,1-x-y) - \int_{0}^1 \int_0^{1-x} G(u,v,1-u-v) \rho(t,u,v) dv du \right] dy dx
\end{dmath*}
This is the analogue of Equations \eqref{eq:PDEwithgridFS} or \eqref{eq:PDEwithgridFD} from the two-type finite volume method studied on the fast-slow and fast-dimer edges of the simplex. The main difference is that the fluxes on the edge of the volume must now be competed as integrals over the edges. When we take the piecewise constant approximation for $\rho(t,x,y)$, we see that an ambiguity arises for the definition of $\rho(t,x,y)$ when either $x$ or $y$ lies on a grid point, and therefore we need to employ a 2D version of the upwinding scheme to resolve these ambiguities. 

We can use the expressions from Equation \eqref{eq:withincelldimorphic} to calculate the flux across the edges. Along vertical edges starting at point $(x_j,y_k)$, the flux is
\begin{align*}
\ds\int_{y_k}^{y_{k+1}} F_1(x_j,y) dy &= \ds\int_{y_k}^{y_{k+1}} \left[ \left(b_S - b_D \right) x_j \left(1 - x_j \right) - \left(b_F - b_D\right) x_j y_k \right] dy \\ &= x_j \left(y_{k+1} - y_k \right) \left[ \left( b_S - b_D \right) \left(1 - x_j \right) - \left( b_F - b_D \right) \left( \frac{y_{k+1} + y_k}{2} \right) \right]
\end{align*}
This flux can potentially take either sign. For sufficiently small $x_j$ and $y_k$, the flux will be positive because slow replicators will reproduce more quickly than the dimers who occupy fraction $1 - x_j - y_k$ of the group. When either $x_j$ or $y_k$ is sufficiently large, the fast replicators will replicate quickly enough to cause a net negative flux in slow replicators across the edge. For the flux along horizontal edges starting at $(x_j,y_k)$, we use the fact that $x_j \leq x_{j+1}$, $x_{j+1} + y_k \leq 1$, and $b_F \geq b_S$ to see that
\begin{align*}
    \int_{x_j}^{x_{j+1}} F_2(x,y_k) dx &=  \int_{x_j}^{x_{j+1}} \left[ \left(b_F - b_D\right) \left(1 - y_k \right) - \left( b_S - b_D \right) x y_k \right] dx  \\
    &= y_k  \left(x_{j+1} - x_j \right) \left[ \left(b_F - b_D \right) \left( 1 - y_k \right) - \left(b_S - b_D \right) \left( \frac{x_{j+1} + x_j}{2} \right)  \right] \\
    & \geq \left( b_F - b_D \right) y_k  \left(x_{j+1} - x_j \right) \underbrace{\left[1 - x_{j+1} - y_k \right]}_{\geq 0} \\ &\geq 0
\end{align*}
For any horizontal edge, the flux will be downward because the within-group dynamics so strongly favor reproduction of the fast replicators.

We now consider the net change in probability for triangular volumes. Using a similar argument with the divergence theorem, we obtain the following expression for the probability on triangular volumes with bottom-left corners $(x_j,x_{N-j-1})$ 
\begin{dmath*} 
\int_{x_j}^{x_{j+1}} \int_{y_{N-j-1}}^{1-x} \dsdel{\rho(t,x,y)}{t} dy dx  = \int_{y_{N-j-1}}^{y_{N-j}} F_1(x_j,y) \rho(t,x,y) dy
+   \int_{x_k}^{x_{k+1}} F_2(x,y_{N-j-1}) \rho(t,x,y) dx
- \frac{1}{\sqrt 2}\oint_{v_2} \left(\begin{smallmatrix} F_1(x,y) \\ F_2(x,y) \end{smallmatrix} \right) \cdot \left(\begin{smallmatrix} 1 \\ 1\end{smallmatrix} \right) \rho(t,x,y) dx dy
+ \int_{x_j}^{x_{j+1}} \int_{y_{N-j-1}}^{1-x} \lambda \rho(t,x,y) \left[G(x,y,1-x-y) - \int_{0}^1 \int_0^{1-x} G(u,v,1-u-v) \rho(t,u,v) dv du \right] dy dx.
\end{dmath*}
Here we still need to figure out how to deal with the flux across the hypotenuse edge of triangular volume, which lies along the boundary line of the simplex given by $y = 1 - x$. Noting that the within-group dynamics of Equation \eqref{eq:withincelldimorphic} satisfy $\ddt{}\left(x+y+z\right) = 0$, we can see that 
\begin{align*} x
\begin{pmatrix} F_1(x,y) \\ F_2(x,y) \end{pmatrix} \cdot \begin{pmatrix} \frac{1}{\sqrt{2}} \\ \frac{1}{\sqrt{2}} \end{pmatrix}  &= \frac{1}{\sqrt{2}}  \left( F_1(x,y) + F_2(x,y) \right) = \ddt{} \left(x + y \right) = - \ddt{z} \\ 
&= z \left[ b_D - \left( b_S x +  b_F y +  b_D z \right)  \right] \\ &= 0 \textnormal{ on $v_2$ edge of volume (because $z = 0$ on this edge).}
\end{align*}
Therefore we see that the flux vanishes everywhere on the hypotenuse edge $v_2$ of the triangular volume, and therefore we only need to worry about flux along the horizontal edge on the bottom of the volume and vertical edge on the left of the volume.

Now that we have discussed how to compute flux along volume boundaries, we can introduce our piecewise-constant approximation to $\rho(t,x,y)$ in which we assume that $\rho(t,x,y)$ takes a constant value $\rho_{j,k}(t)$ for all $(x,y)$ in the volume $C_{j,k}$ whose bottom-left corner is $(x_j,y_k)$. Because all of our volumes have left and bottom edges, but only the rectangular ones have top and right edges within the simplex, we use the following coefficient to handle our flux terms across both cases
 \begin{equation} \label{eq:edgetypevariable}
   \alpha_{j,k} = \left\{
     \begin{array}{lr}
       1 & : j + k < N - 1\\
       0 & : j + k = N - 1
     \end{array}
   \right.
\end{equation} 

For the a volume with bottom-corner $(x_j,y_k)$ we know that the net flux pushes towards higher levels of fast replicator $y$, and therefore we will use upwinding to discretize our advection terms in the $y$-direction as
\begin{equation} \label{eq:ydiscreteadvection}
\left( \int_{x_j}^{x_{j+1}} F_2(x,y_k) dx \right) \rho_{j,k-1} - \alpha_{j,k} \left( \int_{x_j}^{x_{j+1}} F_2(x,y_{k+1}) dx \right) \rho_{j,k} 
\end{equation}
For the advection in the $x$-direction, we know that the flux across vertical edges can take either sign, so we can introduce the following notation to describe our upwinding rule for a vertical edge starting at $(x_j,y_k)$
 \begin{equation} \label{eq:xupwinding}
    U\left(x_j,y_k\right) = \left\{
     \begin{array}{lr}
       \rho_{j-1,k} & :  \ds\int_{y_k}^{y_{k+1}} F_1(x_j,y) dy > 0 \\
       \rho_{j,k} & : \ds\int_{y_k}^{y_{k+1}} F_1(x_j,y) dy < 0
     \end{array} 
   \right. .
\end{equation} 
Then we see that the contribution of flux in the $x$-direction is given by 
\begin{equation} \label{eq:xdiscreteadvection}
\left( \int_{y_k}^{y_{k+1}} F_1(x_j,y) dy \right) U\left(x_j,y_k\right) -  \alpha_{j,k} \left( \int_{y_k}^{y_{k+1}} F_1(x_{j+1},y) dy \right) U\left(x_{j+1},y_k\right).
\end{equation}
Next we can consider the discretized version of between-group competition. We note that the average payoff on square and triangular volumes are given by
\begin{subequations} \label{eq:Gavgcells}
     \begin{align}
        j + l < N - 1 \: : \: G_{j,k} &:= N^2 \int_{x_j}^{x_{j+1}} \int_{y_k}^{y_{k+1}} G(x,y,1-x-y) dx dy \\ 
        j + k = N - 1 \: : \: G_{j,k} &:= 2 N^2 \int_{x_j}^{x_{j+1}} \int_{y_k}^{1 - x} G(x,y,1-x-y) dx dy
     \end{align}
\end{subequations}
Then considering the term describing the gain in probability density through between-group competition, we see for our piecewise constant approximation $\rho_{j,k}(t)$ on square volumes that 
\begin{dmath} \label{eq:Gjksquare}
 \int_{x_j}^{x_{j+1}} \int_{y_k}^{y_{k+1}} G(x,y,1-x-y) \rho(t,x,y) dy dx =  \rho_{j,k}(t) \int_{x_j}^{x_{j+1}} \int_{y_k}^{y_{k+1}} G(x,y,1-x-y) dy dx =  \left(\frac{1}{N^2}\right) G_{j,k} \rho_{j,k}(t).
\end{dmath}
Similarly, we see that the equivalent term on triangular volumes is given by
\begin{dmath} \label{eq:Gjktriangle}
 \int_{x_j}^{x_{j+1}} \int_{y_k}^{1-x} G(x,y,1-x-y) \rho(t,x,y) dy dx =   \left(\frac{1}{ 2 N^2}\right) G_{j,k} \rho_{j,k}(t).
\end{dmath}
Because of the different coefficients for the rectangular and triangular grid volumes, it will be convenient to introduce the following shorthand to handle the cases together
 \begin{equation} \label{eq:edgetypegroupshorthand}
   \beta_{j,k} = \left\{
     \begin{array}{lr}
       1 & : j + k < N - 1\\
       2 & : j + k = N - 1
     \end{array}
   \right.
\end{equation} 
and we will denote integration over our cell with bottom-left corner at $C_{j,k} = (x_j,y_k)$ by $\iint_{C_{j,k}}$. By applying the same reasoning, we can see that the loss of probability density due to between-group competition is given by 
\begin{dmath} \label{eq:aveGcell}
 - \left(\int_0^1 \int_{0}^{1-x} G(u,v,1-u-v) \rho(t,u,v) du dv \right)  \iint_{C_{j,k}} \rho(t,x,y) dx dy = - \rho_{j,k}(t) \left(\frac{1}{\beta_{j,k} N^2} \right) \left( \ds\sum_{l,m \geq 0}^{l+m \leq N-1} G_{l,m} \rho_{l,m}(t) \right).
\end{dmath}
Our last ingredient for putting together a finite-volume approximation for Equation \eqref{eq:multileveltrimorphic} is to describe how the piecewise constant approximation $\rho_{j,k}(t)$ changes in time. From its definition as the average value on the volume $C_{j,k}$, we can see that 
\begin{equation} \label{eq:piecwisetime}
    \dsddt{} \iint_{C_{j,k}} \rho(t,x,y) dy = \frac{1}{\beta_{j,k} N^2} \dsddt{\rho_{j,k}(t)}
\end{equation}
Putting together the terms we calculation in Equations \eqref{eq:ydiscreteadvection}, \eqref{eq:xdiscreteadvection}, \eqref{eq:Gjksquare}, \eqref{eq:Gjktriangle}, \eqref{eq:aveGcell}, and \ref{eq:piecwisetime}, we can finally describe finite-volume approximation of our multilevel slow-fast-dimer dynamics by the system of ODEs given by 
\begin{dmath} \label{eq:trimorphicfinitevolume}
\frac{1}{\beta_{j,k}} \dsddt{\rho_{j,k}(t)} = \left( \int_{y_k}^{y_{k+1}} F_1(x_j,y) dy \right) U\left(x_j,y_k\right) -  \alpha_{j,k} \left( \int_{y_k}^{y_{k+1}} F_1(x_{j+1},y) dy \right) U\left(x_{j+1},y_k\right)  + \left( \int_{x_j}^{x_{j+1}} F_2(x,y_k) dx \right) \rho_{j,k-1} - \alpha_{j,k} \left( \int_{x_j}^{x_{j+1}} F_2(x,y_{k+1}) dx \right) \rho_{j,k} 
+ \frac{\lambda}{\beta_{j,k}} \rho_{j,k}(t) \left[ G_{j,k} -  \ds\sum_{l,m \geq 0}^{l+m \leq N-1} G_{l,m} \rho_{l,m}(t)  \right]
\end{dmath}

\subsubsection{Calculation of Average Reproduction Rate on volumes}

A final step needed to study numerical solutions for our finite volume approximation from Equation \eqref{eq:trimorphicfinitevolume} is to evaluate the values for the average protocell-level reproduction rate $G_{j,k}$ on the volume $C_{j,k}$.

To integrate $G(x,y)$ over our volumes, we first use Equation \eqref{eq:Gtrixy} to rewrite the protocell-level replication rate in the form
\begin{equation} \label{eq:Gxyrewritten}
    G(x,y) = \frac{1}{4} \left( 2 - \eta\right) + \frac{1}{2} \left( 1 - \eta\right) \left( x - y \right) - \frac{\eta}{4} \left( x^2 - 2xy + y^2 \right).
\end{equation}
This tells us that we can can obtain an expression for $\int_{C_{j,k}} G(x,y) dx dy$ by computing the volume-averages of each of the monomials in $x$ and $y$ with degree at most $2$. 

Starting with the fraction of slow replicators $x$, we first look to calculate the mass of slow replicators on rectangularvolumes $C_{j,k}$ with $j+k < N-1$. Integrating in $x$, we see that
\begin{subequations} \label{eq:xmeancells}
\begin{align*}
\iint_{C_{j,k}} x dx dy &= \int_{\frac{k}{N}}^{\frac{k+1}{N}}  \int_{\frac{j}{N}}^{\frac{j+1}{N}} x dx dy =  %\int_{\frac{j}{N}}^{\frac{j+1}{N}}  \left[ \frac{1}{2} x^2  \right] \bigg|_{x = \frac{i}{N}}^{\frac{j}{N}} dy 
 \frac{1}{2} \int_{\frac{j}{N}}^{\frac{j+1}{N}} \left[ \left( \frac{j+1}{N} \right)^2 - \left(\frac{j}{N}\right)^2 \right] dy.
\end{align*}

Integrating with respect to $y$ now allows us to see that
\begin{equation}
    \iint_{C_{j,k}} x dx dy = \frac{2j + 1}{2 N^3} \: \: \mathrm{for} \: \: j + k < N -1.
\end{equation}
Next, we can calculate the mass of slow replicators on the triangular volumes $C_{j,k}$ for $j + k = N - 1$. Noting that $k = N - j - 1$, we can use Equation \eqref{eq:Gavgcells} to see that, for the triangular volumes, 
\begin{equation*}
\int_{C_{j,k}} x dx dy =  \int_{\frac{j}{N}}^{\frac{j+1}{N}} \int_{\frac{N - j - 1}{N}}^{1-x} x dy dx = \int_{\frac{j}{N}}^{\frac{j+1}{N}}  x \left[ \frac{j+1}{N} - x\right] dx = \left( \frac{j+1}{2N} x^2 - \frac{1}{3} x^3 \right) \bigg|_{\frac{j}{N}}^{\frac{j+1}{N}}.
\end{equation*}
Simplifying the righthand side tells us that the mass of slow replicators on triangularvolumes is given by
\begin{equation}
    \int_{C_{j,k}} x dx dy  = \frac{1}{N^3} \left[ \frac{j}{2} + \frac{1}{6} \right] \: \: \mathrm{for} \: \: j + k = N - 1.
\end{equation}
\end{subequations}
Using an analogous calculation, we can see that the mass of fast replicators on rectangular volumes is given by
\begin{subequations} \label{eq:ymeancells}
\begin{align}
    \iint_{C_{j,k}} y dx dy &= \frac{1}{2N^3} \left[ 2 k + 1\right].
\end{align}
For triangular volumes, we can apply Equation \eqref{eq:Gavgcells} and Fubini's theorem to see that
\begin{equation*}
    \iint_{C_{j,k}} y dx dy = \int_{\frac{N-j-1}{N}}^{\frac{N-j}{N}} \int_{\frac{j}{N}}^{1-y} y dx dy = \int_{\frac{N-j-1}{N}}^{\frac{N-j}{N}} y \left[ \frac{N-j}{N} - y \right] dy = \left(\frac{N-j}{2N} y^2 - \frac{1}{3} y^3 \right) \bigg|_{\frac{N-i-1}{N}}^{\frac{N-j}{N}},
\end{equation*}
and simplifying the righthand side tells us that the mass of fast replicators on triangular volumes is given by 
\begin{equation}
     \iint_{C_{j,k}} y dx dy = \frac{1}{N^3} \left[\frac{1}{2} \left( N - j \right) - \frac{1}{3} \right]
\end{equation}
\end{subequations}

\begin{remark}
In addition to their use in calculating volume-averages pf the collective replication rate, the expressions from Equations \eqref{eq:xmeancells} and \eqref{eq:ymeancells} arise in quantifying the mena fraction of slow genes in the population as studied in Figure \ref{fig:peakandmeantrimorphic}. Noting that the fraction of dimers is given by $z = 1 - x - y$ on the simplex, the fraction of slow genes in a replicator can be written as $x + \frac{z}{2} = \frac{x-y + 1}{2}$. Therefore the average fraction of slow genes across the population of protocells under the finite-volume approximation is given by
\begin{dmath}
\int_0^1 \int_0^{1-x} \frac{1}{2} \left[x - y + 1\right] \rho(t,x)  dx dy \approx \frac{1}{2} \ds\sum_{l,m \geq 0}^{l + m \leq N - 1} \iint_{C_{l,m}} \left[x - y + 1\right] \rho_{l,m}(t) dx dy =  \frac{1}{2} \ds\sum_{l,m \geq 0}^{l + m \leq N - 1} \left[ \rho_{l,m}(t) \left(   \iint_{C_{l,m}} \left( x - y \right) dx dy -  \frac{1}{\beta_{l,m} N^2} \right)   \right].
\end{dmath}
Applying Equations \eqref{eq:xmeancells} and \eqref{eq:ymeancells} to the integral on the righthand side then yields the average fraction of the slow gene present over the population of protocells as shown in Figure \ref{fig:peakandmeantrimorphic} for long-time states of the multilevel dynamics.
\end{remark}
Next, we can integrate the quadratic monomials in $x$ and $y$ over the rectangular and triangular volumes. Using the same approach as in the case of the linear terms, we can see that the integrals of $x^2$, $xy$, and $y^2$ over the volumes $C_{j,k}$ are given by 
\begin{subequations} \label{eq:x2meancells}
\begin{align}
   \iint_{C_{j,k}} x^2 dx dy  &= \frac{1}{N^4} \left[ j^2 + j + \frac{1}{3} \right]  \: \: \mathrm{for} \: \: j + k < N - 1\\
    \iint_{C_{j,k}} x^2 dx dy  &= \frac{1}{N^4} \left[ \frac{j^2}{2} + \frac{j}{3} + \frac{1}{12} \right]   \: \: \mathrm{for} \: \: j + k = N - 1,
\end{align}
\end{subequations}
\begin{subequations} \label{eq:y2meancells}
\begin{align}
   \iint_{C_{j,k}} y^2 dx dy  &= \frac{1}{N^4} \left[ k^2 + k + \frac{1}{3} \right]   \: \: \mathrm{for} \: \: j + k < N - 1 \\
    \iint_{C_{j,k}} y^2 dx dy  &= \frac{1}{N^4} \left[ \frac{(N-j)^2}{2} - \frac{2(N-j)}{3} + \frac{1}{4} \right]   \: \: \mathrm{for} \: \: j + k = N - 1,
\end{align}
\end{subequations}
and
\begin{subequations} \label{eq:xymeancells}
\begin{align}
   \iint_{C_{j,k}} xy dx dy  &= \frac{1}{N^4} \left[ jk + \frac{j}{2} + \frac{k}{2} + \frac{1}{4} \right]   \: \: \mathrm{for} \: \: j + k < N - 1 \\
    \iint_{C_{j,k}} xy dx dy  &= \frac{1}{N^3} \left[ \frac{j}{2} + \frac{1}{6} \right] - \frac{1}{N^4} \left[ \frac{j^2}{2} + \frac{j}{2} + \frac{1}{8} \right]   \: \: \mathrm{for} \: \: j + k = N - 1.
\end{align}
\end{subequations}
In addition, we see from Equation \eqref{eq:Gxyrewritten} that the integral of the constant term in $G(x,y)$ is given by
\begin{subequations} \label{eq:constantmeancells}
\begin{align}
   \iint_{C_{j,k}} \frac{1}{4} \left( 2 - \eta\right) dx dy  &= \frac{1}{4 N^2} \left(2 - \eta\right)  \: \: \mathrm{for} \: \: j + k < N - 1 \\
    \iint_{C_{j,k}} \frac{1}{4} \left( 2 - \eta\right) dx dy  &=  \frac{1}{8 N^2} \left( 2 - \eta \right) \: \: \mathrm{for} \: \: j + k = N - 1.
\end{align}
\end{subequations}
Putting all of this together, we can now calculate the integral of the protocell-level replication rate over a given volume by applying the integrals obtained in Equations \eqref{eq:xmeancells}, \eqref{eq:ymeancells}, \eqref{eq:x2meancells}, \eqref{eq:y2meancells}, \eqref{eq:xymeancells}, and \eqref{eq:constantmeancells} and the formula for $G(x,y)$ from Equation \eqref{eq:Gxyrewritten}. We see that the integral of $G(x,y)$ over the rectangular volumes $C_{j,k}$ for $j + k < N -1$ is given by
\begin{subequations}
\label{eq:Gxyintegral}
\begin{dmath}
 \iint_{C_{j,k}} G(x,y) dx dy =  \int_{\frac{k}{N}}^{\frac{k+1}{N}}  \int_{\frac{j}{N}}^{\frac{j+1}{N}} G(x,y) dx dy 
 =\frac{1}{2 N^2} \left[ 1 - \frac{\eta}{2} \right] + \frac{1}{2 N^3} \left( 1 - \eta \right) \left( j - k \right) + \frac{\eta}{ N^4} \left[- \frac{j^2}{4} - \frac{k^2}{4} + \frac{jk}{2} - \frac{1}{24} \right].
\end{dmath}
For triangular volumes, corresponding to indices $j + k = N = 1$, we can similarly compute that
\begin{equation}
  \iint_{C_{j,k}} G(x,y) dx dy = \int_{\frac{j}{N}}^{\frac{j+1}{N}} \int_{\frac{N - j - 1}{N}}^{1-x} G(x,y) dy dx 
  = \frac{1}{N^3} \left[ \frac{j}{2} + \frac{1}{4} \right] - \frac{\eta}{N^4} \left[ \frac{j^2}{2} + \frac{j}{2} + \frac{7}{48} \right].
\end{equation}
\end{subequations}
Furthermore, we can calculate the average protocell-level replication rate $G_{j,k}$ over $C_{j,k}$ from the integrals of Equation \eqref{eq:Gxyintegral} and dividing through by the area of $C_{j,k}$ (equal to $\frac{1}{N^2}$ for the rectangular volumes and $\frac{1}{2 N^2}$ for the triangular volumes). For the rectangular volumes (with indices $j+k < N-1$), we therefore see that average protocell-level replication rates are given by
\begin{subequations}
\label{eq:Gjkmean}
\begin{dmath}
G_{j,k} =  N^2 \int_{\frac{k}{N}}^{\frac{k+1}{N}}  \int_{\frac{j}{N}}^{\frac{j+1}{N}} G(x,y) dx dy 
 =\frac{1}{2 N^2} \left[ 1 - \frac{\eta}{2} \right] + \frac{1}{2 N} \left( 1 - \eta \right) \left( j - k \right) + \frac{\eta}{ N^2} \left[- \frac{j^2}{4} - \frac{k^2}{4} + \frac{jk}{2} - \frac{1}{24} \right],
\end{dmath}
and the average protocell-level replication rates over triangular volumes (whose indices satisfy $j+k = N -1$) are given by
\begin{equation}
  G_{j,k} = 2 N^2 \int_{\frac{j}{N}}^{\frac{j+1}{N}} \int_{\frac{N - j - 1}{N}}^{1-x} G(x,y) dy dx 
  = \frac{1}{N} \left[ j + \frac{1}{2} \right] - \frac{\eta}{N^2} \left[ j^2 + j + \frac{7}{24} \right].
\end{equation}
\end{subequations}
Applying the formulas from Equation \eqref{eq:Gjkmean} for $G_{j,k}$ to the finite-volume approximation of Equation \eqref{eq:trimorphicfinitevolume}, we have now specified the full numerical scheme used in Section \ref{sec:trimorphicnumerics} to study the trimorphic multilevel dynamics. 

\end{document}